\documentclass[11pt]{article}

\usepackage{setspace, longtable, graphicx, hyphenat, hyperref, fancyhdr, ifthen, everypage, enumitem, amsmath, setspace}
\usepackage[left=1in, right=1in, bottom=1in, top=1in]{geometry}

\usepackage[dvipsnames]{xcolor}
\definecolor{danlanzi}{RGB}{167, 168, 189}
\definecolor{carmine}{rgb}{0.59, 0.0, 0.09}
\definecolor{maroon(html/css)}{rgb}{0.5, 0.0, 0.0}
\definecolor{coolblack}{rgb}{0.0, 0.18, 0.39}
\definecolor{prussianblue}{rgb}{0.0, 0.19, 0.33}
\definecolor{winered}{rgb}{0.5, 0, 0}
\definecolor{hanada}{HTML}{006284}
\hypersetup{colorlinks=true, linkcolor=winered, urlcolor=winered, citecolor=winered, filecolor=winered}

\usepackage{graphicx}
\usepackage{scrextend}

\usepackage[hang]{footmisc}

\usepackage{titling}
\thanksfootextra{}{\hspace*{0.45em}}
\usepackage{xurl}

\usepackage{amssymb}
\usepackage{amsthm}
\usepackage{natbib}
\setcitestyle{aysep={}} %Remove the comma between authors and years
\usepackage[title]{appendix}

\usepackage{newpxtext}
\usepackage{newpxmath}
\usepackage[cal=boondoxupr]{mathalpha}
\usepackage[flushleft]{threeparttable}
\usepackage[font=sc]{caption}
\usepackage{subcaption}
\usepackage{tabularx}
\usepackage{booktabs}
\usepackage{chngpage}
\usepackage{dsfont}
\usepackage{mathtools}
\usepackage{bm}
\usepackage{aligned-overset}

\usepackage{tikz}
\usetikzlibrary{shapes.geometric, arrows, trees, calc, decorations.pathreplacing, calligraphy, fit, backgrounds}

\usepackage[flushleft]{threeparttable}
\usepackage{array}
\usepackage{tabularx}
\usepackage{booktabs}
\usepackage{algorithm}
\usepackage{algpseudocode}

\theoremstyle{plain}
\newtheorem{theorem}{Theorem}

\newtheorem{assumption}{Assumption}
\newtheorem{lemma}{Lemma}
\theoremstyle{definition}
\newtheorem{definition}{Definition}[section]

\setlist[itemize]{nolistsep, topsep=0pt}
\setlist[enumerate]{nolistsep, topsep=0pt}
\graphicspath{{figures/}}
\newcommand{\ind}{\mathrel{\perp\!\!\!\perp}}
\newcommand{\ceil}[1]{\left\lceil #1 \right\rceil}

\usepackage{mathtools}
\DeclarePairedDelimiterX{\norm}[1]{\lVert}{\rVert}{#1}
\DeclareMathOperator{\Var}{Var}
\DeclareMathOperator{\Cov}{Cov}

\newcommand\keywords[1]{\textbf{Keywords}: #1}

\title{Conformal Inference for Experimental Attrition\\ in Social Science Research}
\author{Xiangyu Song\thanks{Ph.D. student, Department of Political Science, Washington University in St. Louis. \href{mailto:xiangyusong@wustl.edu}{xiangyusong@wustl.edu}.}}
\date{%\textit{Preliminary draft. Please do not cite or circulate without permission.}
\today}

\begin{document}

    \maketitle
    \begin{abstract}
        Attrition in survey and field experiments presents a challenge for social science research. Common approaches to deal with this problem -- such as complete case analysis, multiple imputation, and weighting methods -- rely on strong assumptions that may not hold in practice. This paper introduces a new method that combines recent advances in statistical inference with established tools for handling missing data. The approach produces prediction intervals for treatment effects that are both robust and precise. Evidence from simulation studies shows that the method achieves better coverage and produces narrower intervals than common alternatives. The reanalysis of two recently published experimental studies illustrates how this framework allows researchers to compare treatment effects across participants who remain in the study, those who drop out, and the full sample. Taken together, these results highlight how the proposed approach provides a stronger foundation for causal inference in the presence of attrition.
    \end{abstract}


    \clearpage

    \section{Introduction}

    Since the seminal contributions by Neyman and Fisher in the 1920s \citep{fisher1937the,splawa-neyman1990application}, the use of randomized experiments (randomized controlled trials) has become a cornerstone for estimating causal effects \citep{huber2012identification,imbens2024causal}. While rare in the 1970s and 1980s \citep[e.g.,][]{lalonde1986evaluating}, experimental studies have expanded rapidly in the social sciences over the past two decades \citep{imbens2024causal}. Political scientists now routinely employ survey and field experiments, often alongside survey instruments, to evaluate theories and estimate causal effects \citep[e.g.,][]{imai2011estimation,kalla2018minimal}. Yet, like all empirical strategies, experiments face threats to validity, particularly attrition \citep{hausman1979attrition}, which complicates inference \citep{gerber2012field}. In survey experiments, non-response due to attrition prevents observation of outcomes for all participants, while in field experiments attrition reduces sample size and statistical power. When attrition is non-random, it introduces selection bias that undermines theoretical claims. Ignoring attrition and restricting analysis to observed samples generally yields biased and inconsistent estimators \citep{coppock2017combining}.

    %Researchers have proposed several remedies to deal with the prevalence of missing data problem induced by experimental attrition. The first line of approach addresses the missingness in the outcome by introducing additional assumptions \citep{coppock2017combining}, e.g., sample selection model \citep{heckman1979sample} and instrumental variable \citep{huber2012identification}. However, such assumptions are hard to verify on the theoretical ground explicitly. A second approach uses (multiple) imputation or complete case analysis to handle missing data \citep[e.g.,][]{fukumoto2022nonignorable,king2001analyzing,rubin1976inference}. Nevertheless, these methods are also limited to strong distributional assumptions and missingness pattern. A third strategy, inverse probability weighting (IPW), reweights observations using covariates to adjust for missingness \citep{horvitz1952generalization,huber2012identification}. However, IPW alone cannot guarantee valid inference for the attrition group, especially when attrition is correlated with treatment or covariates. Finally, partial identification such as nonparametric bounds \citep{horowitz1998censoring,horowitz2000nonparametric,lee2009training} and double sampling with worst-case bounds \citep{coppock2017combining} offer alternatives that make fewer assumptions but have limitations across heterogeneous subpopulations.

    Existing approaches to address attrition each have their own limitations. Some rely on unverifiable assumptions about model specification or data structure, such as sample selection models \citep{heckman1979sample}, instrumental variables \citep{huber2012identification}, and imputation or complete case methods \citep[e.g.,][]{fukumoto2022nonignorable,king2001analyzing,rubin1976inference}. Some reweighting methods, like inverse propensity weighting (IPW) \citep{horvitz1952generalization,huber2012identification}, adjust for missingness using observed data but do not leverage information from attrited observations and can yield biased estimates when missingness depends on treatment and covariates. Other bounding approaches, like partial identification \citep{horowitz1998censoring,horowitz2000nonparametric,lee2009training,coppock2017combining}, relax these parametric assumptions but are not designed to address the covariate shift problem, where the covariate distribution differs between observed and attrited groups.

    In this paper, I propose to leverage conformal inference, a nonparametric and distribution-free method for uncertainty quantification on individual predictions, to construct valid prediction intervals for treatment effects in the presence of attrition. Conformal inference provides a way to quantify uncertainty without relying on strong parametric assumptions on the data generating process (DGP). Intuitively, the method works by comparing the model's predictions to the actual observed outcomes. By assessing how well the model's predictions align with the observed values, it determines a range (or interval) around each prediction that will contain the true value with a user-specified level of confidence (e.g., 90\% or 95\%). The advantage of conformal inference is that, regardless of the complexity of the model or the distribution of the data, it guarantees coverage: if we ask for 95\% coverage, then, on average, about 95\% of the prediction intervals will indeed contain the true outcome. This property makes conformal inference particularly appealing for causal inference, where we often want reliable, distribution-free uncertainty quantification for treatment effects at the individual or subgroup level.

    Recent advances in the conformal inference literature have demonstrated its potential in addressing causal inference problems under attrition. \citet{lei2021conformala} provide a theoretical framework for using conformal inference to construct prediction intervals for counterfactuals and individual treatment effect (ITE). However, their method is not designed specifically to handle missing data problems induced by attrition, as it relies on relatively strong missing data assumptions, and the resulting prediction intervals can be overly wide (non-informative) in practice. Building on this work, \citet{yang2024doubly} and \citet{gao2025role} introduce approaches that lead to more accurate and efficient prediction intervals. Rather than using the potential outcome framework, \citet{yang2024doubly} reformulate the covariate shift problem as a missing data problem, and provide a doubly robust and computationally efficient method for constructing prediction intervals for estimands of interest. However, they do not directly tackle the issue of missing data due to attrition. This is why one of the core motivations of this paper is to formally address missing data problems induced by experimental attrition.
    
    %My discussion aligns with recent development of applying conformal inference to construct prediction intervals for causal estimands of interest \citep{gao2025role,jin2023sensitivity,lei2021conformala}. Nevertheless, it also differs from existing work by adopting the potential outcome framework and placing experimental attrition and causal identification for attrition group at the center of the analysis. I explore the use of conformal inference to address the missing data problem induced by attrition in experiments. Specifically, this paper seeks to answer three questions: How can we construct valid prediction intervals for causal estimands of interest and extrapolate these intervals to the attrition group? How can we construct prediction intervals with guaranteed empirical coverage but narrower interval size? How can we extend the finite-sample nonasymptotic property of prediction intervals to the asymptotic property? By integrating conformal inference with semiparametric efficient estimator, I am able to construct valid and asymptotic prediction intervals for causal estimands of interest in the presence of experimental attrition. 

    This paper makes three contributions to the literature. First, it introduces conformal inference as a useful framework for political science research, demonstrating how it can provide reliable measures of uncertainty for causal estimands such as treatment effects and counterfactuals. By focusing on distribution-free and model-agnostic inference, this approach helps researchers quantify uncertainty even when conventional statistical assumptions may not hold. Second, building on \citet{lei2021conformala} and \citet{gao2025role}, this paper combines conformal inference with flexible estimators to construct prediction intervals for treatment effects among participants who drop out of experiments. It offers a general and practical way to handle missing data problems caused by experimental attrition under the potential outcome framework. Third, this paper presents one of the most comprehensive evaluations of how conformal inference performs across different predictive models, underscoring both its strengths and its potential limitations for empirical applications in social science research.

    %Compared with current methods, the proposed approach produces more precise and robust prediction intervals for the causal estimands of interest. Through extensive Monte Carlo (MC) simulations, this paper shows that the proposed method can construct prediction intervals for the treatment effects in the attrition group with both valid coverage and narrower interval length. It also outperforms existing parametric and nonparametric methods in terms of empirical coverage and average interval length. By reanalyzing two recently published experimental studies, this paper demonstrates how the proposed method can be applied in practice to compare treatment effects across participants who remain in the study, those who drop out, and the full sample. In addition, it enables researchers to aggregate individual-level treatment effects to estimate average treatment effects for different subpopulations.

    I demonstrate the advantages of my proposed approach to deal with experimental attrition in two ways. First, through extensive Monte Carlo (MC) simulations, this paper shows that the proposed method can construct prediction intervals for the treatment effects in the attrition group with both valid coverage and narrower interval length. It also outperforms existing parametric and nonparametric methods in terms of empirical coverage and average interval length. Second, by reanalyzing two recently published experimental studies, this paper demonstrates how the proposed method can be applied in practice to compare treatment effects across participants who remain in the study, those who drop out, and the full sample. In addition, it enables researchers to aggregate individual-level treatment effects to estimate average treatment effects for different subpopulations. Overall, compared with current methods, the proposed approach produces more precise and robust prediction intervals for the causal estimands of interest. 

    The rest of the paper is organized as follows. Section 2 goes over the current approaches to deal with the missing data problem. Section 3 introduces the problem setup. Section 4 provides a brief review of conformal inference and the covariate shift problem. Section 5 discusses the conformal inference method for the missing data problem induced by experimental attrition. Section 6 presents the proposed method for constructing prediction intervals for counterfactuals and ITE in the presence of experimental attrition. Section 7 evaluates the performance of the proposed method through extensive MC simulations. Section 8 illustrates the application of the proposed method by the reanalysis of two recently published experimental studies. Section 9 concludes.

    \section{Current Approaches to Deal with Attrition in Experimental Settings}

    %Researchers have proposed several remedies to deal with the prevalence of missing data problem induced by experimental attrition. The first line of approach addresses the missingness in the outcome by introducing additional assumptions \citep{coppock2017combining}. For example, sample selection model requires strong parametric assumption or an exclusion variable for robust inference \citep{heckman1979sample}. Similarly, the instrumental variable approach requires an instrumental variable that is related to the missingness but has no direct effect on the outcome \citep{huber2012identification}. However, such assumptions are hard to verify on the theoretical ground explicitly. 

    Researchers have proposed several remedies to deal with the prevalence of missing data as a result of experimental attrition. The first set of solutions addresses missingness problems due to attrition in the outcome by introducing parametric assumptions \citep{coppock2017combining}. For example, sample selection models \citep{heckman1979sample}, in which outcomes are observed only for a non-random subset of units, rely on strong assumptions, including correct specification of the selection process, distribution functions, and exclusion restrictions. Similarly, proponents of instrumental variable approaches as a way to deal with attrition assume that there exists a variable that is related to the missingness but has no direct effect on the outcome \citep{huber2012identification}. However, such assumptions are hard to verify on theoretical grounds, and valid instruments are hard to find in practice.

    A second family of approaches employs (multiple) imputation methods or complete case analysis (e.g., listwise or pairwise deletion) to correct for potential bias due to missing data \citep[e.g.,][]{fukumoto2022nonignorable,king2001analyzing,rubin1976inference}. Although not designed to handle causal inference problems directly, we can manipulate these imputation methods to impute the missing potential outcomes and attrited observations. Nevertheless, these methods are also limited by strong assumptions. For example, multiple imputation using Amelia II assumes that the complete data (that is, both observed and unobserved) are multivariate normal and that the missingness is missing at random (MAR). Complete case analysis can cause biased estimates when the missingness does not satisfy missing completely at random (MCAR) or missing at random \citep{shin2024differenceindifferences}. Imputation methods are also problematic under missing not at random (MNAR), which violates the ignorability of missingness, or when the covariates are themselves missing \citep{coppock2017combining,shin2024differenceindifferences}. 
    
    The third line of approach leverages inverse propensity weighting (IPW) to reweight the observations based on observed covariates \citep{horvitz1952generalization,huber2012identification}. However, IPW only provides a treatment effect estimate for the observed group, but as social scientists, we want the ATE for all observations, including observed and attrited groups. Moreover, this IPW approach alone does not resolve the issue of making valid inference on the attrition group -- particularly when the attrition is induced by treatment or other covariates. In that case, although IPW can give an unbiased treatment effect for the observed group, it could still be biased for all observations. Finally, the fourth line of approach tackles the missing data problem through partial identification, for example, nonparametric bounds \citep{horowitz1998censoring,horowitz2000nonparametric,lee2009training} and double sampling with worst-case bounds \citep{coppock2017combining}. Although these nonparametric approaches make fewer assumptions, they overlook the potential covariate shift problems, where the covariate distribution differs across attrited and observed groups. %This can limit the utility for inference across heterogeneous subpopulations. For instance, consider a field experiment where attrition is systematically higher among younger respondents than older ones. If the observed group is disproportionately older, while the attrition group is disproportionately younger, the covariate distribution of age differs across groups. Nonparametric bounds would treat these groups as if they were comparable, ignoring the covariate imbalance.

    %Instead of relying on current approaches with restrictions, this paper takes an alternative nonparametric approach to explore the inference with experiment attrition with minimal assumption. To be specific, this paper seeks to leverage conformal inference, a non-parametric and distribution-free method, to construct prediction intervals on treatment effects in the presence of experimental attrition. By using conformal inference, we can obtain valid prediction intervals for treatment effects under covariate shift without relying on strong parametric assumptions about the data generating process. We will discuss the problem setup and breif introduction to conformal inference in the next two sections.

    In this paper, I propose to leverage conformal inference, a nonparametric and distribution-free method for uncertainty quantification on individual predictions, to construct valid prediction intervals for treatment effects in the presence of attrition. Compared to current approaches, conformal inference does not rely on parametric assumptions on the data distribution, thus offering a more robust and flexible framework for handling missing data problems induced by experimental attrition. The next two sections provide the problem setup, a brief introduction to conformal inference, and why it can be used to deal with experimental attrition. 

    \section{Problem Setup and Notation}

    Consider the following setup for a randomized experiment with attrition induced by survey non-response or experiment dropout. We adopt the potential outcome framework \citep{splawa-neyman1990application,rubin1974estimating}. Let $D_i \in \{0, 1\}$ be a binary treatment indicator, where $D_i = 1$ if unit $i$ is treated and equals 0 otherwise (control), let $(Y_i (1), Y_i(0))$ be the potential outcomes for unit $i$, and $X_i$ the vector including $k$ pretreatment covariates. In a similar way, we define $D, Y(d),$ and $X$ without subscripts to be the population counterparts.
    %Denote $D \in \left\{ 0, 1 \right\}$ the binary treatment indicator, $\left( Y(1), Y(0) \right)$ the pair of potential outcomes, $X$ the vector of other pretreatment covariates. 
    We first assume that the data are independently and identically distributed (i.i.d.):
    \[\left( Y_i(0), Y_i(1), D_i, X_i \right) \overset{\text{i.i.d.}}{\sim} \left( Y(1), Y(0), D, X \right).\]
    In addition, we make the stable unit treatment value assumption (SUTVA). Then, the observed triples $Y_i^{\text{obs}}, D_i, X_i$ are given by 
    \[Y_i^{\text{obs}} = D_i Y_i(1) + (1 - D_i) Y_i(0).\]
    We define individual treatment effect (ITE) $\tau_i$ as follows
    \[\tau_i \coloneq Y_i(1) - Y_i(0).\]   
    We denote $R \in \left\{ 0, 1 \right\}$ the binary indicator of response, with $R = 1$ representing observations without attrition and $R = 0$ representing those with attrition. Following the literature on covariate shift, external validity, and missing data \citep{athey2020combining,egami2023elements,gao2025role,lei2021conformala}, we call the sample with observed outcome ($R = 1$) the source data and those with missing outcomes the target data ($R = 0$). Assume we observe the triples $\left( Y, D, X \right)$ from the source data, but only observe the pair $\left( D, X \right)$ from the target group.
    Table \ref{tab1} summarizes the problem setup. From the table, we can see that there are two missing data problems. First, due to the fundamental problem of causal inference \citep{holland1986statistics}, only one of the potential outcomes is observed for each unit (first four rows of Table \ref{tab1}). Second, due to experiment attrition, both potential outcomes are missing for the attrited units (last two rows of Table \ref{tab1}).

    \begin{table}[h]
        \centering
        \caption{Randomized Experiment with Attrition}
        \label{tab1}
        \begin{threeparttable}
            \begin{tabularx}{0.6\textwidth}{*5{>{\centering\arraybackslash}X}}
                $Y(1)$ & $Y(0)$ & $D$ & $X$ & $R$\\
                \toprule
                $\checkmark$ & ? & 1 & $\checkmark$ & 1 \\
                ? & $\checkmark$ & 0 & $\checkmark$ & 1 \\
                ? & $\checkmark$ & 0 & $\checkmark$ & 1 \\
                $\checkmark$ & ? & 1 & $\checkmark$ & 1 \\
                $\vdots$ & $\vdots$ & $\vdots$ & $\vdots$ & $\vdots$\\
                ? & ? & 0 & $\checkmark$ & 0\\
                ? & ? & 1 & $\checkmark$ & 0\\
                \hline\hline
            \end{tabularx}
            \begin{tablenotes}
                \footnotesize
                \setlength\labelsep{0pt}
                \item\textit{Note}: This table summarizes the setup of randomized experiment with attrition. $Y_i(1)$ and $Y_i(0)$ denote the potential outcomes. $D$ denotes the treatment, $X$ denotes the covariates, and $R$ denotes the indicator of attrition with $R = 1$ as the complete sample and $R = 0$ as the attrition sample. $\checkmark$ stands for observed quantity while $?$ stands for unobserved quantity.
            \end{tablenotes}
        \end{threeparttable}
        
    \end{table}

    We make the following additional assumptions for the problem setup throughout the paper:
    \begin{assumption}[Unconfoundedness]\label{asm1}
        \[\left( Y(1), Y(0) \right) \ind D \mid X.\]
    \end{assumption}
    The unconfoundedness assumption of treatment assignment requires that conditional on the covariates, the treatment assignment is independent of the potential outcomes. Similarly, we assume a version of unconfoundedness for the response indicator $R$:
    \begin{assumption}[Unconfoundedness of Attrition/MAR]\label{asm2}
        \[\left( Y(1), Y(0) \right) \ind R \mid X, D.\]
    \end{assumption}
    In the framework of missing data problem, Assumption \ref{asm2} can be referred to as missing at random (MAR). This assumption allows us to use the observed data from the source group to make inference about the target group, without assuming the covariate distribution remains the same across groups. In other words, the missing data induced by the attrition is independent of the potential outcomes, conditional on the observed covariates.\footnote{When we tackle the missingness pattern as missing completely at random (MCAR), we do not necessarily need Assumption \ref{asm2}. As discussed by \citet[221]{gerber2012field}, under MCAR, we only need to assume that 
    \[\left( Y(1), Y(0) \right) \ind R \mid X.\]
    Often this kind of missingness is relatively innocuous as the difference-in-means estimator remains an unbiased estimator of the average treatment effect (ATE).} 
    \begin{assumption}[Overlap]\label{asm3}
        We assume the overlap condition for both the propensity score of treatment and experiment attrition: for $c > 0$,
        \begin{align*}
            c < e_D(X) &\equiv \mathbb{P}\left( D = 1 \mid X \right) < 1 - c\\
            c < e_R(X, D) &\equiv \mathbb{P}\left( R = 1 \mid X, D \right) < 1 - c.
        \end{align*}
    \end{assumption}
    The overlap assumption is common in the causal inference and missing data literature \citep{imbens2015causal,rosenbaum1983central}. It requires that the probability of receiving treatment and the probability of response are both bounded away from 0 and 1, conditional on the covariates (and treatment). This assumption is crucial for the identification of ITE in the presence of experiment attrition. It ensures that there is a non-negligible probability of observing both treatment and control groups across all levels of covariates, as well as a non-negligible probability of response for both groups, which enables the comparisons across subpopulations.

    To make the following analysis more concrete, we define the following notation throughout the paper. Let $\mathcal{C}_d(X)$ be the prediction interval for the potential outcome $Y(d)$, where $d \in \left\{ 0, 1 \right\}$. Similarly, let $\mathcal{C}_{\text{ITE}}(X)$ be the prediction interval for the ITE. Denote $\check{\mathcal{C}}_{\text{ITE}}(X)$ as the extrapolated prediction interval for the ITE in the attrition group. %Table \ref{tab:notation} summarizes the notation and definitions used in this paper.

    %\begin{table}[ht]
    %    \centering
    %    \caption{Notation and Definition}
    %    \label{tab:notation}
    %    \begin{tabular}{cc}
    %        \hline\hline
    %        Notation & Definition\\
    %        \hline
%
    %        \hline
    %        %$Y(d)$ & Potential outcome under treatment $d \in \left\{ 0, 1 \right\}$\\
    %        $\mathcal{C}_d(X)$ & Prediction interval for potential outcome $Y(d)$\\
    %        $\mathcal{C}_{\text{ITE}}(X)$ & Prediction interval for ITE in observed group\\
    %        $\check{\mathcal{C}}_{\text{ITE}}(X)$ & Prediction interval for ITE in attrition group\\
    %        $V_d, V_\mathcal{C}$ & Nonconformity score for $Y(d)$ and $\mathcal{C}_{\text{ITE}}$\\
    %        $\eta_{\alpha, d}$ & $(1 - \alpha)$ quantile of the nonconformity score $V_d$\\
    %        $\eta_{\gamma, \mathcal{C}}$ & $(1 - \gamma)$ quantile of the nonconformity score $V_\mathcal{C}%$\\
    %        $e_D(X), \pi_D(X)$ & Propensity score of treatment and the odds ratio\\
    %        $e_R(X, D), \pi_R(X, D)$ & Propensity score of attrition and the odds ratio\\
    %        $m_d, m_\mathcal{C}$ & Conditional CDF of $V_d$ and $V_\mathcal{C}$ at $\eta_{\alpha, d}$ and %$\eta_{\gamma, \mathcal{C}}$\\
    %        $\psi_d, \psi_\mathcal{C}$ & Efficient influence function\\
    %        \hline\hline
    %    \end{tabular}
    %\end{table}
    
    \section{Conformal Inference with Covariate Shift}

    This section presents an introduction to conformal inference and covariate shift problems.

    \subsection{Conformal Inference}

    In this study, the primary estimands of interest are counterfactual potential outcomes and the ITE.\footnote{While political scientists often focus on aggregate effects such as the ATE or the average treatment effect on the treated (ATT), these quantities can be obtained by taking expectations over the ITE.} Rather than generating point predictions, I propose using conformal inference to construct valid prediction intervals that cover these random variables.

    Conformal inference provides a practical way to quantify uncertainty in modern prediction settings, especially when we rely on flexible machine-learning models \citep{angelopoulos2024theoretical}. In political science research, we can train a model to predict an outcome based on observed characteristics. For example, \citet{gohdes2020repression} tries to predict regime violence based on death records, and \citet{mueller2024crowd} tries to predict crowd cohesion score from survey responses collected at every protest. While a model produces a single predicted value, it does not automatically tell us how reliable that prediction is. Some predictions are based on abundant, stable patterns in the data, while others may be much more uncertain. 

    Conformal inference addresses this challenge by taking any prediction model as an input and adding a calibrated bound around each prediction, namely a prediction interval, that is designed to include the true outcome with a pre-specified probability (for instance, 95\%).\footnote{This prediction interval is different from a confidence interval. The confidence interval is a frequentist approach, in which we assume a correctly specified model for the DGP, or at least the conditional distribution of the estimand given covariates. Also, it comes from repeated sampling and the Central Limit Theorem, which guarantees the coverage in limit. However, prediction intervals constructed by conformal inference only assume exchangeability of the DGP and guarantees the finite-sample coverage without any reliance on asymptotic approximations.} The intuition here is to assess how well the model's predictions align with the actual observed values in a held-out calibration set. We first use training data to fit a prediction model, which can be any off-the-shelf machine learning algorithm, such as random forests, gradient boosting machines, or neural networks. Next, we apply this trained model to a separate calibration dataset to generate predictions. We then compare these predictions to the actual observed outcomes in the calibration dataset. The discrepancy between the predicted and observed value is called the nonconformity score.\footnote{There are many different measures of this discrepancy. One can imagine the most intuitive ones as the residual or softmax score.} By permuting the nonconformity scores and selecting a pre-specified quantile of the permuted distribution, we can determine a margin of error around each prediction. This margin of error is then used to construct the prediction interval. Then, given any new data point drawn from the same population, we can use the pre-trained model to generate a prediction and use the calibrated margin of error to construct a prediction interval.
    
    Introduced and developed by Vladimir Vovk and collaborators \citep{vovk1999machinelearning,vovk2009online,vovk2019nonparametric}, conformal inference has gained significant attention from the statistics community for regression problems \citep[e.g.,][]{lei2013distributionfree,lei2014distributionfree,lei2018distributionfree} and classification problems \citep[e.g.,][]{romano2020classification,sadinle2019least}. As noted by \citet{angelopoulos2022gentle}, these prediction intervals are valid in a distribution-free sense: they possess explicit, non-asymptotic guarantees even without distributional or model assumptions.  In other words, this coverage guarantee holds without requiring strong assumptions about how the data are generated or whether the underlying model is perfectly specified. Instead, conformal inference only relies on the idea that the observed data points come from a similar underlying process. Under this condition, the method ensures that, on average, the intervals will contain the true value with the desired probability. In addition to the interval coverage, we are also interested in the interval size, which reflects the model's uncertainty in a transparent way. Narrower intervals imply more confidence in the prediction, while wider intervals signal greater uncertainty. This feature allows researchers to understand not only what the model predicts, but also how stable or reliable that prediction appears to be.

    In the context of the missing data problem induced by experimental attrition, we are interested in two kinds of predictions, as indicated by two types of missingness mentioned in Table \ref{tab1}. First, we want to predict the potential outcome under treated (control) that is not observed for each non-attrited control (treatment) unit. Second, we want to predict both the potential outcomes (or the combined ITE) for each attrited unit. In both cases, instead of simply providing a point estimate with unknown uncertainty, we will use conformal inference to construct prediction intervals that quantify the uncertainty around these predictions. However, when making these two kinds of predictions, we face a challenge called covariate shift at both steps.

    \subsection{Covariate Shift}

    A key appeal of conformal inference is its nonparametric nature. It generates prediction intervals with guaranteed marginal coverage for data drawn i.i.d. from any distribution, without requiring parametric assumptions. This property makes conformal inference a fully distribution-free approach. Importantly, the i.i.d. assumption, while sufficient, is not strictly necessary. Conformal inference also applies under the weaker assumption of exchangeability, which allows observations to exhibit some dependence as long as their joint distribution is unchanged by permuting their order. In effect, the data can be dependent but must still be identically distributed in the sense that no observation is ``special'' or carries extra information simply because of when it appears.\footnote{See \ref{exch} for a formal definition. Exchangeability is \textit{the} crucial condition underpinning the validity of conformal inference: it ensures that the nonconformity scores can be meaningfully ranked to construct prediction intervals.} Thus, conformal inference retains theoretical robustness even beyond the standard i.i.d. framework.  

    In our setting, however, these assumptions must be relaxed further. Our goal is two-fold: first, to construct valid prediction intervals for the counterfactuals $\mathcal{C}_d$ and ITE $\mathcal{C}_{\text{ITE}}$ for the observed group; and second, to extrapolate these prediction intervals to the attrition group, yielding interval estimates of the ITE $\check{\mathcal{C}}_{\text{ITE}}$. In this context, the exchangeability assumption no longer holds across groups. To account for this violation, we must address the covariate shift problem, defined as a distribution shift for the covariates \citep{shimodaira2000improving,tibshirani2019conformal}, at both stages of the inference procedure. 

    Ideally, in a randomized experiment without attrition, the joint distribution of covariates and outcomes can be written as $P_X \times P_{Y \mid X}$ for all observations. However, due to the attrition, the joint distribution changes to $Q_X \times P_{Y \mid X}$, with potentially different covariate distribution $Q_X$ across groups. Under the shift from $P_X$ to $Q_X$, our goal is to construct the prediction interval $\mathcal{C}_d$ for the counterfactual $Y(d)$ with desired coverage:
    \begin{align}
        \mathbb{P}_{\left( X, Y(d) \right) \sim Q_X \times P_{Y(d) \mid X}} \left( Y(d) \in \mathcal{C}_d(X) \right) \geq 1 - \alpha,\quad d \in \left\{ 0, 1 \right\},\label{eq:covshift1}
    \end{align}
    and then extrapolate the prediction intervals to the attrition group:
    \begin{align}
        \mathbb{P}_{(X, \mathcal{C}_{\text{ITE}}) \sim Q_X \times P_{\mathcal{C}_{\text{ITE}} \mid X}} \left( \mathcal{C}_{\text{ITE}} \subset \check{\mathcal{C}}_{\text{ITE}}(X) \right) \geq 1 - \gamma. \label{eq:covshift2}
    \end{align}

    Under the ideal randomized experiment with i.i.d. data, these complications would not arise. Random assignment ensures that treatment is independent of covariates, yielding constant propensity scores across all covariate strata. Moreover, with no attrition, the covariate distribution remains identical across groups. In such a setting, covariate shift is not a concern.
    
    With attrition, the situation changes significantly. The covariate distribution may differ across groups in two distinct ways. First, covariates differ between the observed and attrition groups: $\mathbb{P}\left( X \mid D, R = 0 \right) \neq \mathbb{P}\left( X \mid D, R = 1 \right)$. Under the MAR assumption, attrition depends on observed covariates, generating a shift from $P_{X \mid D, R=1}$ to $P_{X \mid D, R=0}$. When constructing prediction intervals for the attrition group, only observations with $\mathcal{C}_{\text{ITE}}$ from the observed group are informative for constructing $\check{\mathcal{C}}_{\text{ITE}}(X)$ for the attrition group, making it necessary to correct for this distributional difference.
    
    Second, attrition also induces covariate shift between treatment and control units within the observed group: $\mathbb{P}\left( X \mid D = 0, R = 1 \right) \neq \mathbb{P}\left( X \mid D = 1, R = 1 \right)$. Even with random assignment of treatment, missingness correlated with covariates can distort the covariate balance originally ensured by randomization. Thus, when constructing prediction intervals for counterfactuals within the observed group, we still need to account for the covariate shift problem. %For example, constructing the prediction intervals $\mathcal{C}_1(X)$ for $Y(1)$ relies the treatment units only. Therefore, if the prediction intervals of the potential outcomes $Y(1)$ are intended for all observed units, we need to address the covariate shift from $P_{X \mid D = 1}$ to $P_X$. Likewise, if the the treated units are used to construct $\mathcal{C}_1(X)$ for control units only, the relevant covariate shift we need to control is from $P_{X \mid D = 1}$ to $P_{X\mid D = 0}$.

    \subsection{Weighted Conformal Inference}

    To tackle the covariate shift problem for interval estimates, one solution is to rely on the weighted conformal inference developed by \citet{tibshirani2019conformal}.\footnote{\citet{lei2021conformala} provide a version of split conformalized quantile regression, which combines the weighted conformal inference with conformalized quantile regression (CQR) for regression problems.} 
    Conformalized quantile regression (CQR), introduced by \citet{romano2019conformalized}, is designed to produce the interval estimate in the form of 
    \begin{align}
        \mathcal{C}(x) = \left[ \hat{q}_{\alpha_{\text{lo}}}(x) - \eta,  \hat{q}_{\alpha_{\text{hi}}}(x) + \eta \right] , \label{eq:cqr}
    \end{align}
    where $\hat{q}_{\alpha_{\text{lo}}}(x)$ and $\hat{q}_{\alpha_{\text{hi}}}(x)$ are the estimates of the pre-specified error margin $\alpha_{\text{lo}}$-th and $\alpha_{\text{hi}}$-th conditional quantiles of $Y \mid X = x$. $\eta$ is the constant computed after ranking the nonconformity scores. To compute the conditional quantile, we can simply change the loss function of a specific machine learning algorithm to a quantile loss, or use quantile regression \citep[e.g.,][]{koenker1978regression,koenker2001quantile}. As proved by \citet{romano2019conformalized}, CQR has finite sample coverage guarantee.

    %Under covariate shift between the source distribution and the target distribution, we need to adjust Equation \eqref{eq:ci} and construct the interval estimates as
    %\begin{align}
    %    \mathbb{P}_{(X, Y) \sim Q_X \times P_{Y \mid X}} \left( Y \in \mathcal{C}(X) \right) \geq 1 - \alpha. \label{eq:cqrcovshift}
    %\end{align}
    The key idea for addressing the covariate shift problem is reweighting. By reweighting each nonconformity score by a probability that is proportional to the likelihood ratio
    \[w(x) = \frac{dQ_X(x)}{dP_X(x)},\]
    weighted conformal inference can achieve the desired coverage \citep{tibshirani2019conformal}.  
    %Denote $V_i$ the nonconformity score for observation $i$. We are changing our interest from the empirical distribution of unweighted nonconformity score
    %\[\frac{1}{n + 1}\sum_{i = 1}^{n} \delta_{V_i(x)} + \frac{1}{n + 1}\delta_\infty\]
    %to
    %\[\sum_{i = 1}^{n} p_i(x) \delta_{V_i(x)} + p_\infty(x) \delta_\infty,\]
    %where $\delta$ is the Dirac measure and the weights are defined as 
    %\[p_i(x) = \frac{w(X_i)}{\sum_{j = 1}^{n} w(X_j) + w(x)}, i \in \left\{ 1, \dots, n \right\} \quad \text{and} \quad p_\infty(x) = \frac{w(x)}{\sum_{j = 1}^{n} w(X_j) + w(x)}.\]
    This weighting scheme resembles the nonconformity score computed on the target population, which makes the nonconformity scores ``look exchangeable'' at the test point \citep{tibshirani2019conformal}. In terms of the CQR, we are still constructing intervals like the one in Equation \eqref{eq:cqr}, $\eta$ is computed to incorporate the weights for a weighted interval estimate.

    %Algorithm \ref{alg:wcqr} sketches the procedure of weighted split CQR proposed by \citet{lei2021conformala}. As \citet{lei2021conformala} underline, the prediction intervals achieve the inequality \eqref{eq:cqrcovshift} if the likelihood ratio $w(x)$ is known \citep{tibshirani2019conformal}.

    To estimate the weight $w(x)$, \citet{lei2021conformala} leverage the propensity score of the treatment $e_D(x)$, which resembles the logic of inverse propensity weighting (IPW) \citep{imbens2015causal}. Under the overlap assumption, we are essentially calibrating the source covariate distribution to the target one. For example, if we are trying to construct prediction intervals of $Y(1)$ for the control group and $Y(0)$ for the treatment group, the weights are estimated as the following:
    \[w_1(x) %= \frac{\mathbb{P}(D = 1)}{\mathbb{P}(D = 0)} \cdot \frac{1 - e_D(x)}{e_D(x)}}
     \propto \frac{1 - e_D(x)}{e_D(x)}, \qquad w_0(x) %= \frac{\mathbb{P}(D = 0)}{\mathbb{P}(D = 1)} \cdot \frac{e_D(x)}{1 - e_D(x)} 
     \propto \frac{e_D(x)}{1 - e_D(x)}.\]
    %\[w_1(x) = \frac{dP_X(x)}{dP_{X \mid D = 1}} = \frac{\mathbb{P}(D = 1)}{e_D(x)} \propto \frac{1}{e_D(x)}.\]
    %Since the weighted conformal inference is invariant to the scaling of the likelihood ratio, the weights reduce to the ratio of propensity score. 
    
    %Therefore, weighted conformal inference relies on the estimation of propensity scores. \citet{lei2021conformala} prove that when either the conditional quantile function or the propensity scores are consistently estimated, the prediction intervals attain the doubly robust property of the guaranteed coverage. 

    %Now that we have the observed outcome and the prediction intervals for the counterfactuals, we can construct the prediction interval for the ITE within the source group $R = 1$. 
    %\[\mathcal{C}_{\text{ITE}}(x; d, y^{\text{obs}}) = \begin{cases}
    %    y^{\text{obs}} - \mathcal{C}_0(x) & d = 1,\\
    %    \mathcal{C}_1(x) - y^{\text{obs}} & d = 0.
    %\end{cases}\]
    %Denote $\mathcal{C}_i = \mathcal{C}_{\text{ITE}}(X_i; D_i, Y_i^{\text{obs}})$.

    \section{Conformal Inference with Attrition}

    The logic of reweighting the nonconformity score in weighted conformal inference offers a principled solution to the covariate shift problem. This paper aims to construct prediction intervals for ITE in the attrition group with guaranteed coverage. Achieving this goal requires addressing two key challenges. First, we need to construct prediction intervals for the counterfactuals and ITE within the observed group with guaranteed coverage. Second, we need to extrapolate these prediction intervals to the attrition group. The covariate shift problem arises in both stages.

    \citet{lei2021conformala} propose a two-step approach for constructing prediction intervals of treatment effects for observations with both potential outcomes missing. Firstly, they construct the prediction intervals for the counterfactuals and ITE within the group with observed outcome using the weighted split-CQR, which they refer to as the nested approach.\footnote{\citet{lei2021conformala} define the nested approach by splitting the data into two folds: using the first fold to construct prediction intervals for counterfactuals and using the second fold to compute the prediction intervals for ITEs.} Secondly, they conduct a second conformal inference on these prediction intervals of ITE to extrapolate the intervals to the unit with both potential outcomes missing. 
    
    While this two-step method offers finite-sample coverage guarantees, it faces three important limitations. First, this method is not designed to address attrition problems in randomized experiments. Second, even reframing the prediction on the ITE for units with both potential outcomes missing as a missing data problem, the method addresses covariate shift only in the first step. The second step employs an unweighted conformal inference, which implicitly requires a stronger assumption of missingness pattern as MCAR. Under MCAR, the distribution of covariates between source group and target group remains the same, justifying the use of unweighted conformal inference. Nevertheless, when the missingness pattern is MAR, their proposed method has less satisfying performance. Third, the prediction intervals are too conservative, in the sense that the coverage is mostly 1 with an interval that is too wide to be useful in practice.\footnote{Algorithm \ref{alg:wcqr}, \ref{alg:uci}, and \ref{alg:iteattr} sketch the conformal inference strategy proposed by \citet{lei2021conformala}. Figure \ref{fig:unwMC} shows the performance of their method in a Monte Carlo simulation.} 

    This paper follows a similar two-step approach. Consider the following condition:
    \begin{align*}
        \mathbb{P}( Y_i(1) - Y_i(0) &\in \mathcal{C}_i ) = \\
        &{\underbrace{\mathbb{P}( Y_i(1) - Y_i(0) \in \mathcal{C}_i \mid R_i = 1 ) P(R_i = 1)}_{\text{I}}}
         + \underbrace{\mathbb{P}\left( Y_i(1) - Y_i(0) \in \mathcal{C}_i \mid R_i = 0 \right)P(R_i = 0)}_{\text{II}},
    \end{align*}
    where $\mathcal{C}_i = \mathcal{C}_{\text{ITE}}(X_i; D_i, Y_i^{\text{obs}})$. Since only one of the potential outcomes is missing for $R = 1$, term I can be reduced to
    \begin{equation}
        \begin{aligned}
            \mathbb{P}( Y_i(1) - Y_i(0) \in \mathcal{C}_i \mid R_i = 1 ) P(R_i = 1) &= \\
            \mathbb{P}( Y_i(0) \in \mathcal{C}_0(X_i) \mid D_i = 1, R_i = 1 ) &P(D_i = 1 \mid R_i = 1) \\
             & + \mathbb{P}( Y_i(1) \in \mathcal{C}_1(X_i) \mid D_i = 0, R_i = 1 ) P(D_i = 0 \mid R_i = 1).
        \end{aligned} \label{eq:termI}
    \end{equation}
    If we are able to construct prediction intervals such that
    \[\mathbb{P}(Y_i(0) \in \mathcal{C}_0(X_i) \mid D_i = 1, R_i = 1) \geq 1 - \alpha \quad \text{and} \quad \mathbb{P}(Y_i(1) \in \mathcal{C}_1(X_i) \mid D_i = 0, R_i = 1) \geq 1 - \alpha,\]
    we can guarantee that
    \begin{align}
        \mathbb{P}( Y_i(1) - Y_i(0) \in \mathcal{C}_i \mid R_i = 1 ) \geq 1 - \alpha \label{eq:ITEcov}
    \end{align}
    %Theoretically, in the first step, by leveraging conformal inference, we can construct the prediction intervals for the counterfactuals ($\mathcal{C}_d$) with guaranteed coverage ($1 - \alpha$). Then, we can construct the prediction intervals for the ITE ($\mathcal{C}_{\text{ITE}}$) within the source group as 
    is satisfied by the following decomposition:
    \begin{align}\label{eq:iteintR1}
        \mathcal{C}_{\text{ITE}}(X; D, Y^{\text{obs}}) = \begin{cases}
            Y^{\text{obs}} - \mathcal{C}_0(X) & \text{if } D = 1\\
            \mathcal{C}_1(X) - Y^{\text{obs}} & \text{if } D = 0
        \end{cases},
    \end{align}
    where the intervals $\mathcal{C}_{\text{ITE}}$ can be perceived as the surrogate intervals or pseudo outcomes for ITE. 
    
    For term II, to extrapolate these intervals to the target group with attrition, we follow \citet{lei2021conformala} to find an interval expansion function $\check{\mathcal{C}}(\cdot)$ that maps a covariate value to an interval such that 
    \begin{align}
        \mathbb{P}\left( \mathcal{C}_{\text{ITE}} \subset \check{\mathcal{C}}_{\text{ITE}}(X) \mid R = 0 \right) \geq 1 - \gamma, \label{eq:ext}
    \end{align}
    By Bonferroni correction, we can construct the prediction intervals that satisfy the condition (\ref{eq:ext}).
    \begin{align*}
        \mathbb{P}\left(Y(1) - Y(0)  \notin \check{\mathcal{C}}_{\text{ITE}}(X) \mid R = 0\right) &\leq \\
        \mathbb{P}\left( Y(1) - Y(0) \notin \mathcal{C}_{\text{ITE}} \mid R = 0 \right) &+ \mathbb{P}\left( \mathcal{C}_{\text{ITE}} \not\subset \check{\mathcal{C}}_{\text{ITE}}(X) \mid R = 0 \right) \leq \alpha + \gamma.
    \end{align*}

    As mentioned above, we need to account for the covariate shift problems in both steps: (1) from treatment (control) units to the control (treatment) units inside the observed group, and (2) from the observed group ($R = 1$) to the attrition group ($R = 0$). Although weighted conformal inference provides a principled solution to covariate shift, it still relies on the empirical distribution of nonconformity scores and therefore can produce relatively wide prediction intervals in finite samples. Moreover, because our framework requires addressing two covariate shift problems sequentially, estimation error from the two propensity score models may compound, leading to further efficiency losses in the resulting intervals. Finally, our goal is to construct prediction intervals that attain valid asymptotic coverage under some regularity conditions. Therefore, following \citet{yang2024doubly} and \citet{gao2025role}, instead of deriving the threshold of nonconformity score from the empirical distribution which produces overly wide prediction intervals with only finite sample coverage, we leverage the semiparametric efficiency theory for identification. Specifically, as it will be explained in detail in the next section, we identify the threshold of nonconformity score by deriving the efficient influence function (EIF). We adopt the following notations for simplicity:
    \begin{align*}
        \mathbb{P}\left( Y(d) \in \mathcal{C}_{d} \mid D = 1 - d, R = 1 \right) &= \mathbb{P}\left( V_{d} < \eta_{\alpha, d} \mid D = 1 - d, R = 1 \right)\\
        \mathbb{P}\left( \mathcal{C}_{\text{ITE}} \subset \check{\mathcal{C}}_{\text{ITE}} \mid R = 0 \right) &= \mathbb{P}\left( V_{\mathcal{C}} < \eta_{\gamma, \mathcal{C}} \mid R = 0 \right),
    \end{align*}
    where $V_d$ and $V_\mathcal{C}$ are the nonconformity scores for the counterfactuals and ITE, respectively, $\eta_{\alpha, d}$ and $\eta_{\gamma, \mathcal{C}}$ are the thresholds of the nonconformity scores for the counterfactuals and ITE, respectively. Our goal is to identify $\eta_{\alpha, d}$ as $(1 - \alpha)$ quantile of $V_{d}$ and $\eta_{\gamma, \mathcal{C}}$ as $(1 - \gamma)$ quantile of $V_{\mathcal{C}}$.

    \section{Conformal Inference with Semiparametric Efficient Estimator}

    In this section, I first derive the EIF for identifying the semiparametric efficient estimator used to construct prediction intervals: $\eta_{\alpha, d}$ for the observed group and $\eta_{\gamma, \mathcal{C}}$ for the attrition group. Next, I present the details for implementing the identification algorithm. %Lastly, I derive the asymptotic properties for the prediction interval of ITE for the attrition group.

    \subsection{Conformal Inference for Counterfactuals and ITE on Observed Group}

    As outlined above, we adopt a two-step approach to construct prediction intervals for the ITE in the attrition group. In the first step, we construct prediction intervals for the counterfactuals $Y(d)$ for those with $D = 1 - d$ for $d \in \left\{ 0, 1 \right\}$ and then prediction intervals for ITE within the observed group using the semiparametric efficient estimator of the quantile of interest $\eta_{\alpha, d}$. In other words, to do so, it is important to extrapolate the information from units with $D = d$ to the group with $D = 1 - d$, for $d \in \left\{ 0, 1 \right\}$. Therefore, given a desired coverage level $(1 - \alpha)$, we identify $\eta_{\alpha, d}$ by linking the observed distribution of nonconformity score, $V_d \mid D = d, R = 1$, to the target distribution, $V_d \mid D = 1 - d, R = 1$. % This step extrapolates information from units with $D = d$ to the group with $D = 1 - d$.

    \begin{lemma}[Setting 1 in Theorem 1 of \citet{gao2025role}]\label{lem:ident1}
        Under Assumption \ref{asm1}, we have
        \begin{align*}
            1 - \alpha &= \mathbb{P}\left( V_d < \eta_{\alpha, d} \mid D = 1 - d, R = 1 \right)\\
            &= \mathbb{E}_X\left[ \mathbb{P}\left( V_d < \eta_{\alpha, d} \mid D = d, R = 1, X \right) \mid D = 1 - d, R = 1 \right].
        \end{align*}
    \end{lemma}

    \begin{proof}
        %Lemma \ref{lem:ident1} is immediate following the law of iterated expectations and the unconfoundedness.
        See Section \ref{sec:pflem12} for the proof of Lemma \ref{lem:ident1}.
    \end{proof}

    This identification formula is the foundation for estimating $\eta_{\alpha, d}$ and ensures the constructed prediction intervals $\mathcal{C}_d$ and $\mathcal{C}_{\text{ITE}}$ achieve desired coverage in the target data. We now derive the EIF for $\eta_{\alpha, d}$ under some regularity conditions following \citet{yang2024doubly} and setting 1 of \citet{gao2025role}. 
    
    \begin{theorem}[Lemma 1 of \citet{yang2024doubly}]\label{thm:eifcf}
        Suppose $\mathbb{E}\left[ \frac{\left( \mathbb{P}\left( D = 1 - d, R = 1 \mid X \right) \right)^2}{\mathbb{P}\left( D = d, R = 1 \mid X \right)} \right]$ is finite and that the density of the conditional distribution of $V_d \mid D = d, R = 1$ at $\eta_{\alpha, d}$ is bounded away from zero. Then under Assumption \ref{asm1} and \ref{asm3}, the efficient influence function of the conditional quantile $\eta_{\alpha, d}$ under $D = 1$ for constructing prediction interval of $Y(1)$ is given up to a proportionality constant by
        \begin{equation}
            \begin{aligned}
            \psi_1\left( \eta_{\alpha, 1}, X; m, e_R, \pi_D \right) &= R(1 - D)\left[ m_1(\eta_{\alpha, 1}, X) - (1 - \alpha) \right] \\
            &+ \frac{DR e_R(X, 0)}{\pi_D(X) e_R(X, 1)} \left[ \mathds{1}_{\left\{ V_1 < \eta_{\alpha, 1} \right\}} - m_1(\eta_{\alpha, 1}, X) \right],
        \end{aligned}\label{eq:eif1}
        \end{equation}
        and the one under $D = 0$ for constructing prediction interval of $Y(0)$ is given up to a proportionality constant by
        \begin{equation}
            \begin{aligned}
                \psi_0\left( \eta_{\alpha, 0}, X; m, e_R, \pi_D \right) &= RD \left[m_0(\eta_{\alpha, 0}, X) - (1 - \alpha)\right]\\
                & + \frac{(1 - D) R e_R(X, 1) \pi_D(X)}{e_R(X, 0)}\left[ \mathds{1}_{\left\{ V_0 < \eta_{\alpha, 0} \right\}} - m_0(\eta_{\alpha, 0}, X) \right],
            \end{aligned}
        \end{equation}
        where 
        \begin{align*}
            m_d(\eta_{\alpha, d}, X) &\coloneq \mathbb{E}\left[ \mathds{1}_{\left\{ V_d < \eta_{\alpha, d} \right\}} \mid X, D = d, R = 1 \right]\\
            \pi_D(X) \coloneq \frac{e_D(X)}{1 - e_D(X)} \quad &\text{and} \quad e_R(X, d) \coloneq \mathbb{P}(R = 1 \mid D = d, X).
        \end{align*}
    \end{theorem}

    \begin{proof}
        See Section \ref{sec:pfeifcf} for the proof of Theorem \ref{thm:eifcf}.
    \end{proof}

    In practice, by finding the smallest value that satisfies the sample moment conditions for the EIFs
    \begin{align*}
        \frac{1}{N} \sum_{i = 1}^{N} \psi_d\left( \hat{\eta}_{\alpha, d}, X_i; \hat{m}, \hat{e}_R, \hat{\pi}_D \right) \geq 0,
    \end{align*}
    we can identify $\eta_{\alpha, d}$ and construct the prediction intervals for the counterfactuals, and then the ITE as in Equation \eqref{eq:iteintR1}.

    Now we have been able to identify $\eta_{\alpha, d}$ for the observed group. We need to similarly identify $\eta_{\gamma, \mathcal{C}}$ to extrapolate the prediction intervals of ITE to the attrition group.

    \subsection{Conformal Inference for ITE with Attrition}

    The second step is to extrapolate the prediction intervals of ITE from the observed group to the attrition group with attrition. Intuitively, we are constructing intervals that cover the prediction intervals of ITE constructed in the first step. We build on the interval expansion approach proposed by \citet{lei2021conformala} to construct the prediction intervals for ITE in the attrition group. Specifically, we treat the prediction intervals obtained in the first step as surrogates for the ITE intervals in the attrition group and jointly calibrate the left and right endpoints of $\mathcal{C}_{\text{ITE}}$ to construct the prediction interval $\check{\mathcal{C}}_{\text{ITE}}(X)$ for the attrition group. Unlike the unweighted conformal inference that fails to account for the covariate shift problem, we again derive the EIF for this secondary conformal inference step \citep{gao2025role,yang2024doubly}. 
    
    Analogous to Lemma \ref{lem:ident1}, we extend the information from the observed group in the first step to the attrition group by relating the study distribution of nonconformity score $V_\mathcal{C} \mid R = 1$ to the target distribution $V_\mathcal{C} \mid R = 0$. This yields the following identification formula for the quantile of interest $\eta_{\gamma, \mathcal{C}}$ for ITE in the attrition group.

    \begin{lemma}\label{lem:ident2}
        Under Assumption \ref{asm2}, we have the coverage for the interval of ITE from the attrition group as 
        \begin{align*}
            1 - \gamma &= \mathbb{P}\left(\mathcal{C} \subset \check{\mathcal{C}}_{\text{ITE}}(X)\right)\\
            &= \mathbb{P}\left( V_\mathcal{C} < \eta_{\gamma, \mathcal{C}} \mid R = 0 \right)\\
            &= \mathbb{E}_{X, D} \left[ \mathbb{P}\left( V_\mathcal{C} < \eta_{\gamma, \mathcal{C}} \mid X, D, R = 1 \right) \mid R = 0 \right].
        \end{align*}
    \end{lemma}

    \begin{proof}
        See Section \ref{sec:pflem12} for the proof of Lemma \ref{lem:ident2}.
    \end{proof}

    Similarly, this identification formula guarantees the constructed prediction intervals $\check{\mathcal{C}}_{\text{ITE}}(X)$ achieve the desired coverage in the attrition group and motivates the EIF for $\eta_{\gamma, \mathcal{C}}$.

    \begin{theorem}\label{thm:eifint}
        Suppose $\mathbb{E}\left[ \frac{\left( \mathbb{P}\left( R = 0 \mid X, D \right)^2 \right)}{\mathbb{P}\left( R = 1 \mid X, D \right)} \right]$ is finite and the density of the conditional distribution of $V_\mathcal{C} \mid R = 0$ at $\eta_{\gamma}$ is bounded away from zero. Then under Assumption \ref{asm2} and \ref{asm3}, the efficient influence function of the conditional quantile $\eta_{\mathcal{C}, \gamma}$ is given up to a proportionality constant by
        \begin{equation}
            \begin{aligned}
                \psi_\mathcal{C}\left( \eta_{\gamma, \mathcal{C}}, X; m_{\mathcal{C}}, \pi_R \right) &= (1 - R) \left[ m_{\mathcal{C}}\left( \eta_{\gamma, \mathcal{C}}, X, D  \right) - (1 - \gamma) \right] \\
                &+ \frac{R}{\pi_R(X, D)} \left[ \mathds{1}_{\left\{ V_\mathcal{C} < \eta_{\gamma, \mathcal{C}} \right\}} - m_\mathcal{C}\left( \eta_{\gamma, \mathcal{C}}, X, D \right) \right],
            \end{aligned}
        \end{equation}
        where 
        \begin{align*}
            m_\mathcal{C}\left( \eta_{\gamma, \mathcal{C}}, X, D \right) &\coloneq \mathbb{E}\left[ \mathds{1}_{\left\{ V_\mathcal{C} < \eta_{\gamma, \mathcal{C}} \right\}} \mid X, D, R = 1 \right]\\
            \pi_R(X, D) &\coloneq \frac{\mathbb{P}(R = 1 \mid X, D)}{\mathbb{P}(R = 0 \mid X, D)}.
        \end{align*}
    \end{theorem}

    \begin{proof}
        See Section \ref{sec:pfeifint} for the proof of Theorem \ref{thm:eifint}.
    \end{proof}

    Lemma \ref{lem:ident2} and Theorem \ref{thm:eifint} differ from Theorem 4 in \citet{gao2025role} in one key point. Under the MAR assumption, we do not require the conditional independence $\mathbb{P}\left( D = d \mid X, R = 1 \right) = \mathbb{P}\left( D = d \mid X, R = 0 \right)$ of the treatment and the missingness pattern.\footnote{By contrast, \citet{gao2025role} impose the conditional independence when deriving their identification formula. If $\mathbb{P}\left( D = d \mid X, R = 1 \right) = \mathbb{P}\left( D = d \mid X, R = 0 \right)$, Lemma \ref{lem:ident2} reduces to $1 - \gamma = \mathbb{E}_X\left[ \mathbb{P}\left( V_\mathcal{C} < \eta_{\gamma, \mathcal{C}} \mid X, R = 1 \right) \mid R = 0 \right]$, which is the result in \citet{gao2025role}. However, this requires the stronger assumption that attrition is conditionally independent of treatment status.}

    In practice, we can identify $\eta_{\gamma, \mathcal{C}}$ by finding the smallest value that satisfies the sample moment conditions for the EIF
    \begin{align*}
        \frac{1}{N} \sum_{i = 1}^{N} \psi_\mathcal{C}\left( \hat{\eta}_{\gamma, \mathcal{C}}, X_i; \hat{m}_\mathcal{C}, \hat{\pi}_R \right) \geq 0.
    \end{align*}

    \subsection{Estimation Algorithm}

    \tikzstyle{block} = [rectangle, minimum width=2cm, minimum height=1cm, text centered, draw=black, line width=1pt]
    \tikzstyle{param} = [rectangle, minimum width=2cm, minimum height=1cm, text centered, line width=1pt]
    \tikzstyle{arrow} = [thick,->,>=stealth]
    \tikzstyle{arrow1} = [thick]
    \tikzstyle{edge from parent}=[->,>=stealth, thick,draw]

        \begin{figure}
            \centering
            \caption{Workflow of Algorithm}
            \label{fig:algflow}
            \begin{threeparttable}
                \makebox[\textwidth][c]{
                \begin{tikzpicture}[scale=0.85, transform shape, node distance=2cm, auto, edge from parent fork down]
                    \tikzstyle{level 1}=[sibling distance=60mm,level distance=15ex] 
                    \tikzstyle{level 2}=[sibling distance=30mm,level distance=15ex] 
    
                    %\draw [thick, black,
                    %        decorate, 
                    %        decoration = {calligraphic brace,
                    %            raise=5pt,
                    %            amplitude=5pt}] (-9,-8.5) --  (-9, 0)
                    %node[pos=0.5,left=10pt,black]{Step I};
    %
                    %\draw [thick, black,
                    %        decorate, 
                    %        decoration = {calligraphic brace,
                    %            raise=5pt,
                    %            amplitude=5pt}] (-9,-19.5) --  (-9,-11)
                    %node[pos=0.5,left=10pt,black]{Step II};

                    \node (input) [block] {Data $(X, Y, D, R)$}
                    child{node (pretrain) [block] {Pretraining}
                        child{node (nuis) [param] {$\hat{q}_{Y(d)}, \hat{\pi}_D(X), \hat{e}_R(X, d)$}}
                        %child{node (ps) [param] {$\hat{e}_D(X), \hat{e}_R(D, X)$}}
                    }
                    child{node (train) [block] {Training}
                        child{node (fold1) [block] {Fold 1}}
                        child{node (fold2) [block] {Fold 2}}
                    }
                    child{node (calibration) [block] {Calibration}
                        child{node (Q) [param] {$V_d$}}
                        child{node (q) [param] {$\hat{\eta}_{\alpha, d}$}}
                    };
                    \node (int) [block, below of = train, yshift = -4cm] {Prediction Intervals for Counterfactuals and ITE with $R = 1$};

                    \node (R1) [block, below of = int, xshift = -4cm] {Observed Group ($R = 1$)}
                    child {node (R1tr) [block] {Training}
                        [sibling distance=10mm]
                        child {node (nuisance) [param] {$\hat{h}_{\mathcal{C}}^\text{L}, \hat{h}_{\mathcal{C}}^\text{R}, \hat{\pi}_R(X, D), \hat{m}_{\mathcal{C}}$}}
                        %child {node (mint) [param] {$\hat{m}_{\mathcal{C}}$}}
                    }
                    child {node (R1cal) [block] {Calibration}};
                    \node (R0) [block, right of = R1cal, xshift = 3cm] {Attrition Group ($R = 0$)};
                    \node (QC) [param, below of = R1cal, xshift = 1cm, yshift = -1.5cm] {$V_{\mathcal{C}}$};
                    \node (qC) [param, below of = R0, xshift = -1cm, yshift = -1.5cm] {$\hat{\eta}_{\gamma, \mathcal{C}}$};
                    
                    \node (ITEint) [block, below of = int, yshift = -9cm] {Prediction Intervals of ITE with $R = 0$};
    
                    \draw [arrow, dashed] (nuis.south) |- ++(0, -0.4) -| (fold1.south);
                    \draw [arrow, dashed] (fold1) -- node[anchor = north]{$\hat{\eta}_{\alpha, d}^{\text{init}}$} (fold2);
                    \coordinate (aux1) at ($(calibration.west) - (1.5, 0)$);
                    %\coordinate (vertmid) at ($(calibration.west |- aux1)$);
                    \draw [arrow, dashed] (fold2.east) -| node[anchor = east, pos = 0.75] {$\hat{m}_d$} (aux1)   -- (calibration.west);
                    \coordinate (aux2) at ($(calibration.west) - (0.9, 0)$);
                    \draw [arrow1, dashed] (nuis.south) |- ++(0, -1) -| (aux2);
                    %\coordinate (aux3) at ($(calibration.west) - (1.2, 0)$);
                    %\draw [arrow1, dashed] (nuis.south) |- ++(0, -0.7) -| (aux3);
                    \draw [arrow, dashed] (Q) -- (q);
                    \path (calibration) -- coordinate [midway] (eif) (q);
                    \node [above] at (eif) {$\psi_d$};
                    \coordinate (aux4) at ($(calibration.south) - (0, 3)$);
                    \draw [arrow1] (Q.south) |- (aux4);
                    \draw [arrow1] (q.south) |- (aux4);
                    \coordinate (aux5) at ($(int.north) - (0, -0.8)$);
                    \draw [arrow] (aux4) |- (aux5) -- (int);
                    \draw [arrow] (int.south) |- ++(0, -0.5) -| (R1.north);
                    \draw [arrow] (input.east) -| ++(7, 0) |- (R0.east);
                    \draw [arrow, dashed] (nuisance.east) -| ++(1.4, 0) |- (R1cal.west);
                    \path (R1cal.south) -- coordinate [midway] (aux6) (R0.south);
                    \coordinate (aux7) at ($(aux6) - (0, 1)$);
                    \draw [arrow1] (R1cal.south) |- (aux7);
                    \draw [arrow1] (R0.south) |- (aux7);
                    \draw [arrow] (aux7) |- ++(0, -0.7) -| (QC.north);
                    \draw [arrow] (aux7) |- ++(0, -0.7) -| (qC.north);
                    \path (aux7) -- coordinate [midway] (psiC) (qC.north);
                    \node [above] at (psiC) {$\psi_{\mathcal{C}}$};
                    \draw [arrow, dashed] (QC.east) -- (qC.west);
                    \coordinate (aux8) at ($(aux7) - (0, 3)$);
                    \draw [arrow1] (QC.south) |- (aux8);
                    \draw [arrow1] (qC.south) |- (aux8);
                    \draw [arrow] (aux8) |- ++(0, -0.6) -| (ITEint.north);
    
                    \node (leftup) [left of = input, xshift = -7cm, yshift = 0.8cm] {};
                    \node (rightup) [right of = input, xshift = 7cm, yshift = 0.8cm] {};
                    \node (leftdown1) [left of = input, xshift = -7cm, yshift = -9.4cm] {};
                    \node (rightdown1) [right of = input, xshift = 7cm, yshift = -9.4cm] {};
    
                    \node (leftdown2) [left of = input, xshift = -7cm, yshift = -10cm] {};
                    \node (rightdown2) [right of = input, xshift = 7cm, yshift = -10cm] {};
                    \node (leftbottom) [left of = input, xshift = -7cm, yshift = -20.5cm] {};
                    \node (rightbottom) [right of = input, xshift = 7cm, yshift = -20.5cm] {};
    
                    \begin{scope}[on background layer]
                    \node[
                    fill=gray!10,
                    rounded corners,
                    inner xsep=16pt,
                    inner ysep=0pt,
                    fit=(leftup) (rightup) (leftdown1) (rightdown1),
                    label={[anchor=north west,font=\bfseries]north west:Step I}
                    ] {};
                    \end{scope}
    
                    \begin{scope}[on background layer]
                    \node[
                    fill=gray!20,
                    rounded corners,
                    inner xsep=16pt,
                    inner ysep=0pt,
                    fit=(leftdown2) (rightdown2) (leftbottom) (rightbottom),
                    label={[anchor=north west,font=\bfseries]north west:Step II}
                    ] {};
                    \end{scope}

                \end{tikzpicture}
                }
            \begin{tablenotes}
                \footnotesize
                \setlength\labelsep{0pt}
                \item \textit{Note}: This figure illustrates the overall workflow of the proposed algorithm for constructing prediction intervals for ITE with attrition. The light shaded part represents Step I, which constructs prediction intervals for counterfactuals and ITE within the observed group using semiparametric efficient estimators. The dark shaded part represents Step II, which extrapolates the prediction intervals of ITE to the attrition group using the interval expansion approach with semiparametric efficient estimators.
            \end{tablenotes}
            \end{threeparttable}

        \end{figure}
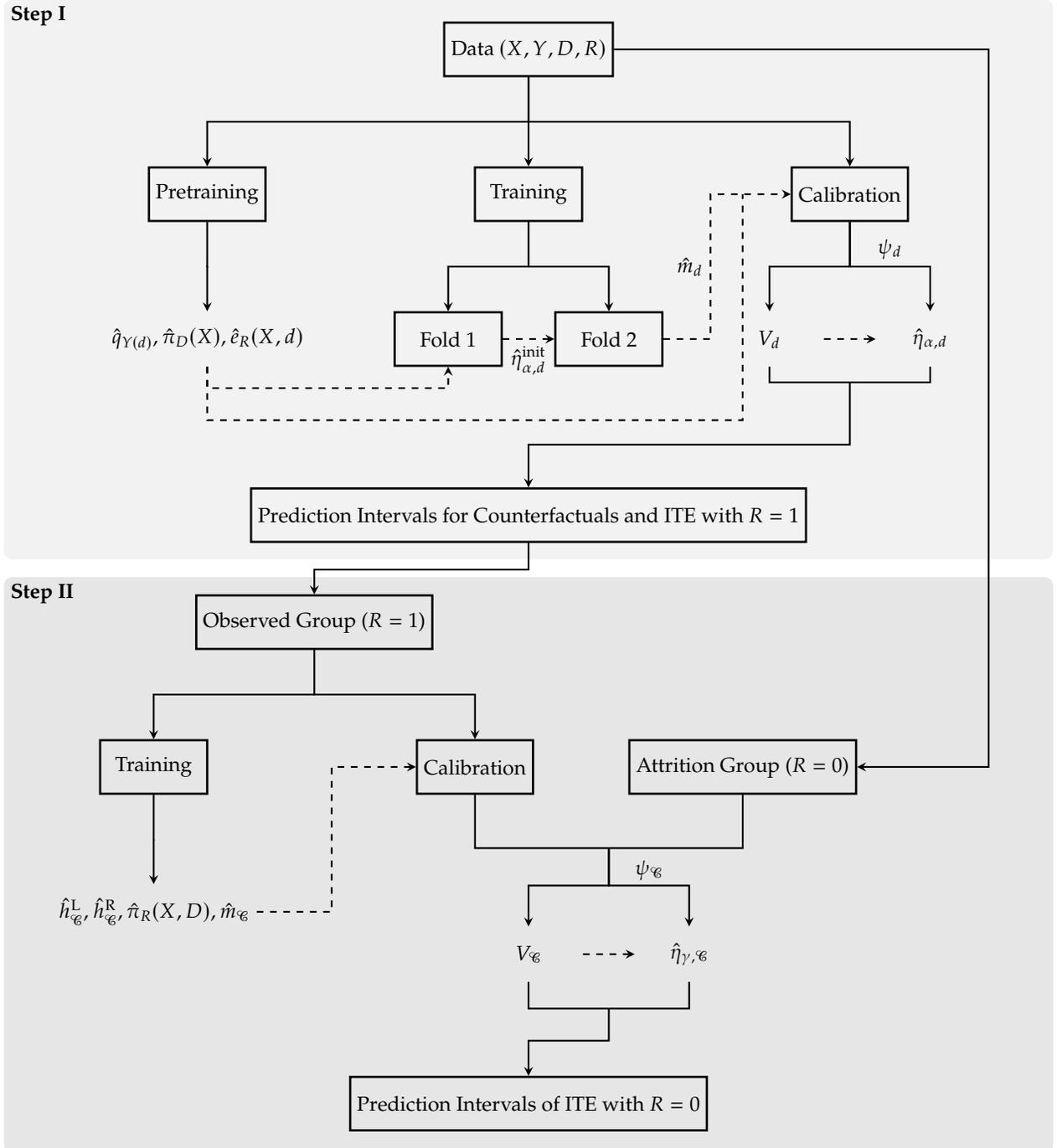
        
    Now that we have discussed how to obtain the EIFs, we describe the algorithm used. Figure \ref{fig:algflow} illustrates the overall workflow, and Algorithm \ref{alg:iteattreif} sketches the implementation details. To construct prediction intervals for ITE in the attrition group, we adopt the split conformal inference \citep{lei2014distributionfree} in combination with the exact nested approach for interval estimation of ITE \citep{lei2021conformala}. The procedure consists of two main steps. First, following Theorem \ref{thm:eifcf}, we use data from the observed group to construct prediction intervals for counterfactuals and ITE. Second, following Theorem \ref{thm:eifint}, we take the interval estimates of the ITE from the first step and use them to construct prediction intervals for the ITE in the attrition group. %of the split conformal inference and exact nested framework for constructing ITE prediction intervals under attrition, using semiparametric efficient estimators.  

    We describe the full algorithm as follows. As shown by the light shaded part in Figure \ref{fig:algflow}, Step 1 begins by splitting the data into three folds: pretraining ($\mathcal{Z}_{\text{pr}}$), training ($\mathcal{Z}_{\text{tr}}$), and calibration ($\mathcal{Z}_{\text{ca}}$). This differs from the original split conformal inference procedure of \citet{lei2014distributionfree} and the method proposed by \citet{gao2025role}, which uses a two-fold random split and trains the learner on the first fold. In our setting, identifying the semiparametric efficient estimator $\eta_{\alpha, d}$ using the EIF requires fitting the model on a hold-out fold for nuisance parameter estimation. A two-fold split risks overfitting the nuisance parameter estimators, which undermines both predictive performance and EIF identification. 
    
    On $\mathcal{Z}_{\text{pr}}$, we train the conditional quantile models $\hat{q}_\beta(\cdot)$ and learners for the treatment propensity score $\hat{\pi}_D(X)$ and the attrition propensity score $\hat{e}_R(d, X)$. We then evaluate the nonconformity score $V_i$ for each observation in $\mathcal{Z}_{\text{tr}} \cup \mathcal{Z}_{\text{ca}}$ following CQR: $V_i = \max\left\{ \hat{q}_{\alpha_{\text{lo}}}(X_i; \mathcal{Z}_{\text{pr}}) - Y_i, Y_i - \hat{q}_{\alpha_{\text{hi}}}(X_i; \mathcal{Z}_{\text{pr}}) \right\}$ to identify $\eta_{\alpha, d}$. Following \citet{gao2025role}, we further split the training fold into two subfolds $\mathcal{Z}_{\text{tr, 1}}$ and $\mathcal{Z}_{\text{tr, 2}}$ to reduce the computational burden of estimating the full conditional distribution $m_d\left( \eta_{\alpha, d}, X \right)$ for root finding. We also adopt the localized debiased machine learning approach of \citet{kallus2024localized}: $\mathcal{Z}_{\text{tr, 1}}$ is used to construct an initial estimator $\hat{\eta}_{\alpha, d}^{\text{init}}$ of $\eta_{\alpha, d}$, and $\mathcal{Z}_{\text{tr, 2}}$, is then used to train the learner for $\hat{m}_{d}\left( \hat{\eta}_{\alpha, d}^{\text{init}}, X \right)$.

    Then, on $\mathcal{Z}_{\text{ca}}$, we identify the semi-parametric efficient estimator $\hat{\eta}_{\alpha, d}$ by finding the smallest values that satisfy 
    \begin{align*}
        \sum_{i \in \mathcal{I}_{\text{ca}}} \psi_d\left( \hat{\eta}_{\alpha, d}, X_i; \hat{m}_d, \hat{e}_R, \hat{\pi}_D \right) \geq 0,
    \end{align*}
    where $\hat{m}_d$ is trained on $\mathcal{Z}_{\text{tr, 2}}$ and $\hat{e}_R$ and $\hat{\pi}_D$ are trained on $\mathcal{Z}_{\text{pr}}$. We then can construct the prediction interval for counterfactuals as 
    \begin{align*}
        \mathcal{C}_{d}(X) = [\hat{Y}_i^{\text{L}}(d), \hat{Y}_i^{\text{R}}(d)] = [\hat{q}_{\alpha_{\text{lo}}}(X_i; \mathcal{Z}_{\text{pr}}) - \hat{\eta}_{\alpha, 1 - d}, \hat{q}_{\alpha_{\text{hi}}}(X_i; \mathcal{Z}_{\text{pr}}) + \hat{\eta}_{\alpha, 1 - d}].
    \end{align*}
    Leveraging these prediction intervals for counterfactuals, we can construct the prediction interval for ITE in the source group as in Equation \eqref{eq:iteintR1}
    \begin{align*}
        \mathcal{C}_{\text{ITE}}(X) = \begin{cases}
            Y^{\text{obs}} - \mathcal{C}_0(X) & \text{if } D = 1\\
            \mathcal{C}_1(X) - Y^{\text{obs}} & \text{if } D = 0
        \end{cases}
    \end{align*}

    Step 2 shown in the dark shaded area aims to construct the prediction intervals for ITE in the attrition group. We take all observations with observed outcomes and prediction intervals for ITE from $\mathcal{Z}_{\text{ca}} \equiv \left( X_i, Y_i, D_i, R_i, \mathcal{C}_i \right)$ in Step 1 and treat them as the observed data $\mathcal{Z}_{\text{obs}}$. This set is randomly split into a training fold $\mathcal{Z}_{\text{obstr}}$ and a calibration fold $\mathcal{Z}_{\text{obsca}}$. On $\mathcal{Z}_{\text{obstr}}$, we train $\hat{h}^{\text{L}}(\cdot)$ and $\hat{h}^{\text{R}}(\cdot)$ to model the conditional mean of the lower and upper endpoints of $\mathcal{C}_i$, respectively. We also train the attrition propensity score $\hat{\pi}_R(X, D)$ and the conditional distribution $\hat{m}_{\mathcal{C}}\left( \eta_{\gamma, \mathcal{C}}, X, D \right)$. For each observation in $\mathcal{C}_{\text{obs}}$, we compute the nonconformity score $V_\mathcal{C} = \max\left\{ \hat{h}^{\text{L}}(X_i; \mathcal{Z}_{\text{obstr}}) - \mathcal{C}_i^{\text{L}}, \mathcal{C}_i^{\text{R}} - \hat{h}^{\text{R}}(X_i, \mathcal{Z}_{\text{obstr}}) \right\}$. Next, we combine the calibration fold $\mathcal{Z}_{\text{obsca}}$ with all observations from the attrition group $\mathcal{Z}_{\text{att}}$. This combination of $R = 1$ and $R = 0$ explicitly addresses the covariate shift between the observed and attrition groups. Finally, for observations from $\mathcal{Z}_{\text{obsca}} \cup \mathcal{Z}_{\text{att}}$, we identify the semi-parametric efficient estimator $\hat{\eta}_{\gamma, \mathcal{C}}$ by finding the smallest value that satisfies
    \begin{align*}
        \sum_{i \in \mathcal{I}_{\text{obsca}} \cup \mathcal{I}_{\text{att}}} \psi_\mathcal{C}\left( \hat{\eta}_{\gamma, \mathcal{C}}, X_i; \hat{m}_\mathcal{C}, \hat{\pi}_R \right) \geq 0,
    \end{align*}
    where $\hat{m}_{\mathcal{C}}$ and $\hat{\pi}_R$ are trained on $\mathcal{Z}_{\text{obstr}}$. Finally, we can construct the prediction interval for ITE in the target group with attrition as
    \begin{align*}
        \check{\mathcal{C}}_{\text{ITE}}(x) = \left[\hat{h}^{\text{L}}(x; \mathcal{Z}_{\text{obstr}}) - \hat{\eta}_{\gamma, \mathcal{C}}, \hat{h}^{\text{R}}(x; \mathcal{Z}_{\text{obstr}}) + \hat{\eta}_{\gamma, \mathcal{C}}\right].
    \end{align*}

    \subsection{Asymptotic Coverage of the Prediction Interval for ITE with Attrition}

    Our choice to use semiparametric efficient estimators for constructing prediction intervals for the ITE under attrition is motivated by two considerations. First, the EIF approach allows us to address covariate shift in both steps of the procedure. Second, semiparametric efficiency ensures that the resulting estimator achieves the desired asymptotic coverage for the prediction interval, in contrast to the finite-sample, nonasymptotic coverage guarantees obtained from the weighted split-CQR approach.

    To establish the asymptotic coverage of the prediction interval for the ITE under attrition, we must ensure that both semiparametric efficient estimators $\eta_{\alpha, d}$ and $\eta_{\gamma, \mathcal{C}}$ possess the requisite asymptotic properties. These properties guarantee that the constructed prediction intervals achieve the desired asymptotic coverage. Theorem 3 of \citet{gao2025role} formally states the asymptotic properties of $\eta_{\alpha, d}$ and demonstrates the resulting asymptotic coverage of the ITE prediction intervals for the observed group. Following the same reasoning, we establish the asymptotic property of $\eta_{\gamma, \mathcal{C}}$ and the corresponding asymptotic coverage for the ITE prediction intervals in the attrition group. See Appendix \ref{sec:asym} for the regularity conditions, theorem, and technical details.

    Theorem \ref{thm:asyite} establishes a coverage guarantee for the ITE prediction interval in the attrition group. The deviation from the nominal coverage is governed by the sum of two terms: (i) a term of order $O\left( N^{-1 / 2} \right)$ when the data is split into two folds of similar size, and (ii) the product bias from the estimation of nuisance parameters $\pi_R(X, D)$ and $m_\mathcal{C}\left( \eta_{\gamma, \mathcal{C}}, X, D \right)$, which is negligible if either $\norm{\hat{\pi}_R(X, D) - \pi_R(X, D)}_2 = o_p(1)$ or $\norm{\hat{m}_\mathcal{C}\left( \hat{\eta}_{\gamma, \mathcal{C}}, X, D \right) - m_\mathcal{C}\left( \hat{\eta}_{\gamma, \mathcal{C}}, X, D \right)}_2 = o_p(1)$. The first term comes from approximating $\mathbb{E}\left[ \psi_\mathcal{C} \right]$ with $\mathbb{P}_{\mathcal{I}_2}\left( \psi_\mathcal{C} \right)$. The second term comes from the double robustness property of the EIF $\psi_\mathcal{C}$. This rate double robustness property ensures that the asymptotic coverage is robust to small perturbations in the estimation of nuisance parameters, with estimation errors affecting coverage only through second-order terms \citep{chernozhukov2018double}.

    As shown by \citet{gao2025role}, prediction intervals of counterfactuals have asymptotic coverage of $(1 - \alpha)$. Also, Theorem \ref{thm:asyite} shows that the prediction intervals of ITE in the attrition group have asymptotic coverage of $(1 - \gamma)$. Then, by Bonferroni correction,
    \begin{equation}\label{eq:itecov}
        \begin{aligned}
            \mathbb{P}&\left( Y_1 - Y_0 \notin \check{\mathcal{C}}_{\text{ITE}}(X) \mid R = 0 \right)\\
            &\leq \mathbb{P}\left( Y_1 - Y_0 \notin \mathcal{C}_{\text{ITE}} \mid R = 0 \right) + \mathbb{P}\left( \mathcal{C}_{\text{ITE}} \not\subset \check{\mathcal{C}}_{\text{ITE}}(X) \mid R = 0 \right) \leq \alpha + \gamma + o_p(1), 
        \end{aligned}
    \end{equation}
    which shows that the probability of the ITE not being contained in the prediction interval $\check{\mathcal{C}}_{\text{ITE}}(X)$ is bounded above by $(\alpha + \gamma)$ up to a negligible term. This interval expansion function $\check{\mathcal{C}}_{\text{ITE}}$ for ITE prediction intervals in the attrition group is designed to jointly calibrate the upper and lower ends of the interval estimates of ITE $\mathcal{C}_{\text{ITE}}$ from the observed group, ensuring that the overall coverage is at least $1 - (\alpha + \gamma)$. Thus, we have 
    \begin{align*}
        \mathbb{P}\left( Y_1 - Y_0 \in \check{\mathcal{C}}_{\text{ITE}}(X) \mid R = 0 \right) \geq 1 - (\alpha + \gamma),
    \end{align*}
    holds asymptotically, which guarantees the desired coverage.

    \section{Simulation Studies}\label{sec:MCsim}
    
    In this section, we evaluate the performance of the proposed method through simulation studies. Following \citet{lei2021conformala}, we adapt a variant of the simulation setting in \citet{wager2018estimation}. The simulation design assesses the performance of ITE prediction intervals under attrition using two primary metrics. First, we examine the empirical marginal coverage and average length of the ITE prediction intervals in the attrition group across different learning algorithms to identify the best-performing learner. Second, we compare the proposed method against two alternative approaches: a parametric method -- multiple imputation using Amelia II -- and a nonparametric method -- weighted CQR with unweighted nested framework.

    The data generating process (DGP) is summarized as follows. The covariate vector $X = \left( X_1, \dots, X_k \right)^\top$ is an equicorrelated multivariate Gaussian vector with mean zero and $\Var(X_i) = 1$ and $\Cov(X_i, X_j) = \rho$ for $i \neq j$. When $\rho = 0$, the covariates are independent. When $\rho > 0$, the covariates are positively correlated. The potential outcomes are generated as follows:
    \begin{align*}
        Y_1 = f(X_1)f(X_2) + \epsilon &, \quad Y_0 = \epsilon\\
        f(x) = \frac{2}{1 + \exp\left( -12(x - 0.5) \right)}&, \quad \epsilon \sim \mathcal{N}(0, 1).
    \end{align*}
    %Following \citet{wager2018estimation}, we consider the homoscedastic noise. 
    The propensity score of treatment $e_D(X)$ is generated as 
    \begin{align*}
        e_D(X) = \frac{1}{4}\left( 1 + \beta_{2, 4}(X_1) \right),
    \end{align*}
    where $\beta_{a, b}$ is the CDF of the beta distribution with parameters $a$ and $b$. According to \citet{lei2021conformala}, this ensures that $e_D(X) \in [0.25, 0.5]$, thereby providing sufficient overlap between the treatment and control group that satisfies Assumption \ref{asm3}. The propensity score of attrition $e_R(X, D)$ is generated as
    \begin{align*}
        e_R(X, D) = \text{logit}^{-1}\left( -0.25 + 0.5D + 0.2X_1 - 0.3X_2 \right),
    \end{align*}
    which ensures the MAR assumption that the missingness is correlated with observed covariates. Throughout the simulation, we set the dimension of the covariates $k = 10$. We also consider another more complicated DGP. See Section \ref{sec:DGP2} for details.

    We are interested in evaluating the properties of the ITE prediction intervals in the attrition group. Corresponding to Table \ref{tab1}, we observe $\left( X, Y, D \right)$ in the observed group $(R = 1)$ while only observe $\left( X, D \right)$ in the attrition group $(R = 0)$. The nonconformity scores $V_{\alpha, d}$ and $V_{\gamma, \mathcal{C}}$ are constructed using conformalized quantile residuals from the quantile regression model. We use \texttt{quantreg} package in \texttt{R} following \citet{lei2021conformala} and set $\alpha_{\text{lo}} = \frac{\alpha}{2}, \alpha_{\text{hi}} = 1 - \frac{\alpha}{2}$ for $V_{\alpha, d}$ and set $\gamma_{\text{lo}} = \frac{\gamma}{2}, \gamma_{\text{hi}} = 1 - \frac{\gamma}{2}$ for $V_{\gamma, \mathcal{C}}$ as the lower and upper quantile for the nonconformity scores. We estimate the propensity scores and other nuisance parameters using Generalized Linear Model (glm), Lasso and Elastic-Net Regularized Generalized Linear Model (glmnet), Random Forest, Bayesian Additive Regression Trees (BART), and Extreme Gradient Boosting (XGBoost)  as base learners via \texttt{SuperLearner} package in R. We first split 20\% of the data as the pretraining fold, and then use 75\% of the remaining data as the training fold, as suggested by \citet{sesia2020comparison}. In the second step, we further split the observed data with prediction interval of ITE from the first step into two folds with the same size for training and calibration. 

    To evaluate the overall performance of conformal inference for ITE estimation under attrition across different learning algorithms, we conduct 100 Monte Carlo simulations for sample sizes $\left\{500, 1000, 5000 \right\}$. We compare five different learning algorithms under two levels of covariate correlation: $\rho = 0$ and $\rho = 0.9$. Following \citet{lei2021conformala}, we set $\alpha = \gamma = 0.025$ to construct ITE prediction intervals for the attrition group with nominal 95\% coverage. The empirical marginal coverage of the ITE is then estimated as
    \begin{align*}
        \frac{1}{|\mathcal{I}_{R = 0}|} \sum_{R = 0} \mathds{1}_{\left\{ Y_i(1) - Y_i(0) \in \check{\mathcal{C}}_{\text{ITE}}(X_i) \right\}}.
    \end{align*}
    As noted by \citet{lei2021conformala}, a reliable method should, at very least, achieve coverage close to or above the nominal level 0.95. For benchmarking purposes, we also compute the average length of the oracle intervals, defined by the true 0.025th and 0.975th conditional quantiles. In our DGP, the errors are normally distributed, so the expected length is given by $5.54 (\approx 2 \times 1.96 \times \sqrt{2})$.

    Figure \ref{fig:MCallcov1} presents the simulation results for the empirical coverage of ITE prediction intervals in the attrition group, constructed using conformal inference with semiparametric efficient estimator across different learners. XGBoost achieves the nominal 95\% coverage when the sample size is $N = 500$, but its coverage deteriorates for $N = 1000$ and $N = 5000$. The other four algorithms maintain better coverage than XGBoost and achieve the nominal level across all sample sizes. Among these, BART attains the highest coverage, reaching nearly 1.0 for $N = 1000$ and $N = 5000$. When the sample size is $N = 5000$, all these four algorithms yield coverage close to 1.0. The empirical coverage remains stable across different covariate correlations.

    \begin{figure}[h]
        \centering
        \caption{MC Simulation Results of Conformal Inference for ITE with Attrition}
        \label{fig:MCallcov1}
        \begin{threeparttable}
            \includegraphics[width=\textwidth]{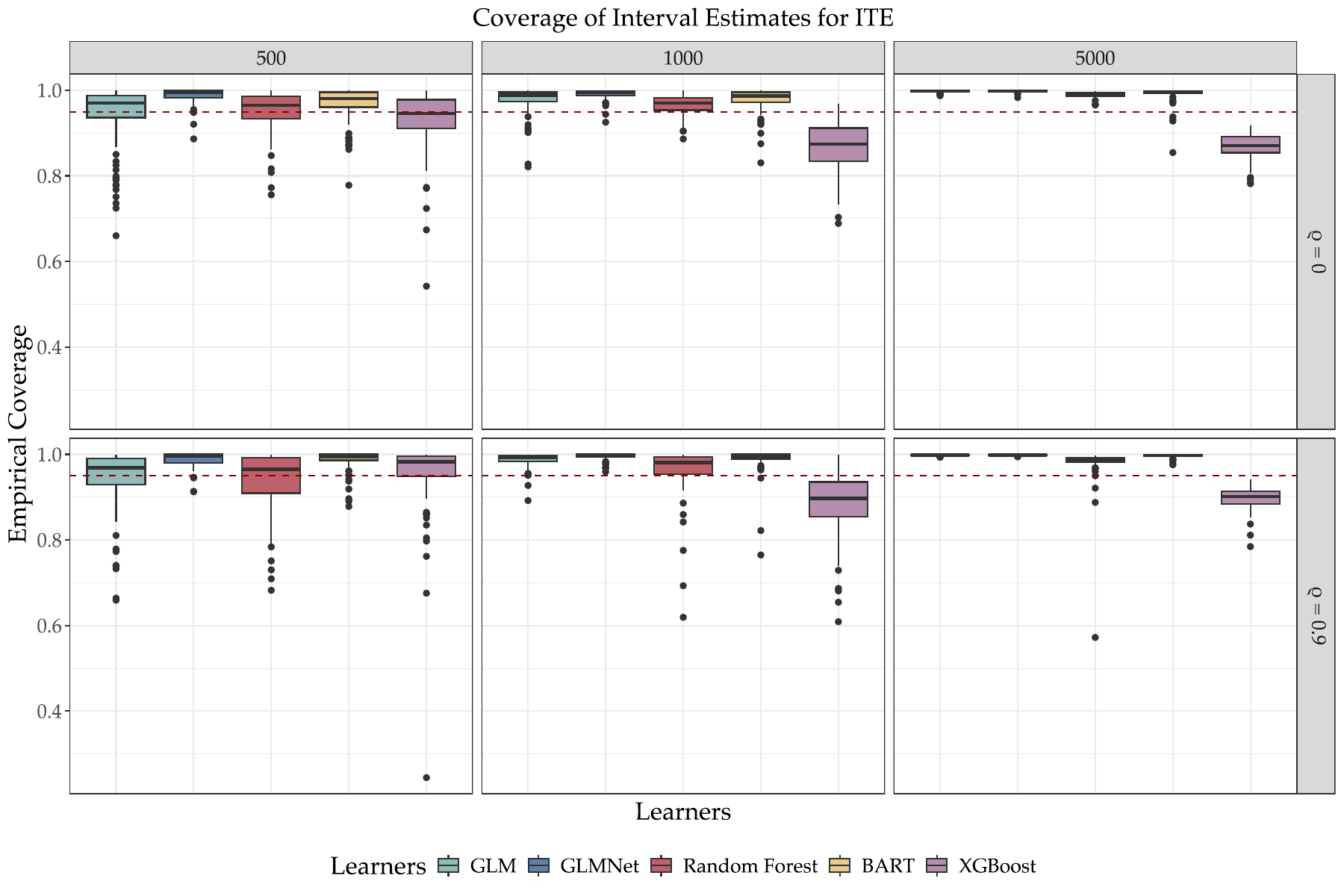}
            \begin{tablenotes}
                \footnotesize
                \setlength\labelsep{0pt}
                \item \textit{Note}: This figure shows the simulation results for the empirical coverage of prediction intervals constructed by semiparametric efficient estimator for ITE of attrition group following DGP1. The red horizontal line corresponds to the target coverage of $95\%$.
            \end{tablenotes}
        \end{threeparttable}
    \end{figure}
    
    Figure \ref{fig:MCalllen1} presents the simulation results for the average length of ITE prediction intervals in the attrition group, constructed using conformal inference with semiparametric efficient estimator across different learners. XGBoost produces the shortest intervals with poor empirical coverage. Among the remaining four algorithms, BART yields the longest intervals, consistent with its high empirical coverage. While all four achieve coverage above the nominal level, random forest offers the best balance between coverage and average length: it maintains nominal coverage while producing the shortest average interval length among these methods. Overall, the simulation results indicate that conformal inference with semiparametric efficient estimator yields ITE prediction intervals in the attrition group that achieve the desired empirical coverage and maintain reasonable interval lengths, regardless of whether the covariates are correlated or not.

    \begin{figure}[h]
        \centering
        \caption{MC Simulation Results of Conformal Inference for ITE with Attrition}
        \label{fig:MCalllen1}
        \begin{threeparttable}
            \includegraphics[width=\textwidth]{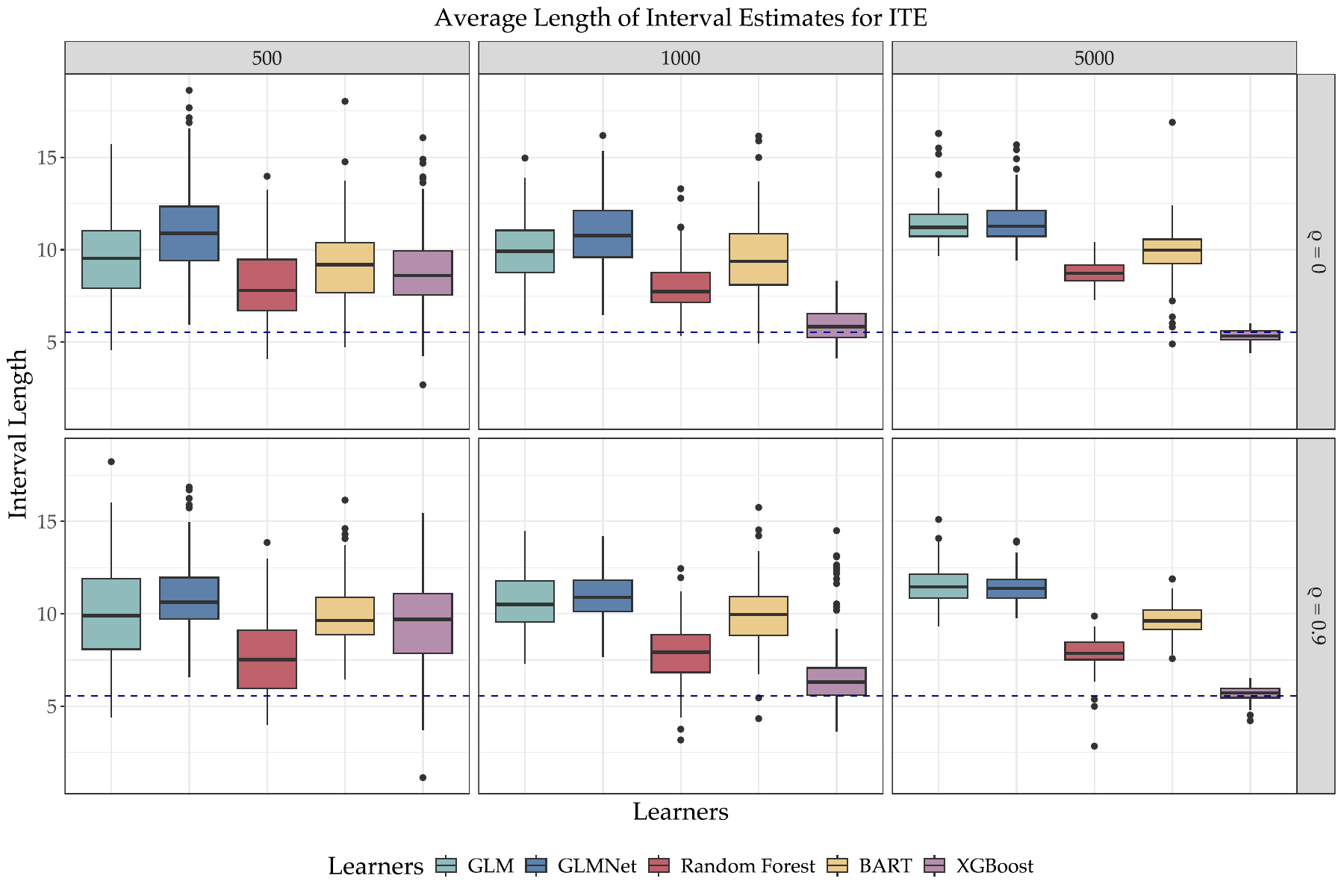}
            \begin{tablenotes}
                \footnotesize
                \setlength\labelsep{0pt}
                \item \textit{Note}: This figure shows the simulation results for the average length of prediction intervals constructed by semiparametric efficient estimator for ITE of attrition group following DGP1. The blue horizontal line corresponds to the length of oracle intervals.
            \end{tablenotes}
        \end{threeparttable}
    \end{figure}

    We further conduct a second simulation study to compare the performance of the proposed conformal inference with semiparametric efficient estimator with multiple imputation, which is widely used in political science studies, and the exact method with nested framework of conformal inference for interval outcomes proposed by \citet{lei2021conformala}. For multiple imputation, we use Amelia II, a fast algorithm using expectation-maximization (EM) with bootstrapping for multiple imputation of missing data \citep{honaker2011amelia,king2001analyzing}. Since the missingness induced by the fundamental problem of causal inference -- only one potential outcome is observed at one time -- and the missingness induced by experiment attrition are different theoretically, we use multiple imputation to impute one missingness type at a time. Specifically, we first impute the potential outcomes for the observed group, i.e., we use outcomes of observations from group $D = 1$ \& $R = 1$ to impute the counterfactuals $Y(1)$ of observations from group $D = 0$ \& $R = 1$. Similarly, we impute the counterfactuals $Y(0)$ of observations from group $D = 1$ \& $R = 1$ using outcomes of observations from group $D = 0$ \& $R = 1$. Then, we use outcomes of the observed group to impute the attrition group and construct the prediction interval for ITE in the attrition group using the corrected standard error. This procedure corresponds to the logic of the conformal inference algorithm described above. We also conduct two other multiple imputation practice for constructing the prediction interval of ITE with attrition. See Section \ref{sec:compDGP2} for the implementation details and simulation results.

    In this second simulation study, we choose random forest as the base learner for conformal inference as it performs comparatively best in the first study, and set $\alpha = \gamma = 0.025$ as well. We conduct 100 MC simulations for sample sizes among $\left\{ 500, 1000, 5000 \right\}$ and consider two scenarios of covariate correlation: uncorrelated with $\rho = 0$ and correlated with $\rho = 0.9$ \citep{lei2021conformala}. The empirical marginal coverage of conformal inference prediction interval is evaluated the same as in the first study.

    Figure \ref{fig:MCcomp1covDGP1} presents the coverage of prediction intervals for ITE with attrition constructed by three different methods. As expected, the weighted CQR with unweighted nested approach for interval estimates by \citet{lei2021conformala} has the highest coverage, almost 1 across all sample sizes and covariate correlations. Here we are using the exact method under the nested framework, which has been shown to have larger coverage in Figure \ref{fig:unwMC}. As mentioned above, this high coverage of \citet{lei2021conformala} is not surprising since this approach does not assume MAR, which cannot deal with the covariate shift problems from observed group to the attrition group. The multiple imputation method with Amelia has poor coverage across all sample sizes and covariate correlations. There could be several reasons underlying this low coverage. First, the imputation model in Amelia assumes that the complete data are multivariate normal. However, under our DGP, outcomes do not follow a normal distribution. Multiple imputation failed to tackle this distribution, which violates the normality assumption. Second, although Amelia also requires MAR assumption, it does not account for the two different covariate shift problems discussed before. Among all three methods, conformal inference with semiparametric efficient estimator performs the best. Although the coverage is around 95\% when the sample size is 500, it goes beyond nominal level when the sample size increases and is stable across different covariate correlations.

    \begin{figure}[h]
        \centering
        \caption{Comparison of Empirical Coverage of Prediction Intervals for ITE with Attrition}
        \label{fig:MCcomp1covDGP1}
        \begin{threeparttable}
            \includegraphics[width=\textwidth]{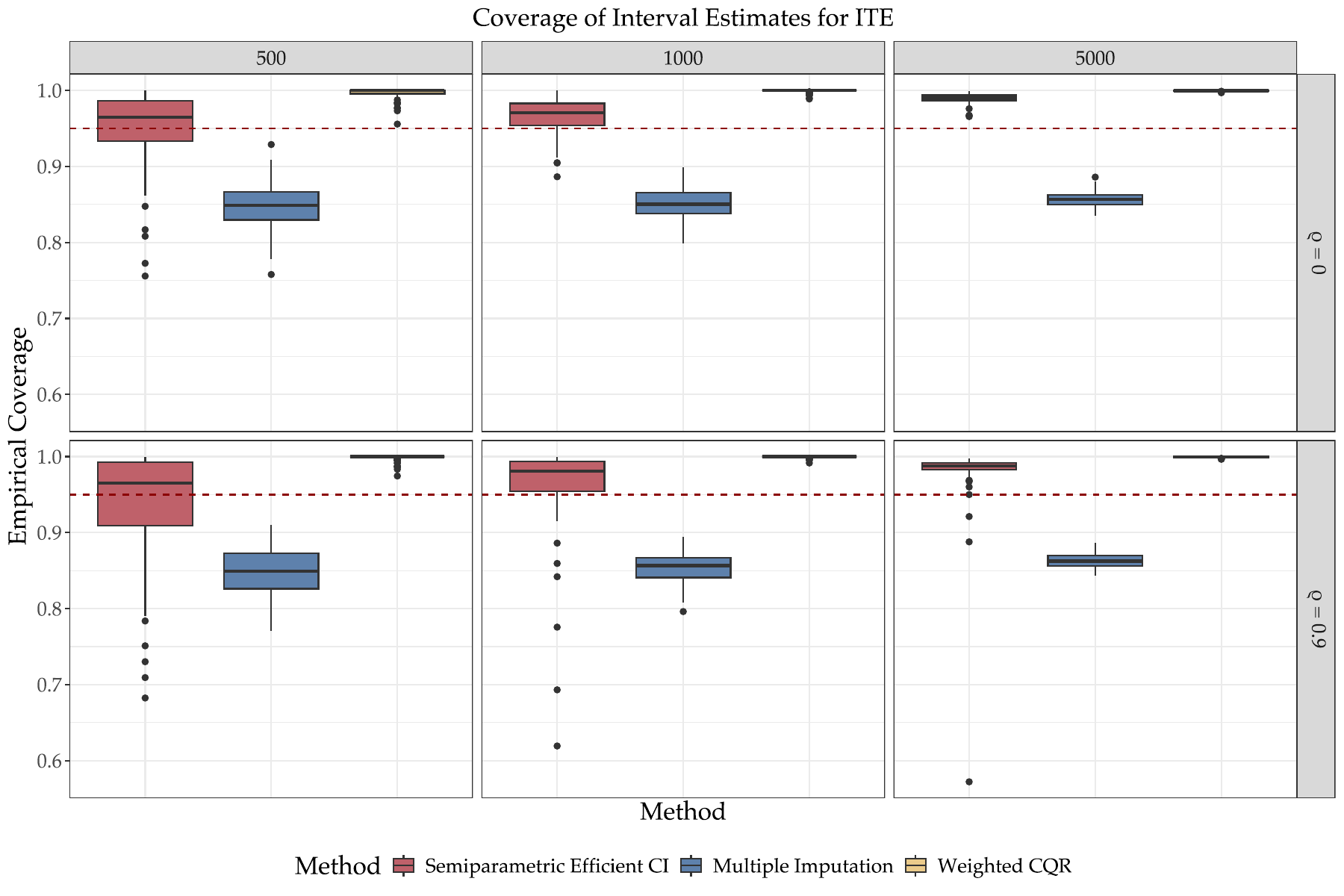}
            \begin{tablenotes}
                \footnotesize
                \setlength\labelsep{0pt}
                \item \textit{Note}: This figure shows the simulation results for the empirical coverage of prediction intervals constructed by conformal inference with semiparametric efficient estimator, multiple imputation with Amelia, and weighted CQR with unweighted nested approach for ITE of attrition group. The red horizontal line corresponds to the target coverage of $95\%$.
            \end{tablenotes}
        \end{threeparttable}
    \end{figure}

    Figure \ref{fig:MCcomp1lenDGP1} presents the average length of prediction intervals for ITE with attrition constructed by three different methods. Among three methods, the weighted CQR with unweighted nested approach has the widest prediction intervals, which is expected since it has the highest empirical coverage. However, such wide intervals are too conservative and not informative enough for inference, as they always cover the true ITE. The average length of intervals constructed by Amelia is the shortest, comparable with the length of oracle confidence intervals. However, this is misleading and not desirable since the empirical coverage is much lower than the nominal level. Notably, intervals constructed by conformal inference with semiparametric efficient estimator have the best performance across different sample sizes and covariate correlations, corresponding to the empirical coverage.

    \begin{figure}[ht]
        \centering
        \caption{Comparison of Average Length of Prediction Intervals for ITE with Attrition}
        \label{fig:MCcomp1lenDGP1}
        \begin{threeparttable}
            \includegraphics[width=\textwidth]{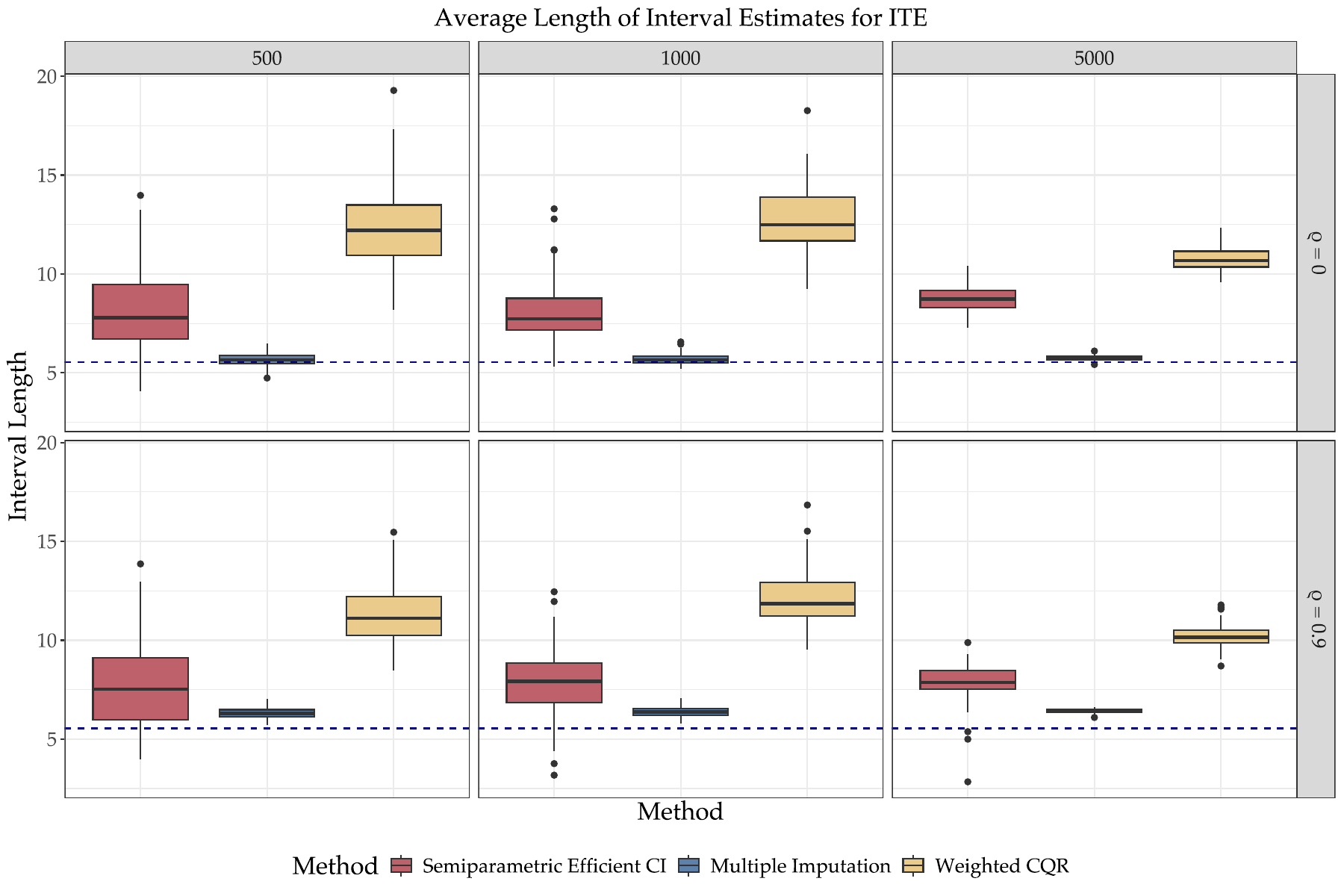}
            \begin{tablenotes}
                \footnotesize
                \setlength\labelsep{0pt}
                \item \textit{Note}: This figure shows the simulation results for the average length of prediction intervals constructed by conformal inference with semiparametric efficient estimator, multiple imputation with Amelia, and weighted CQR with unweighted nested approach for ITE of attrition group. The red horizontal line corresponds to the length of oracle intervals.
            \end{tablenotes}
        \end{threeparttable}
    \end{figure}

    Based on the simulation results from both studies, conformal inference with semiparametric efficient estimator yields ITE prediction intervals in the attrition group that achieve the desired empirical coverage and maintain reasonable interval lengths across different sample sizes and DGP. Moreover, when compared to current methods like multiple imputation and weighted CQR with unweighted nested approach, the proposed method demonstrates superior performance in terms of both empirical coverage and average interval length across different DGP. These findings highlight the effectiveness of conformal inference with semiparametric efficient estimator for constructing reliable ITE prediction intervals under attrition.

    \section{Empirical Application}

    In this section, I revisit the study by \citet{margalit2021how} and \citet{finkel2024can} to illustrate the proposed method. I show that by using conformal inference with semiparametric efficient estimator, we can construct prediction intervals for ITE in the attrition group with desired coverage. Moreover, we can aggregate ITEs to the ATE of interest for the attrition group and compute the point and interval estimates of ATE for all observations.

    \subsection{Reanalysis of \citet{margalit2021how}}

    This paper leverages a field experiment to evaluate the impact of financial markets on socioeconomic values and political preferences. In the experiment, the authors designed one asset treatment and three subtreatments within the treatment group. There were 2,183 participants assigned to the treatment group and 521 participants assigned to the control group. The asset treatment required participants to make investment decisions over a period of 6 consecutive weeks. The outcome variable of interest is participants' Socioeconomic values (SEV) measured by a principal component analysis on four items pertaining to issues of personal responsibility, economic fairness, inequality, and redistribution. As mentioned by \citet{margalit2021how}, the attrition rates are different for different treatments, although they did not find evidence for selective attrition and systematic differences between attrited and non-attrited participants.

    Table 2 of \citet[484]{margalit2021how} shows the overall treatment effect on SEV, without distinguishing between subtreatments. This reanalysis replicates column (4), which includes pretreatment SEV, political controls, and demographic controls. We consider three methods with 500 MC simulations: conformal inference with semiparametric efficient estimator (CISE), multiple imputation (MI), and weighted CQR (WCQR) with unweighted nested approach for interval estimates. After constructing the prediction intervals for ITE in the attrition group, we aggregate these ITEs to the ATE for the attrition group and compute the standard error across 500 MC simulations. Then we use a weighted average of ATE for the observed group and the attrition group to compute the ATE for all observations. The standard error of ATE for all observations is computed as 
    \begin{align*}
        \text{SE}\left( \text{ATE}_{\text{all}} \right) = \sqrt{\left( \frac{|\mathcal{I}_{R = 1}| }{|\mathcal{I}_{R = 1}| + |\mathcal{I}_{R = 0}|} \right)^2 \text{SE}\left( \text{ATE}_{R = 1} \right)^2 + \left( \frac{|\mathcal{I}_{R = 0}| }{|\mathcal{I}_{R = 1}| + |\mathcal{I}_{R = 0}|} \right)^2  \text{SE}\left( \text{ATE}_{R = 0} \right)^2},
    \end{align*}
    where $\text{SE}\left( \text{ATE}_{R = 1} \right)$ is the standard error of the ATE estimate from the regression model, which is the same as reported by \citet{margalit2021how}, $\text{SE}\left( \text{ATE}_{R = 0} \right)$ is the standard error of the ATE estimate for the attrition group across 500 MC simulations. We also use IPW to estimate the ATE for the observed group as a comparison to the result reported by \citet{margalit2021how}.
    
    We also evaluate the average length of the prediction intervals for ATE in the attrition group. Since we do not have control over the DGP with the real data, we know nothing about the ground truth to evaluate the coverage and the average length as in the simulation studies. Therefore, it is only meaningful to evaluate the average interval length here.
    %Instead, we compute the empirical coverage and average length of the prediction intervals based on the observed %data. The empirical coverage is computed as 
    %\begin{align*}
    %    \frac{1}{|\mathcal{I}_{R = 0}|} \sum_{R = 0} \mathds{1}_{\left\{ \text{CI}_{R = 1} \subset \check{\mathcal{C}}_%{i, \text{ITE}_{R = 0}} \right\}},
    %\end{align*}
    %where $\text{CI}_{R = 1}$ is the confidence interval for ATE in the observed group. 
    The average length of the prediction intervals is computed as
    \begin{align*}
        \frac{1}{|\mathcal{I}_{R = 0}|} \sum_{R = 0} \left( \check{\mathcal{C}}^{\text{R}}_{i, \text{ITE}_{R = 0}} - \check{\mathcal{C}}^{\text{L}}_{i, \text{ITE}_{R = 0}} \right),
    \end{align*}
    where $\check{\mathcal{C}}^{\text{R}}_{i, \text{ITE}_{R = 0}}$ and $\check{\mathcal{C}}^{\text{L}}_{i, \text{ITE}_{R = 0}}$ are the upper and lower bounds of the prediction interval for ITE in the attrition group. Similarly, we compute the standard error for the empirical coverage and average length of the prediction intervals for ATE in the attrition group across 500 MC simulations.
    
    Table \ref{tab:rep1} presents the reanalysis results for \citet{margalit2021how} Table 2 column (4). Among 2,703 participants, 480 participants are in the attrition group, which is 17.8\% of the total sample size. As column (4) shows, contrary to the common belief that prediction intervals constructed by conformal inference are much wider than any alternative methods, the average length of the prediction intervals by CISE and MI is comparable, while the average length of the prediction intervals by WCQR is much wider. Although in the simulation studies in the previous section, we show that CISE has the best performance in terms of empirical coverage and average length, we are considering the extreme DGP which violates the normality assumption. However, in this empirical application with real data, MI performs better than simulation studies. This may be because the real data is not as extreme as the simulation DGP, and the normality assumption is more likely to hold. Another note is that ATE estimate from IPW is almost the same as the result reported by \citet{margalit2021how}. However, such an estimate is only for the observed group, which may lead to biased estimates of the overall ATE when there is attrition. Figure \ref{fig:repcomplen1} shows the simulation results for the average length of prediction intervals constructed by three methods.

    \begin{table}[t]
        \caption{Reanalysis of \citet{margalit2021how} Table 2 Column (4)}
        \label{tab:rep1}
        \centering
        \begin{threeparttable}
            \begin{tabular}[t]{ccccc}
            \hline\hline
            & (1) & (2) & (3) & (4) \\
            Method & ATER1 & ATER0 & ATEall & Length\\
            \hline
            
            \hline
            CISE & 0.098 & 0.190 & 0.114 & 6.894\\
            & (0.052) & (0.180) & (0.054) & (0.860)\\
            \hline
            MI & 0.098 & 0.129 & 0.103 & 5.668\\
            & (0.052) & (0.040) & (0.044) & (0.079)\\
            \hline
            WCQR & 0.098 & 0.205 & 0.117 & 10.495\\
            & (0.052) & (0.109) & (0.047) & (0.830)\\
            \hline
            IPW & 0.110 &  &  & \\
            & (0.053) &  &  & \\
            \hline
            Observations & 2,223 & 480 & 2,703 &  \\
            \hline\hline
            \end{tabular}
            \begin{tablenotes}
                \footnotesize
                \setlength\labelsep{0pt}
                \item \textit{Note}: This table presents the replication results for \citet{margalit2021how} Table 2 column (4). Column (1) - (3) report point and interval estimates of ATE for observed group (ATER1), attrition group (ATER0), and all observations (ATEall) using conformal inference with semiparametric efficient estimator (CISE), multiple imputation (MI), weighted CQR (WCQR) with unweighted nested approach for interval estimates, and IPW. Column (4) reports the average length of prediction intervals of ATE for the attrition group. The standard errors are reported in parentheses.
            \end{tablenotes}
        \end{threeparttable}
    \end{table}

    \begin{figure}[ht]
        \centering
        \caption{Comparison of Average Length of Prediction Intervals}
        \label{fig:repcomplen1}
        \begin{threeparttable}
            \includegraphics[width=\textwidth]{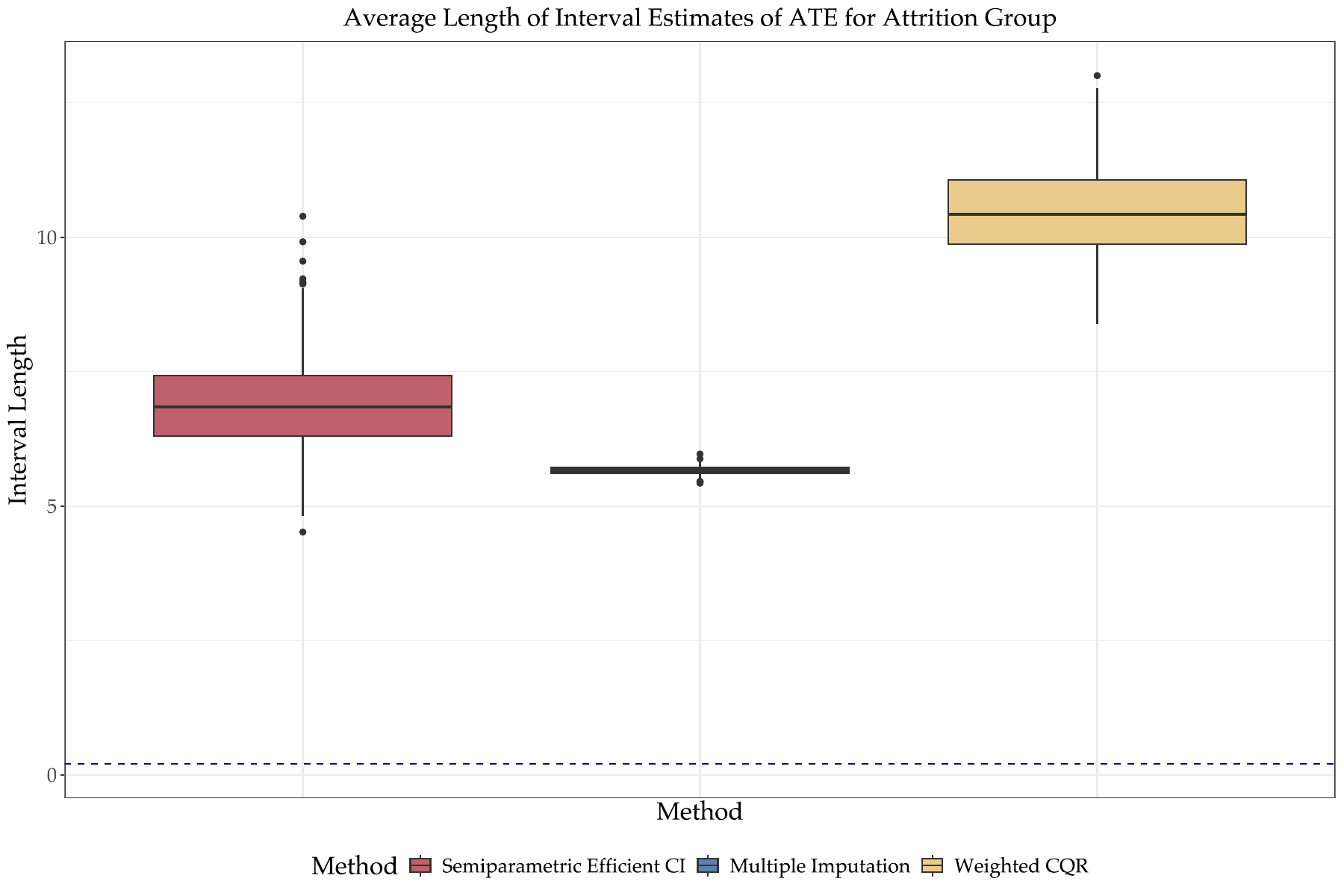}
            \begin{tablenotes}
                \footnotesize
                \setlength\labelsep{0pt}
                \item \textit{Note}: This figure shows the simulation results for the average length of prediction intervals constructed by conformal inference with semiparametric efficient estimator, multiple imputation with Amelia, and weighted CQR with unweighted nested approach for ITE of attrition group. The blue horizontal line corresponds to the length of oracle intervals.
            \end{tablenotes}
        \end{threeparttable}
    \end{figure}

    ATE estimates are different between groups. For the attrition group, as column (2) shows, the point estimate by CISE is 0.19, which is almost twice the ATE of the observed group. However, the point estimate of ATE for all observations slightly increases to 0.114 and remains significant at 5\% level, which does not change the substantive conclusion when only using the observed group. This may because that although the attrition group has a larger ATE, it only accounts for 17.8\% of the total sample size and the point estimate is not significant at the conventional level.

    \subsection{Reanalysis of \citet{finkel2024can}}

    This paper studies how the online civic education induces democratic values and behaviors in new democracies by designing and testing three civic education interventions and using Tunisia as a case study. The authors recruited young Tunisians between 18 and 35 through paid advertisements via the online platforms Facebook and Instagram. In total, 5,069 participants started a baseline survey and 2,073 finished the first and last question of the survey, which leads to a completion rate of 45.5\%. The authors designed three different treatments through exposing participants to different videos emphasizing democracy frames: democracy gain frame, democracy loss frame, and democratic self-efficacy frame, as well as a control group. The outcome variables include evaluations of political regimes and political engagement, broadly. In our empirical application, we focus on the results from Table 2 of \citet[623]{finkel2024can}, which compares all treatment groups versus the control group on regime evaluations. The authors considered four measures related to evaluations: democratic regime rating (M1), Ben Ali regime rating (M2), non-democratic regime alternatives (M3), and regime democratic performance (M4).

    Table 2 of \citet{finkel2024can} presents estimates of all treatment groups pooled together versus the control group. This reanalysis replicates column (2) - (5), which considers four different measures of regime evaluations with pre-treatment controls: gender, age, education, employment status, prior registration status, prior support for democracy, interest in political matters, and animal-related matters. Similar to the first application, we consider three methods with 500 MC simulations and construct the prediction intervals for ITE in the attrition group. We also aggregate these ITEs to the ATE for the attrition group and compute the standard errors. Similarly, we compute ATE estimates by IPW as a comparison to the estimates reported in their paper. The evaluation of the average length of the prediction intervals for ATE in the attrition group and their standard errors follows the same logic as in the first application. 

    Table \ref{tab:rep2} presents the reanalysis results for \citet{finkel2024can} Table 2 column (2) - (5). The attrition rates are different for different measures of regime evaluations, which is 36.4\% for M1, 40.8\% for M2, 40.7\% for M3, and 36.8\% for M4. As column (4) shows, the average length of the prediction intervals by CISE is shorter than that by MI for all four outcomes in this practice. The average length of the prediction intervals by WCQR is almost twice the length of that by CISE, which is consistent with the simulation results in the previous section. IPW still reports similar ATE estimates for the observed group, but fails to provide estimates for the attrition group and all observations. Figure \ref{fig:repcomplen2} shows the simulation results for the average length of prediction intervals constructed by three methods. Overall, CISE has the best performance both in terms of empirical coverage and average length.

    \begin{table}[htpb]
        \caption{Reanalysis of \citet{finkel2024can} Table 2 Column (2) - (5)}
        \label{tab:rep2}
        \centering
        \renewcommand{\arraystretch}{0.95}
        \begin{tabular}[t]{ccccc}
        \hline\hline
        \multicolumn{5}{c}{M1 Democratic Regime Rating}\\
        \hline 

        \hline
        & (1) & (2) & (3) & (4) \\
        Method & ATER1 & ATER0 & ATEall & Length\\
        \hline

        \hline
        CISE & 0.014 & 0.012 & 0.013 & 1.030\\
        & (0.015) & (0.021) & (0.012) & (0.333)\\
        \hline
        MI & 0.014 & 0.005 & 0.010 & 1.469\\
        & (0.015) & (0.007) & (0.009) & (0.013)\\
        \hline
        WCQR & 0.014 & 0.018 & 0.016 & 2.051\\
        & (0.015) & (0.006) & (0.009) & (0.040)\\
        \hline
        IPW & 0.009 &  &  & \\
        & (0.015) &  &  & \\
        \hline
        Observations & 2,190 & 1,254 & 3,444 & \\
        \hline\hline
        \multicolumn{5}{c}{M2 Ben Ali Regime Rating}\\
        \hline

        \hline
        & (1) & (2) & (3) & (4) \\
        Method & ATER1 & ATER0 & ATEall & Length\\
        \hline

        \hline
        CISE & -0.041 & -0.059 & -0.048 & 1.241\\
        & (0.016) & (0.021) & (0.013) & (0.281)\\
        \hline
        MI & -0.041 & -0.035 & -0.039 & 1.566\\
        & (0.016) & (0.008) & (0.010) & (0.013)\\
        \hline
        WCQR & -0.041 & -0.069 & -0.053 & 2.143\\
        & (0.016) & (0.011) & (0.011) & (0.039)\\
        \hline
        IPW & -0.037 &  &  & \\
        & (0.017) &  &  & \\
        \hline
        Observations & 2,197 & 1,517 & 3,714 &   \\
        \hline\hline
        \multicolumn{5}{c}{M3 Non-democratic Regime Alternatives}\\
        \hline

        \hline
        & (1) & (2) & (3) & (4) \\
        Method & ATER1 & ATER0 & ATEall & Length\\
        \hline

        \hline
        CISE & -0.004 & -0.078 & -0.034 & 0.963\\
        & (0.011) & (0.033) & (0.015) & (0.206)\\
        \hline
        MI & -0.004 & 0.003 & -0.001 & 1.076\\
        & (0.011) & (0.005) & (0.007) & (0.011)\\
        \hline
        WCQR & -0.004 & -0.064 & -0.028 & 2.249\\
        & (0.011) & (0.016) & (0.009) & (0.200)\\
        \hline
        IPW & -0.004 &  &  & \\
        & (0.011) &  &  & \\
        \hline
        Observations & 2,203 & 1,511 & 3,714 &   \\
        \hline\hline
        \multicolumn{5}{c}{M4 Regime Democratic Performance}\\
        \hline

        \hline
        & (1) & (2) & (3) & (4) \\
        Method & ATER1 & ATER0 & ATEall  & Length\\
        \hline

        \hline
        CISE & 0.040 & 0.048 & 0.043 & 0.819\\
        & (0.012) & (0.017) & (0.010) & (0.219)\\
        \hline
        MI & 0.040 & 0.039 & 0.040 & 1.200\\
        & (0.012) & (0.006) & (0.008) & (0.012)\\
        \hline
        WCQR & 0.040 & 0.024 & 0.034 & 2.208\\
        & (0.012) & (0.018) & (0.010) & (0.161)\\
        \hline
        IPW & 0.041 &  &  & \\
        & (0.013) &  &  & \\
        \hline
        Observations & 2,346 & 1,368 & 3,714 &   \\
        \hline\hline
        \end{tabular}
    \end{table}

    \begin{figure}[ht]
        \centering
        \caption{Comparison of Average Length of Prediction Intervals}
        \label{fig:repcomplen2}
        \begin{threeparttable}
            \includegraphics[width=\textwidth]{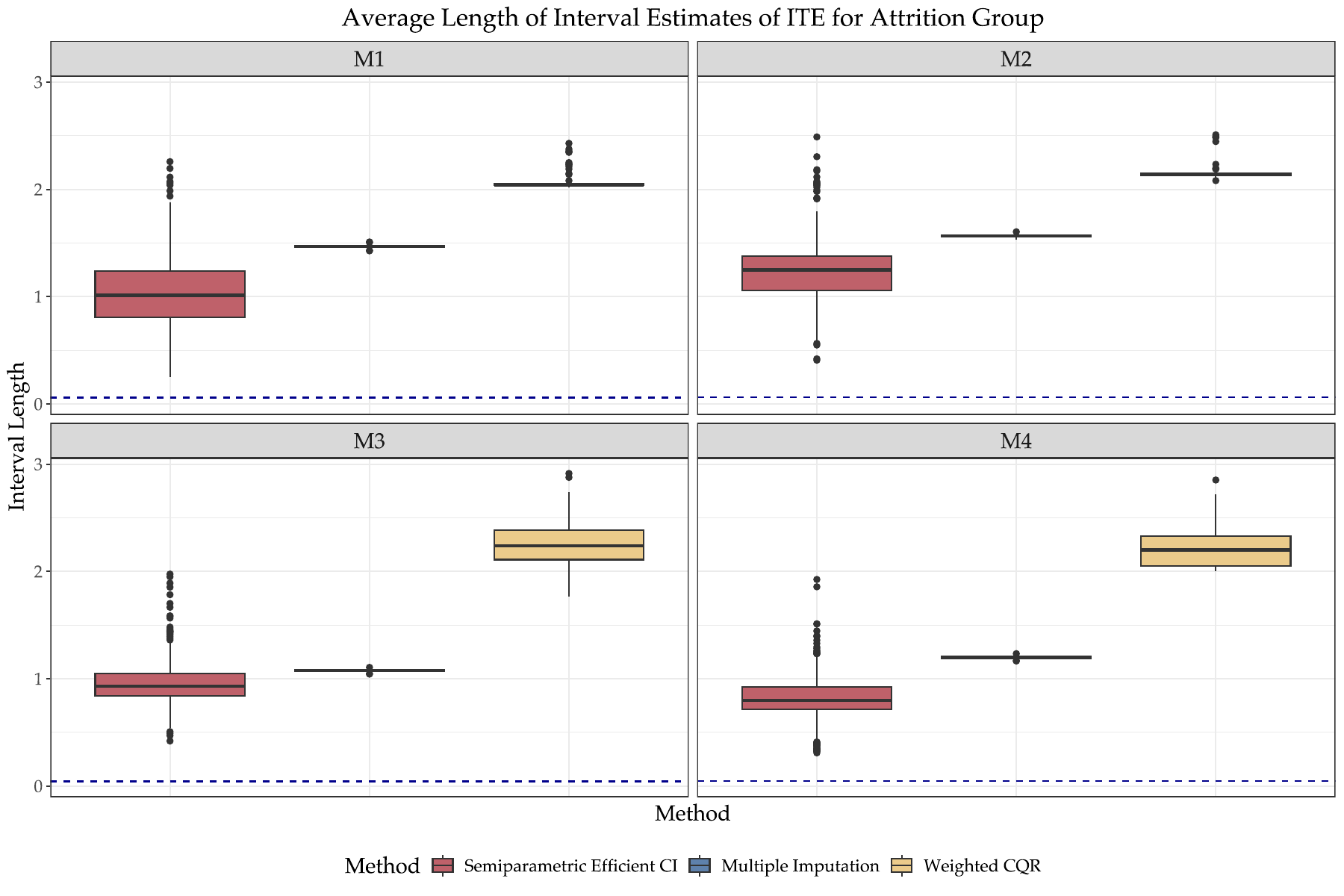}
            \begin{tablenotes}
                \footnotesize
                \setlength\labelsep{0pt}
                \item \textit{Note}: This figure shows the simulation results for the average length of prediction intervals constructed by conformal inference with semiparametric efficient estimator, multiple imputation with Amelia, and weighted CQR with unweighted nested approach for ITE of attrition group. The blue horizontal line corresponds to the length of oracle intervals.
            \end{tablenotes}
        \end{threeparttable}
    \end{figure}

    For ATE estimates, the results for M1 and M4 remain unchanged in both substantive conclusions and significance levels. In contrast, for M2, the ATE estimate for all observations becomes significant at the 1\% level, while the ATE for the observed group is only significant at the 5\% level. For M3, the ATE for all observations reaches significance at the 5\% level, changing the original substantive conclusion reported by the authors. This shift likely occurs because the ATE for the attrition group is not only statistically significant at conventional levels but also substantially different from that of the observed group.

    To understand the heterogeneous ATE estimates for the observed group and the attrition group for M3, we take a step back to analyze the ITEs instead of the ATE. Figure \ref{fig:obsite_attrite} shows the density of ITE estimates for the observed and attrition group across 500 MC simulations. As Panel (B) shows, the ITE estimates for the attrition groups indicate the presence of heterogeneous effects, resulting in a different aggregated ATE. To discover the potential source, we interact the pretreatment covariates with the treatment indicator. Table \ref{tab:balance_finkel} shows that those who have expressed a lot of interest in politics experience a larger treatment effect. However, as the regression results in Table \ref{tab:hte} show, we fail to discover a negative interaction term as expected. This may suggest that there are other unobserved confounders leading to the heterogeneous ATE between the observed and the attrition group.
  
    \begin{figure}[htbp]
        \centering
        \caption{ITE Estimates for Observed and Attrition Groups}
        \label{fig:obsite_attrite}
        \begin{threeparttable}
            \begin{subfigure}[t]{0.48\textwidth}
                \centering
                \includegraphics[width=\textwidth]{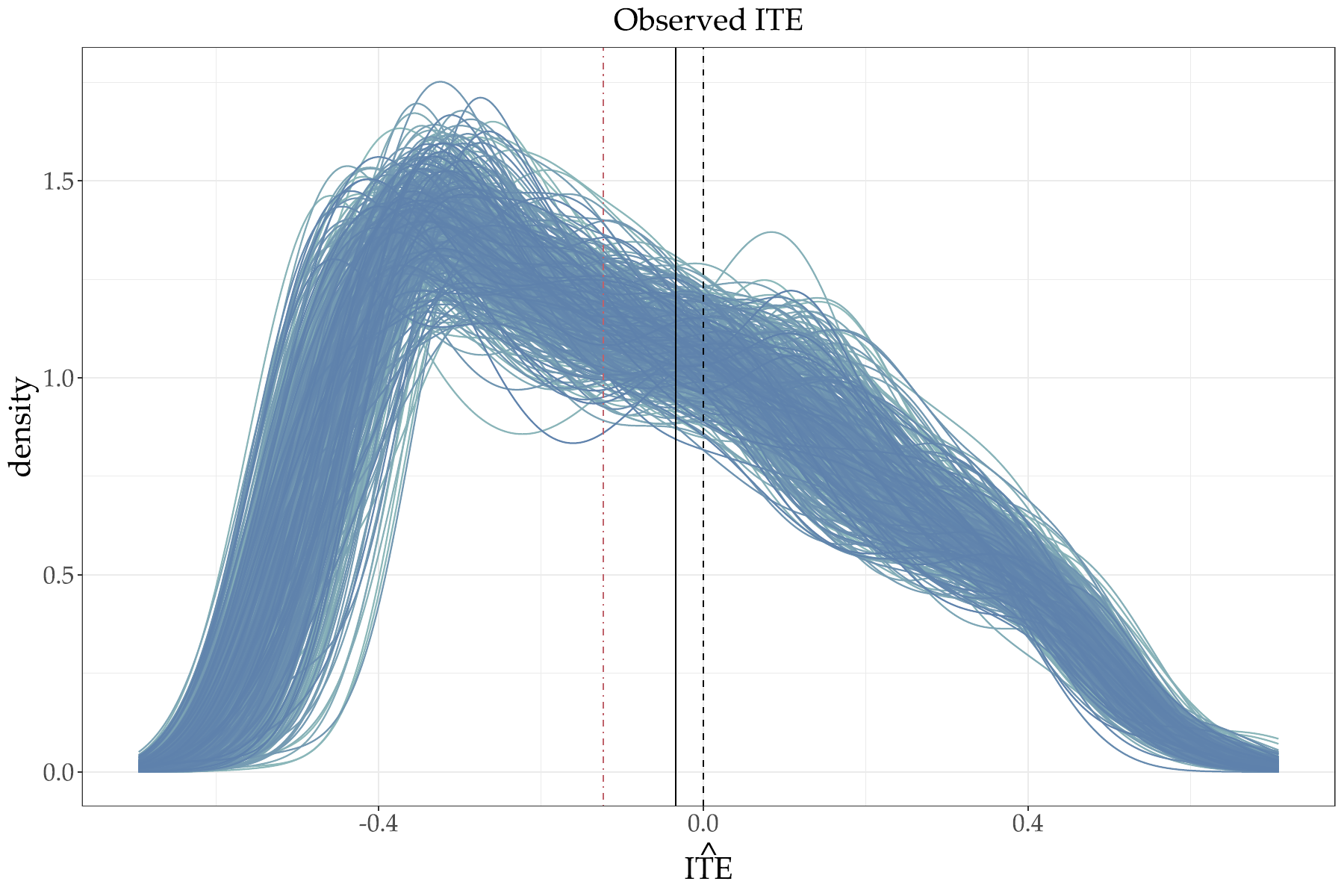}
                \caption{Observed Group}
            \end{subfigure}
            \hfill
            \begin{subfigure}[t]{0.48\textwidth}
                \centering
                \includegraphics[width=\textwidth]{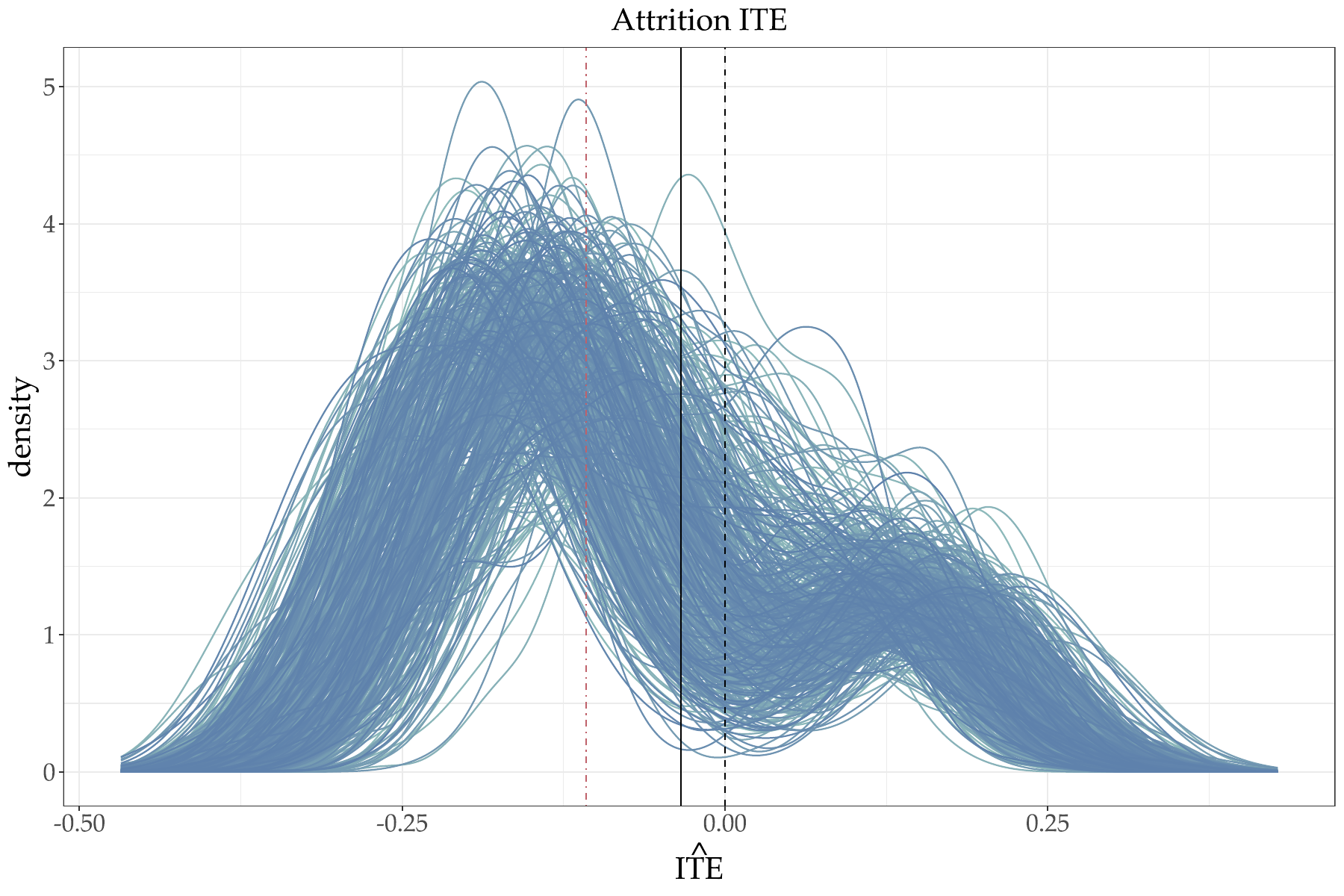}
                \caption{Attrition Group}
            \end{subfigure}
            \begin{tablenotes}
                \footnotesize
                \setlength\labelsep{0pt}
                \item \textit{Note}: These two plots exhibit the density of ITE estimates for the observed and attrition group across 500 MC simulations. The red dotted line denotes the median of the ITE estimates, the black solid line denotes the ATE estimate for the attrition group (-0.034), and the black dotted line denotes the 0.
            \end{tablenotes}
        \end{threeparttable}
    \end{figure}

    Overall, these two empirical applications illustrate that the proposed conformal inference with semiparametric efficient estimator can construct prediction intervals for ITE in the attrition group with desired coverage and narrower width, and can easily be aggregated to the ATE of interest for the attrition group. The point and interval estimates of ATE for all observations can be computed as well. The application results also serve as a cautionary note that simply ignoring the attrition group and only using observed group may lead to biased estimates and different substantive conclusions, especially when the attrition rate is high. However, there are also limitations of the proposed method. First, given that the DGP of the real data may satisfy the normality assumption, conformal inference with semiparametric efficient estimator may not always outperform the multiple imputation in interval length. Second, the attrition rate can largely influence both the point and interval estimates of ATE of the attrition group and all observations. When the attrition rate is really high, everything may become noisier for inference. Third, the proposed method can only provide the prediction intervals of ITE for part of the observed group because some data needs to be excluded for the model training and evaluation purpose. Therefore, it is hard to evaluate the coverage and average length of the prediction intervals of ITE or ATE for all observations.

    \section{Concluding Remarks}

    Attrition in survey and field experiments is a common problem in political science research. Simply ignoring the missing data induced by attrition and performing complete case analysis may lead to biased and inconsistent estimates, potentially resulting in misleading substantive conclusions. While existing methods such as multiple imputation, inverse probability weighting, and partial identification offer tools to address attrition, they often rely on additional assumptions and typically do not account for the covariate shift between observed and missing data. Employing the conformal inference framework, I identify three primary challenges when the outcomes are missing at random (MAR): (1) the covariate shift problems between treatment groups in observed data, as well as between the observed and the attrition group; (2) prediction intervals that are both valid and sufficiently narrow to be practically useful; and (3) the extension of conformal inference to ensure asymptotic properties. This paper addresses these challenges by employing a nonparametric conformal inference approach augmented with a semiparametric efficient estimator to construct valid prediction intervals for causal estimands in the presence of attrition.

    Through extensive Monte Carlo simulation studies and empirical applications, I demonstrate that the proposed method can construct prediction intervals for the estimands of interest in the attrition group with both valid coverage and narrower widths. It also outperforms the existing methods like multiple imputation with Amelia II and weighted CQR in terms of empirical coverage and average length. In addition, researchers can easily aggregate the individual-level estimates to the group level and compare the treatment effects among observed, attrited, and all observations. Nevertheless, the approach has several limitations. First, while the method does not require the assumption of ignorability of missingness, it still requires the MAR assumption. Second, when the attrition rate is high, the point and interval estimates could become noisy. Third, due to the data-splitting procedure required for model training and evaluation, it is challenging to construct prediction intervals for all the units in the observed group.

    This paper also suggests several avenues for future research. First, extending the conformal inference framework to the missing not at random (MNAR) setting would enhance its robustness and applicability, particularly in cases where attrition is correlated with unobserved confounders \citep{huber2012identification,jin2023sensitivity}. Second, integrating conformal inference with double machine learning could yield a more flexible and robust framework for estimating treatment effects, especially in the presence of high-dimensional nuisance parameters and the need for cross-fitting procedures \citep{chernozhukov2018double,kennedy2023semiparametric}. Finally, adapting conformal inference methods to panel data structures would expand their utility for analyzing dynamic processes and handling attrition in longitudinal studies \citep{chernozhukov2021exact}.

    \clearpage
    \bibliographystyle{apsr}
    %\begin{spacing}{1.4}
        \bibliography{ref}

@article{angelopoulos2022gentle,
  title   = {A gentle introduction to conformal prediction and distribution-free uncertainty quantification},
  author  = {Angelopoulos, Anastasios N and Bates, Stephen},
  journal = {arXiv preprint arXiv:2107.07511},
  year    = {2022}
}

@article{angelopoulos2024theoretical,
  title   = {Theoretical foundations of conformal prediction},
  author  = {Angelopoulos, Anastasios N and Barber, Rina Foygel and Bates, Stephen},
  journal = {arXiv preprint arXiv:2411.11824},
  year    = {2024}
}

@article{athey2020combining,
  title         = {Combining {{Experimental}} and {{Observational Data}} to {{Estimate Treatment Effects}} on {{Long Term Outcomes}}},
  author        = {Athey, Susan and Chetty, Raj and Imbens, Guido},
  year          = {2020},
  journal       = {arXiv preprint arXiv:2006.09676},
  eprint        = {2006.09676},
  primaryclass  = {stat},
  publisher     = {arXiv},
  doi           = {10.48550/arXiv.2006.09676},
  archiveprefix = {arXiv}
}

@article{chernozhukov2018double,
  title      = {Double/Debiased Machine Learning for Treatment and Structural Parameters},
  author     = {Chernozhukov, Victor and Chetverikov, Denis and Demirer, Mert and Duflo, Esther and Hansen, Christian and Newey, Whitney and Robins, James},
  year       = {2018},
  journal    = {The Econometrics Journal},
  volume     = {21},
  number     = {1},
  pages      = {C1-C68},
  doi        = {10.1111/ectj.12097}
}

@article{chernozhukov2021exact,
  title      = {An {{Exact}} and {{Robust Conformal Inference Method}} for {{Counterfactual}} and {{Synthetic Controls}}},
  author     = {Chernozhukov, Victor and W{\"u}thrich, Kaspar and Zhu, Yinchu},
  year       = {2021},
  journal    = {Journal of the American Statistical Association},
  volume     = {116},
  number     = {536},
  pages      = {1849--1864},
  publisher  = {ASA Website},
  doi        = {10.1080/01621459.2021.1920957}
}

@article{egami2023elements,
  title      = {Elements of {{External Validity}}: {{Framework}}, {{Design}}, and {{Analysis}}},
  shorttitle = {Elements of {{External Validity}}},
  author     = {Egami, Naoki and Hartman, Erin},
  year       = {2023},
  journal    = {American Political Science Review},
  volume     = {117},
  number     = {3},
  pages      = {1070--1088},
  doi        = {10.1017/S0003055422000880}
}

@book{fisher1937the,
  author    = {Fisher,R. A.},
  title     = {The Design of Experiments.},
  year      = {1937},
  pages     = {xi + 260 pp.},
  publisher = {Oliver \& Boyd, Edinburgh \& London},
  language  = {English},
  item-type = {Book},
  issue     = {2nd Ed}
}

@article{gao2025role,
  title         = {On the {{Role}} of {{Surrogates}} in {{Conformal Inference}} of {{Individual Causal Effects}}},
  author        = {Gao, Chenyin and Gilbert, Peter B. and Han, Larry},
  year          = {2025},
  journal       = {arXiv preprint arXiv:2412.12365},
  eprint        = {2412.12365},
  primaryclass  = {stat},
  publisher     = {arXiv},
  doi           = {10.48550/arXiv.2412.12365},
  archiveprefix = {arXiv}
}

@article{gohdes2020repression,
  title      = {Repression {{Technology}}: {{Internet Accessibility}} and {{State Violence}}},
  shorttitle = {Repression {{Technology}}},
  author     = {Gohdes, Anita R.},
  year       = {2020},
  journal    = {American Journal of Political Science},
  volume     = {64},
  number     = {3},
  pages      = {488--503},
  doi        = {10.1111/ajps.12509}
}

@article{huber2012identification,
  title      = {Identification of {{Average Treatment Effects}} in {{Social Experiments Under Alternative Forms}} of {{Attrition}}},
  author     = {Huber, Martin},
  year       = {2012},
  journal    = {Journal of Educational and Behavioral Statistics},
  volume     = {37},
  number     = {3},
  pages      = {443--474},
  publisher  = {American Educational Research Association},
  doi        = {10.3102/1076998611411917}
}

@article{imbens2024causal,
  title     = {Causal {{Inference}} in the {{Social Sciences}}},
  author    = {Imbens, Guido W.},
  year      = {2024},
  journal   = {Annual Review of Statistics and Its Application},
  volume    = {11},
  number    = {Volume 11, 2024},
  pages     = {123--152},
  publisher = {Annual Reviews},
  doi       = {10.1146/annurev-statistics-033121-114601}
}

@article{koenker1978regression,
  title      = {Regression {{Quantiles}}},
  author     = {Koenker, Roger and Bassett, Gilbert},
  year       = {1978},
  journal    = {Econometrica},
  volume     = {46},
  number     = {1},
  eprint     = {1913643},
  eprinttype = {jstor},
  pages      = {33--50},
  publisher  = {[Wiley, Econometric Society]},
  doi        = {10.2307/1913643}
}

@article{lei2013distributionfree,
  title      = {Distribution-{{Free Prediction Sets}}},
  author     = {Lei, Jing and Robins, James and Wasserman, Larry},
  year       = {2013},
  journal    = {Journal of the American Statistical Association},
  volume     = {108},
  number     = {501},
  pages      = {278--287},
  publisher  = {ASA Website},
  doi        = {10.1080/01621459.2012.751873},
  pmid       = {25237208}
}

@article{lei2014distributionfree,
  title      = {Distribution-Free {{Prediction Bands}} for {{Non-parametric Regression}}},
  author     = {Lei, Jing and Wasserman, Larry},
  year       = {2014},
  journal    = {Journal of the Royal Statistical Society Series B: Statistical Methodology},
  volume     = {76},
  number     = {1},
  pages      = {71--96},
  doi        = {10.1111/rssb.12021}
}

@article{lei2018distributionfree,
  title      = {Distribution-{{Free Predictive Inference}} for {{Regression}}},
  author     = {Lei, Jing and G'Sell, Max and Rinaldo, Alessandro and Tibshirani, Ryan J. and Wasserman, Larry},
  year       = {2018},
  journal    = {Journal of the American Statistical Association},
  volume     = {113},
  number     = {523},
  pages      = {1094--1111},
  publisher  = {ASA Website},
  doi        = {10.1080/01621459.2017.1307116}
}

@article{lei2021conformala,
  title      = {Conformal {{Inference}} of {{Counterfactuals}} and {{Individual Treatment Effects}}},
  author     = {Lei, Lihua and Cand{\`e}s, Emmanuel J.},
  year       = {2021},
  journal    = {Journal of the Royal Statistical Society Series B: Statistical Methodology},
  volume     = {83},
  number     = {5},
  pages      = {911--938},
  doi        = {10.1111/rssb.12445}
}

@article{mueller2024crowd,
  title   = {Crowd {{Cohesion}} and {{Protest Outcomes}}},
  author  = {Mueller, Lisa},
  year    = {2024},
  journal = {American Journal of Political Science},
  volume  = {68},
  number  = {1},
  pages   = {42--57},
  doi     = {10.1111/ajps.12725}
}

@inproceedings{romano2019conformalized,
  title     = {Conformalized {{Quantile Regression}}},
  booktitle = {Advances in {{Neural Information Processing Systems}}},
  author    = {Romano, Yaniv and Patterson, Evan and Candes, Emmanuel},
  year      = {2019},
  volume    = {32},
  publisher = {Curran Associates, Inc.}
}

@inproceedings{romano2020classification,
  title     = {Classification with Valid and Adaptive Coverage},
  booktitle = {Proceedings of the 34th {{International Conference}} on {{Neural Information Processing Systems}}},
  author    = {Romano, Yaniv and Sesia, Matteo and Cand{\`e}s, Emmanuel J.},
  year      = {2020},
  series    = {{{NIPS}} '20},
  pages     = {3581--3591},
  publisher = {Curran Associates Inc.},
  address   = {Red Hook, NY, USA},
  isbn      = {978-1-71382-954-6}
}

@article{rubin1974estimating,
  title     = {Estimating Causal Effects of Treatments in Randomized and Nonrandomized Studies},
  author    = {Rubin, Donald B.},
  year      = {1974},
  journal   = {Journal of Educational Psychology},
  volume    = {66},
  number    = {5},
  pages     = {688--701},
  publisher = {American Psychological Association},
  address   = {US},
  doi       = {10.1037/h0037350}
}

@article{sadinle2019least,
  title      = {Least {{Ambiguous Set-Valued Classifiers With Bounded Error Levels}}},
  author     = {Sadinle, Mauricio and Lei, Jing and Wasserman, Larry},
  year       = {2019},
  journal    = {Journal of the American Statistical Association},
  volume     = {114},
  number     = {525},
  pages      = {223--234},
  publisher  = {ASA Website},
  doi        = {10.1080/01621459.2017.1395341}
}

@article{splawa-neyman1990application,
  title      = {On the {{Application}} of {{Probability Theory}} to {{Agricultural Experiments}}. {{Essay}} on {{Principles}}. {{Section}} 9, transl. by D. M. Dabrowska and T. P. Speed},
  author     = {{Splawa-Neyman}, Jerzy},
  year       = {1990(1923)},
  journal    = {Statistical Science},
  volume     = {5},
  number     = {4},
  eprint     = {2245382},
  eprinttype = {jstor},
  pages      = {465--472},
  publisher  = {Institute of Mathematical Statistics}
}

@inproceedings{tibshirani2019conformal,
  title     = {Conformal {{Prediction Under Covariate Shift}}},
  booktitle = {Advances in {{Neural Information Processing Systems}}},
  author    = {Tibshirani, Ryan J and Foygel Barber, Rina and Candes, Emmanuel and Ramdas, Aaditya},
  year      = {2019},
  volume    = {32},
  publisher = {Curran Associates, Inc.}
}

@inproceedings{vovk1999machinelearning,
  title     = {Machine-{{Learning Applications}} of {{Algorithmic Randomness}}},
  booktitle = {Proceedings of the {{Sixteenth International Conference}} on {{Machine Learning}}},
  author    = {Vovk, Volodya and Gammerman, Alexander and Saunders, Craig},
  year      = {1999},
  series    = {{{ICML}} '99},
  pages     = {444--453},
  publisher = {Morgan Kaufmann Publishers Inc.},
  address   = {San Francisco, CA, USA},
  isbn      = {978-1-55860-612-8}
}

@article{vovk2009online,
  title      = {On-Line Predictive Linear Regression},
  author     = {Vovk, Vladimir and Nouretdinov, Ilia and Gammerman, Alex},
  year       = {2009},
  journal    = {The Annals of Statistics},
  volume     = {37},
  number     = {3},
  pages      = {1566--1590},
  publisher  = {Institute of Mathematical Statistics},
  doi        = {10.1214/08-AOS622}
}

@article{vovk2019nonparametric,
  title      = {Nonparametric Predictive Distributions Based on Conformal Prediction},
  author     = {Vovk, Vladimir and Shen, Jieli and Manokhin, Valery and Xie, Min-ge},
  year       = {2019},
  journal    = {Machine Learning},
  volume     = {108},
  number     = {3},
  pages      = {445--474},
  doi        = {10.1007/s10994-018-5755-8}
}

@article{lalonde1986evaluating,
  title = {Evaluating the {{Econometric Evaluations}} of {{Training Programs}} with {{Experimental Data}}},
  author = {LaLonde, Robert J.},
  year = {1986},
  journal = {The American Economic Review},
  volume = {76},
  number = {4},
  eprint = {1806062},
  eprinttype = {jstor},
  pages = {604--620},
  publisher = {American Economic Association}
}

@article{kalla2018minimal,
  title = {The {{Minimal Persuasive Effects}} of {{Campaign Contact}} in {{General Elections}}: {{Evidence}} from 49 {{Field Experiments}}},
  shorttitle = {The {{Minimal Persuasive Effects}} of {{Campaign Contact}} in {{General Elections}}},
  author = {Kalla, Joshua L. and Broockman, David E.},
  year = {2018},
  journal = {American Political Science Review},
  volume = {112},
  number = {1},
  pages = {148--166},
  doi = {10.1017/S0003055417000363}
}

@article{imai2011estimation,
  title = {Estimation of {{Heterogeneous Treatment Effects}} from {{Randomized Experiments}}, with {{Application}} to the {{Optimal Planning}} of the {{Get-Out-the-Vote Campaign}}},
  author = {Imai, Kosuke and Strauss, Aaron},
  year = {2011},
  journal = {Political Analysis},
  volume = {19},
  number = {1},
  pages = {1--19},
  doi = {10.1093/pan/mpq035}
}

@article{hausman1979attrition,
  title = {Attrition {{Bias}} in {{Experimental}} and {{Panel Data}}: {{The Gary Income Maintenance Experiment}}},
  shorttitle = {Attrition {{Bias}} in {{Experimental}} and {{Panel Data}}},
  author = {Hausman, Jerry A. and Wise, David A.},
  year = {1979},
  journal = {Econometrica},
  volume = {47},
  number = {2},
  eprint = {1914193},
  eprinttype = {jstor},
  pages = {455--473},
  publisher = {[Wiley, Econometric Society]},
  doi = {10.2307/1914193}
}

@article{coppock2017combining,
  title = {Combining {{Double Sampling}} and {{Bounds}} to {{Address Nonignorable Missing Outcomes}} in {{Randomized Experiments}}},
  author = {Coppock, Alexander and Gerber, Alan S. and Green, Donald P. and Kern, Holger L.},
  year = {2017},
  journal = {Political Analysis},
  volume = {25},
  number = {2},
  pages = {188--206},
  doi = {10.1017/pan.2016.6}
}

@article{heckman1979sample,
  title = {Sample {{Selection Bias}} as a {{Specification Error}}},
  author = {Heckman, James J.},
  year = {1979},
  journal = {Econometrica},
  volume = {47},
  number = {1},
  eprint = {1912352},
  eprinttype = {jstor},
  pages = {153--161},
  publisher = {[Wiley, Econometric Society]},
  doi = {10.2307/1912352}
}

@article{fukumoto2022nonignorable,
  title = {Nonignorable {{Attrition}} in {{Pairwise Randomized Experiments}}},
  author = {Fukumoto, Kentaro},
  year = {2022},
  journal = {Political Analysis},
  volume = {30},
  number = {1},
  pages = {132--141},
  doi = {10.1017/pan.2020.51}
}

@article{rubin1976inference,
  title = {Inference and Missing Data},
  author = {Rubin, Donald B},
  year = 1976,
  journal = {Biometrika},
  volume = {63},
  number = {3},
  pages = {581--592},
  doi = {10.1093/biomet/63.3.581}
}

@article{king2001analyzing,
  title = {Analyzing {{Incomplete Political Science Data}}: {{An Alternative Algorithm}} for {{Multiple Imputation}}},
  shorttitle = {Analyzing {{Incomplete Political Science Data}}},
  author = {King, Gary and Honaker, James and Joseph, Anne and Scheve, Kenneth},
  year = {2001},
  journal = {American Political Science Review},
  volume = {95},
  number = {1},
  pages = {49--69},
  doi = {10.1017/S0003055401000235}
}

@article{shin2024differenceindifferences,
  title = {Difference-in-Differences {{Design}} with {{Outcomes Missing Not}} at {{Random}}},
  author = {Shin, Sooahn},
  year = {2024},
  journal = {arXiv preprint arXiv:2411.18772},
  eprint = {2411.18772},
  primaryclass = {stat},
  publisher = {arXiv},
  doi = {10.48550/arXiv.2411.18772},
  archiveprefix = {arXiv}
}

@article{horvitz1952generalization,
  title = {A {{Generalization}} of {{Sampling Without Replacement}} from a {{Finite Universe}}},
  author = {Horvitz, D. G. and {Thompson}, D. J.},
  year = {1952},
  journal = {Journal of the American Statistical Association},
  volume = {47},
  number = {260},
  pages = {663--685},
  publisher = {ASA Website},
  doi = {10.1080/01621459.1952.10483446}
}

@article{horowitz1998censoring,
  title = {Censoring of Outcomes and Regressors Due to Survey Nonresponse: {{Identification}} and Estimation Using Weights and Imputations},
  shorttitle = {Censoring of Outcomes and Regressors Due to Survey Nonresponse},
  author = {Horowitz, Joel L. and Manski, Charles F.},
  year = {1998},
  journal = {Journal of Econometrics},
  volume = {84},
  number = {1},
  pages = {37--58},
  doi = {10.1016/S0304-4076(97)00077-8}
}

@article{horowitz2000nonparametric,
  title = {Nonparametric {{Analysis}} of {{Randomized Experiments}} with {{Missing Covariate}} and {{Outcome Data}}},
  author = {Horowitz, Joel L. and Manski, Charles F.},
  year = {2000},
  journal = {Journal of the American Statistical Association},
  volume = {95},
  number = {449},
  eprint = {2669526},
  eprinttype = {jstor},
  pages = {77--84},
  publisher = {[American Statistical Association, Taylor \& Francis, Ltd.]},
  doi = {10.2307/2669526}
}

@article{lee2009training,
  title = {Training, {{Wages}}, and {{Sample Selection}}: {{Estimating Sharp Bounds}} on {{Treatment Effects}}},
  shorttitle = {Training, {{Wages}}, and {{Sample Selection}}},
  author = {Lee, David S.},
  year = {2009},
  journal = {The Review of Economic Studies},
  volume = {76},
  number = {3},
  pages = {1071--1102},
  doi = {10.1111/j.1467-937X.2009.00536.x}
}

@article{holland1986statistics,
  title = {Statistics and {{Causal Inference}}},
  author = {Holland, Paul W.},
  year = {1986},
  journal = {Journal of the American Statistical Association},
  volume = {81},
  number = {396},
  eprint = {2289064},
  eprinttype = {jstor},
  pages = {945--960},
  publisher = {[American Statistical Association, Taylor \& Francis, Ltd.]},
  doi = {10.2307/2289064}
}

@article{jin2023sensitivity,
  title = {Sensitivity Analysis of Individual Treatment Effects: {{A}} Robust Conformal Inference Approach},
  shorttitle = {Sensitivity Analysis of Individual Treatment Effects},
  author = {Jin, Ying and Ren, Zhimei and Cand{\`e}s, Emmanuel J.},
  year = {2023},
  journal = {Proceedings of the National Academy of Sciences},
  volume = {120},
  number = {6},
  pages = {e2214889120},
  publisher = {Proceedings of the National Academy of Sciences},
  doi = {10.1073/pnas.2214889120}
}

@book{gerber2012field,
  title={Field Experiments: Design, Analysis, and Interpretation},
  author={Gerber, A.S. and Green, D.P.},
  isbn={9780393979954},
  lccn={2011052337},
  url={https://books.google.com/books?id=yxEGywAACAAJ},
  year={2012},
  publisher={W. W. Norton}
}

@book{imbens2015causal,
  title={Causal Inference in Statistics, Social, and Biomedical Sciences},
  author={Imbens, Guido W and Rubin, Donald B},
  year={2015},
  publisher={Cambridge university press}
}

@article{shimodaira2000improving,
  title = {Improving Predictive Inference under Covariate Shift by Weighting the Log-Likelihood Function},
  author = {Shimodaira, Hidetoshi},
  year = {2000},
  journal = {Journal of Statistical Planning and Inference},
  volume = {90},
  number = {2},
  pages = {227--244},
  doi = {10.1016/S0378-3758(00)00115-4}
}

@article{koenker2001quantile,
  title = {Quantile {{Regression}}},
  author = {Koenker, Roger and Hallock, Kevin F.},
  year = {2001},
  journal = {Journal of Economic Perspectives},
  volume = {15},
  number = {4},
  pages = {143--156},
  doi = {10.1257/jep.15.4.143}
}

@article{yang2024doubly,
  title = {Doubly Robust Calibration of Prediction Sets under Covariate Shift},
  author = {Yang, Yachong and Kuchibhotla, Arun Kumar and Tchetgen Tchetgen, Eric},
  year = {2024},
  journal = {Journal of the Royal Statistical Society Series B: Statistical Methodology},
  volume = {86},
  number = {4},
  pages = {943--965},
  doi = {10.1093/jrsssb/qkae009}
}

@article{rosenbaum1983central,
  title = {The Central Role of the Propensity Score in Observational Studies for Causal Effects},
  author = {Rosenbaum, Paul R and Rubin, Donald B},
  year = 1983,
  journal = {Biometrika},
  volume = {70},
  number = {1},
  pages = {41--55}
}

@article{kennedy2023semiparametric,
  title = {Semiparametric Doubly Robust Targeted Double Machine Learning: A Review},
  shorttitle = {Semiparametric Doubly Robust Targeted Double Machine Learning},
  author = {Kennedy, Edward H.},
  year = {2023},
  journal = {arXiv preprint arXiv:2203.06469},
  eprint = {2203.06469},
  primaryclass = {stat},
  publisher = {arXiv},
  doi = {10.48550/arXiv.2203.06469},
  archiveprefix = {arXiv}
}

@article{newey1994asymptotic,
  title = {The {{Asymptotic Variance}} of {{Semiparametric Estimators}}},
  author = {Newey, Whitney K.},
  year = {1994},
  journal = {Econometrica},
  volume = {62},
  number = {6},
  eprint = {2951752},
  eprinttype = {jstor},
  pages = {1349--1382},
  publisher = {[Wiley, Econometric Society]},
  doi = {10.2307/2951752}
}

@article{kallus2024localized,
  title = {Localized {{Debiased Machine Learning}}: {{Efficient Inference}} on {{Quantile Treatment Effects}} and {{Beyond}}},
  shorttitle = {Localized {{Debiased Machine Learning}}},
  author = {Kallus, Nathan and Mao, Xiaojie and Uehara, Masatoshi},
  year = {2024},
  journal = {Journal of Machine Learning Research},
  volume = {25},
  number = {16},
  pages = {1--59}
}

@article{wager2018estimation,
  title = {Estimation and {{Inference}} of {{Heterogeneous Treatment Effects}} Using {{Random Forests}}},
  author = {Wager, Stefan and Athey, Susan},
  year = {2018},
  journal = {Journal of the American Statistical Association},
  volume = {113},
  number = {523},
  pages = {1228--1242},
  publisher = {ASA Website},
  doi = {10.1080/01621459.2017.1319839}
}

@article{sesia2020comparison,
  title = {A Comparison of Some Conformal Quantile Regression Methods},
  author = {Sesia, Matteo and Cand{\`e}s, Emmanuel J.},
  year = {2020},
  journal = {Stat},
  volume = {9},
  number = {1},
  pages = {e261},
  doi = {10.1002/sta4.261}
}

@article{honaker2011amelia,
  title = {Amelia {{II}}: {{A Program}} for {{Missing Data}}},
  shorttitle = {Amelia {{II}}},
  author = {Honaker, James and King, Gary and Blackwell, Matthew},
  year = {2011},
  journal = {Journal of Statistical Software},
  volume = {45},
  pages = {1--47},
  doi = {10.18637/jss.v045.i07}
}

@article{margalit2021how,
  title = {How {{Markets Shape Values}} and {{Political Preferences}}: {{A Field Experiment}}},
  shorttitle = {How {{Markets Shape Values}} and {{Political Preferences}}},
  author = {Margalit, Yotam and Shayo, Moses},
  year = {2021},
  journal = {American Journal of Political Science},
  volume = {65},
  number = {2},
  pages = {473--492},
  doi = {10.1111/ajps.12517}
}

@article{finkel2024can,
  title = {Can {{Online Civic Education Induce Democratic Citizenship}}? {{Experimental Evidence}} from a {{New Democracy}}},
  shorttitle = {Can {{Online Civic Education Induce Democratic Citizenship}}?},
  author = {Finkel, Steven E. and Neundorf, Anja and Rasc{\'o}n Ram{\'i}rez, Ericka},
  year = {2024},
  journal = {American Journal of Political Science},
  volume = {68},
  number = {2},
  pages = {613--630},
  doi = {10.1111/ajps.12765}
}
    %\end{spacing}
    
    \clearpage

    \appendix

    \counterwithin*{equation}{section} % reset 'equation' counter whenever '\section' is executed
    \counterwithin{algorithm}{section} % reset 'algorithm' counter whenever '\section' is executed
    \counterwithin{figure}{section} % reset 'figure' counter whenever '\section' is executed
    \counterwithin{table}{section} % reset 'table' counter whenever '\section' is executed
    \renewcommand\theequation{\thesection\arabic{equation}} % how to display the equation "number"
    \counterwithin{theorem}{section}
    \counterwithin{lemma}{section}
    \renewcommand{\appendixpagename}{Appendix}

    %\section*{\LARGE Appendix}
    \appendixpage

    %\tableofcontents

    %\clearpage

    \section{Conformal Inference}

    \subsection{Marginal Coverage}

    \begin{theorem}\label{thm:mcoverage}
        Suppose that $(X_1, Y_1), \dots, (X_{n + 1}, Y_{n + 1})$ are \textit{exchangeable} and $s$ is a symmetric conformal score function. Then, the prediction interval $\mathcal{C}(X_{n + 1})$ satisfies the marginal coverage guarantee,
        \[\mathbb{P}\left( Y_{n + 1} \in \mathcal{C}(X_{n + 1}) \right) \geq 1 - \alpha.\]
    \end{theorem}

    \subsection{Exchangeability}

    \begin{definition}[Exchangeability]\label{exch}
        Let $Z_1, \dots, Z_n \in \mathcal{Z}$ be random variables with a joint distribution. We say that the random vector $(Z_1, \dots, Z_n)$ is exchangeable if, for every permutation $\sigma \in \mathcal{S}_n$,
        \[(Z_1, \dots, Z_n) \overset{\mathrm{d}}{=} (Z_{\sigma(1)}, \dots, Z_{\sigma(n)}),\]
        where $\overset{\mathrm{d}}{=}$ denotes equality in distribution, and $\mathcal{S}_n$ is the set of all permutations on $[n] \coloneq \left\{ 1, \dots, n \right\}$. 

        Similarly, let $Z_1, Z_2, \dots \in \mathcal{Z}$ be an infinite sequence of random variables with joint distribution. We say that this infinite sequence is exchangeable if $(Z_1, \dots, Z_n)$ is exchangeable for every $n \geq 1$.
    \end{definition}

    \clearpage

    \section{Asymptotic Coverage of the Prediction Interval for ITE with Attrition}\label{sec:asym}

    In line with \citet{yang2024doubly}, we connect the targert coverage level $(1 - \gamma)$ and the EIF for $\eta_{\gamma, \mathcal{C}}$.

    \begin{lemma}\label{lem:cov}
        Let $\pi_R: \left\{ 0, 1 \right\} \times \mathcal{X} \to \mathbb{R}^{+}$ and $m_\mathcal{C}: \left\{ 0, 1  \right\} \times \mathcal{X} \times \mathbb{R} \to [0, 1]$ be any two functions. Then, under Assumption \ref{asm2} and \ref{asm3}, the following coverage for the EIF $\psi_\mathcal{C}\left( \eta_{\gamma, \mathcal{C}}, X \right)$.
        \begin{align*}
            \mathbb{P}\left( V_\mathcal{C} < \eta_{\gamma, \mathcal{C}} \mid R = 0 \right) = 1 - \gamma + \frac{\mathbb{E}\left[ \psi_\mathcal{C}\left( \eta_{\gamma, \mathcal{C}}, X; m_\mathcal{C}, \pi_R \right) \right]}{\mathbb{P}(R = 0)},
        \end{align*}
        holds true whenever either of the following holds true:
        \begin{enumerate}
            \item $\pi_R(X, D)$ is estimated consistently or
            \item $m_\mathcal{C}\left( \eta_{\gamma, \mathcal{C}}, X, D \right)$ is estimated consistently.
        \end{enumerate}
    \end{lemma}

    \begin{proof}
        See Section \ref{sec:pflemcov} for the proof of Lemma \ref{lem:cov}. 
    \end{proof}

    Lemma \ref{lem:cov} states that the asymptotic coverage of the ITE prediction interval for the attrition group deviates from the nominal level $1 - \gamma$ by a term that is proportional to the expectation of the EIF. It further establishes the double robustness of both the influence function and the resulting coverage: as long as either $\pi_R$ or $m_\mathcal{C}$ is consistently estimated, the desired coverage can be achieved. We have a set of regularity conditions required to ensure the asymptotic coverage of the estimator $\eta_{\gamma, \mathcal{C}}$.

    \begin{itemize}[noitemsep]
        \item[(A1)]  The function $(\eta, x) \mapsto \hat{m}_{\mathcal{C}}(\eta, X, D)$ is bounded, i.e., there exists $m_0$ such that for all $\eta \in \mathbb{R}$ and $x \in \mathbb{R}^d$, $|\hat{m}(\eta, x, d)| \leq m_0$. The function $x \mapsto \hat{\pi}_R(x, d)$ is bounded from below by a positive constant, i.e., there exists $\pi_0 > 0$ such that for all $x \in \mathbb{R}^d$, $|\hat{\pi}_R(x, d)| \geq \pi_0$.
        \item[(A2)] The estimator $\hat{m}_\mathcal{C}(\eta, x, d)$ is a non-decreasing function of $\eta$.
    \end{itemize}

    Following the two-step conformal inference framework for constructing interval estimates of the ITE under attrition, Theorem \ref{thm:asyite} establishes that the proposed method attains the desired asymptotic coverage for the ITE prediction intervals in the attrition group.

    \begin{theorem}\label{thm:asyite}
        Under Assumption \ref{asm1} to \ref{asm3} and regularity conditions (A1) and (A2), there exists some universal constants $C_0$ and $C_1$ such that for any $\delta > 0$,
        \begin{align*}
            \mathbb{P}\left( V_\mathcal{C} < \hat{\eta}_{\gamma, \mathcal{C}} \mid R = 0 \right) \geq& 1 - \gamma\\
            &- C_0 \frac{\pi_0^{-1} (1 + m_0)}{\mathbb{P}(R = 0)} \sqrt{\frac{\log(1 / \delta) + 1}{|\mathcal{I}_2|}}\\
            &- C_1 \frac{\norm{\hat{\pi}_R(X, D) - \pi_R(X, D)}_2 }{\mathbb{P}(R = 0)} \cdot \sup_{\eta} \norm{\hat{m}_\mathcal{C}\left( \hat{\eta}_{\gamma, \mathcal{C}}, X, D \right) - m_\mathcal{C}\left( \hat{\eta}_{\gamma, \mathcal{C}}, X, D \right)}_2
        \end{align*}
        holds with probability at least $1 - \delta$, where $|\mathcal{I}_2|$ is the Lebesgue measure of the calibration data when conducting conformal inference on the interval estimates of ITE with attrition. 
    \end{theorem}

    \begin{proof}
        See Section \ref{sec:pfasyite} for the proof of Theorem \ref{thm:asyite}.
    \end{proof}

    \clearpage

    \section{Technical Proofs}

    \subsection{Proof of Lemma \ref{lem:ident1} and \ref{lem:ident2}}\label{sec:pflem12}

    \begin{proof}[Proof of Lemma \ref{lem:ident1}]
       The result is immediate following the law of iterated expectations and the unconfoundedness. 
    \end{proof}

    \begin{proof}[Proof of Lemma \ref{lem:ident2}]
        The result is immediate following the law of iterated expectations and the MAR assumption.
    \end{proof}

    \subsection{Proof of Theorem \ref{thm:eifcf}}\label{sec:pfeifcf}

    \begin{proof}[Proof of Theorem \ref{thm:eifcf}]
        To prove Theorem \ref{thm:eifcf} and derive the EIF, we leverage Gateaux derivatives (directional derivatives) and the Riesz representation theorem.

        Define $\mathcal{O} \coloneq (X, Y, D, R)$. The efficient influence function is defined as the unique mean-zero function $\psi$ such that the pathwise derivative of the target parameter $\eta_{\alpha, 1}$ satisfies:
        \begin{align}
            \frac{\partial}{\partial t} \eta_{\alpha, 1} (P_t) \big|_{t = 0} = \mathbb{E}\left[ \psi(\mathcal{O}) \cdot s(\mathcal{O}) \right], \label{eq:if}
        \end{align}
        where $P_t$ is a perturbation of the true distribution $P$ along a score function $s(\mathcal{O})$, and $s(\mathcal{O})$ is any valid score function. According to \citet{kennedy2023semiparametric}, the influence function $\psi$ has been referred to as pathwise derivative, gradient, and Neyman orthogonal score \citep[e.g.,][]{chernozhukov2018double,newey1994asymptotic}.

        Define the functional 
        \begin{align*}
            \Psi(\mathcal{O}) &\coloneq \mathbb{E}\left[ m_1(\eta_{\alpha, 1}, X) \mid D = 0, R = 1\right] - (1 - \alpha)\\
            &= \frac{\mathbb{E}\left[\left(  m_1(\eta_{\alpha, 1}, X) - (1 - \alpha) \right) \cdot \mathds{1}_{\left\{ D = 0, R = 1 \right\}} \right] }{\mathbb{E}\left[ \mathds{1}_{\left\{ D = 0, R = 1 \right\}} \right]}\\
            &= \mathbb{E}\left[ R(1 - D)\left( m_1(\eta_{\alpha, 1}, X) - (1 - \alpha) \right) \right].
        \end{align*}
        Then, the parameter $\eta_{\alpha, 1}$, which is the $(1 - \alpha)$-quantile of the nonconformity score $V_1 = V(X, Y(1))$ given $D = 0, R = 1$, is identified by the moment condition
        \begin{align*}
            \Psi(\mathcal{O}) = \mathbb{E}\left[ R(1 - D)\left( m_1(\eta_{\alpha, 1}, X) - (1 - \alpha) \right) \right] = 0.
        \end{align*}
        Consider a pathwise perturbation of the true distribution $P$ along a score function $s(\mathcal{O})$, where $s(\mathcal{O})$ satisfies:
        \[\mathbb{E}[s(\mathcal{O})] = 0, \quad \mathbb{E}\left[ s(\mathcal{O})^2 \right] < \infty.\]
        The perturbed distribution is:
        \[P_t(\mathcal{O}) = (1 + t \cdot s(\mathcal{O})) P(\mathcal{O}),\]
        for small $t$, ensuring $P_t$ remains a valid probability distribution.

        The perturbed parameter $\eta_{\alpha, 1}(P_t)$ satisfies the perturbed moment condition:
        \begin{align}
            \mathbb{E}_t\left[ R(1 - D) \left( m_{1, t}(\eta_{\alpha, 1}(P_t), X) - (1 - \alpha)\right) \right] = 0
        \end{align}
        To find how $\eta_{\alpha, 1}(P_t)$ changes with $t$, we take the total derivative of the perturbed moment condition at w.r.t. $t$ at $t = 0$:
        \begin{align}
            \frac{\partial}{\partial t} \mathbb{E}_t\left[ R(1 - D) \left( m_{1, t}(\eta_{\alpha, 1}(P_t), X) - (1 - \alpha)\right) \right] \bigg|_{t = 0} = 0.
        \end{align}

        We have the result that the expectation under $P_t$ of any function $g(\mathcal{O})$ is 
        \begin{align*}
            \mathbb{E}_t[g(\mathcal{O})] &= \int g(\mathcal{O})P_t(\mathcal{O}) d \mathcal{O} = \int g(\mathcal{O}) ( 1 + t \cdot s(\mathcal{O})) P(\mathcal{O}) d\mathcal{O}\\
            &= \mathbb{E}[g(\mathcal{O})] + t \mathbb{E}\left[ g(\mathcal{O}) \cdot s(\mathcal{O}) \right].
        \end{align*}
        Then, for the function $g(\mathcal{O}) = R(1 - D) \left( m_{1, t}(\eta_{\alpha, 1}(P_t), X) - (1 - \alpha)\right)$, the perturbed moment condition is
        \begin{align*}
            \mathbb{E}_t\left[ R(1 - D) \left( m_{1, t}(\eta_{\alpha, 1}(P_t), X) - (1 - \alpha)\right) \right] &= \mathbb{E}\left[ R(1 - D)\left( m_1(\eta_{\alpha, 1}, X) - (1 - \alpha) \right) \right]\\
            &+ t \mathbb{E}\left[ s(\mathcal{O})\cdot R(1 - D) \left( m_1(\eta_{\alpha, 1}, X) - (1 - \alpha) \right) \right].
        \end{align*}

        By the chain rule, the derivative of the perturbed moment condition at $t = 0$ has three components: the perturbation of the expectation $\mathbb{E}$, the perturbation of the conditional probability $m_1$, and the perturbation of the quantile $\eta_{\alpha, 1}$. So, we can write the total derivatives explicitly.
        By the result above, the perturbation of the expectation $\mathbb{E}$ gives
        \begin{align*}
            \text{I} &= \frac{\partial}{\partial t} \mathbb{E}_t\left[ R(1 - D) \left( m_{1, t}(\eta_{\alpha, 1}(P_t), X) - (1 - \alpha)\right) \right] \bigg|_{t = 0} \\
            &= \mathbb{E}\left[ s(\mathcal{O}) \cdot R(1 - D) \left( m_{1}(\eta_{\alpha, 1}, X) - (1 - \alpha)\right)\right].
        \end{align*}
        The perturbation of the conditional probability $m_1$ gives
        \begin{align*}
            \text{II} = \mathbb{E}\left[ R(1 - D) \frac{\partial m_{1, t}(\eta_{\alpha, 1}, X)}{\partial t} \bigg|_{t = 0}\right]
        \end{align*}
        Under the perturbed distribution $P_t$, by Bayes' rule, the conditional probability becomes
        \begin{align*}
            m_{1, t}(\eta_{\alpha, 1}, X) &= P_t(V_1 < \eta_{\alpha, 1} \mid X, D = 1, R = 1)\\
            &= \frac{P_t(V_1 < \eta_{\alpha, 1}, D = 1, R = 1 \mid X)}{P_t(D = 1, R = 1 \mid X)}\\
            &= \frac{\mathbb{E}_t\left[ \mathds{1}_{\left\{ V_1 < \eta_{\alpha, 1} \right\}} DR \mid X \right]}{\mathbb{E}_t\left[ DR \mid X \right]}\\
            &= \frac{\mathbb{E}\left[ \mathds{1}_{\left\{ V_1 < \eta_{\alpha, 1} \right\}} DR \mid X \right] + t \mathbb{E}\left[s(\mathcal{O}) \cdot \mathds{1}_{\left\{ V_1 < \eta_{\alpha, 1} \right\}} DR \mid X \right]}{\mathbb{E}\left[ DR \mid X \right] + t \mathbb{E}\left[ s(\mathcal{O}) \cdot DR \mid X\right]}\\
            &\approx m_1(\eta_{\alpha, 1}, X) + t \left( \frac{\mathbb{E}\left[s(\mathcal{O}) \cdot \mathds{1}_{\left\{ V_1 < \eta_{\alpha, 1} \right\}} DR \mid X \right]}{\mathbb{E}\left[ DR \mid X \right]} - m_1(\eta_{\alpha, 1}, X) \cdot \frac{\mathbb{E}\left[ s(\mathcal{O}) \cdot DR \mid X\right]}{\mathbb{E}\left[ DR \mid X \right]} \right),
        \end{align*}
        where the last approximation comes from the first-order Taylor expansion. Then, we take the derivative of $m_{1, t}(\eta_{\alpha, 1}, X)$ w.r.t. $t$ at $t = 0$:
        \begin{align*}
            \frac{\partial m_{1, t}(\eta_{\alpha, 1}, X)}{\partial t} \bigg|_{t = 0} &= \frac{\mathbb{E}\left[s(\mathcal{O}) \cdot \mathds{1}_{\left\{ V_1 < \eta_{\alpha, 1} \right\}} DR \mid X \right]}{\mathbb{E}\left[ DR \mid X \right]} - m_1(\eta_{\alpha, 1}, X) \cdot \frac{\mathbb{E}\left[ s(\mathcal{O}) \cdot DR \mid X\right]}{\mathbb{E}\left[ DR \mid X \right]}\\
            &= \mathbb{E}\left[s(\mathcal{O}) \cdot \left( \frac{\mathds{1}_{\left\{ V_1 < \eta_{\alpha, 1} \right\}} DR }{\mathbb{E}\left[ DR \mid X \right]} - m_1(\eta_{\alpha, 1}, X) \cdot \frac{DR}{\mathbb{E}\left[ DR \mid X \right]} \right) \mid X \right]\\
            &= \mathbb{E}\left[ s(\mathcal{O}) \cdot \left( \mathds{1}_{\left\{ V_1 < \eta_{\alpha, 1}\right\}} - m_1(\eta_{\alpha, 1}, X) \right) \cdot \frac{DR}{P(D = 1, R = 1 \mid X)} \mid X \right].
        \end{align*}
        Then, plug in to II, we have
        \begin{align*}
            \text{II} &= \mathbb{E} \cdot \left[ R(1 - D) \mathbb{E}\left[ s(\mathcal{O}) \cdot \left( \mathds{1}_{\left\{ V_1 < \eta_{\alpha, 1}\right\}} - m_1(\eta_{\alpha, 1}, X) \right) \cdot \frac{DR}{P(D = 1, R = 1 \mid X)} \mid X \right] \right]\\
            &= \mathbb{E}\left[ s(\mathcal{O}) \cdot DR \cdot \frac{P(D = 0, R = 1 \mid X)}{P(D = 1, R = 1 \mid X)} \cdot \left( \mathds{1}_{\left\{ V_1 < \eta_{\alpha, 1}\right\}} - m_1(\eta_{\alpha, 1}, X) \right)\right]\\
            &= \mathbb{E}\left[ s(\mathcal{O}) \cdot DR \cdot \frac{(1 - e_D(X)) e_R(X, 0)}{e_D(X) e_R(X, 1)} \cdot \left( \mathds{1}_{\left\{ V_1 < \eta_{\alpha, 1}\right\}} - m_1(\eta_{\alpha, 1}, X) \right) \right]\\
            &= \mathbb{E}\left[ s(\mathcal{O}) \cdot DR \cdot \frac{e_R(X, 0)}{\pi_D(X) e_R(X, 1)} \cdot \left( \mathds{1}_{\left\{ V_1 < \eta_{\alpha, 1}\right\}} - m_1(\eta_{\alpha, 1}, X) \right)\right].
        \end{align*}
        The perturbation of the quantile $\eta_{\alpha, 1}$ gives
        \begin{align*}
            \text{III} &= \mathbb{E}\left[ R(1 - D) \frac{\partial m_1(\eta_{\alpha, 1}, X)}{\partial \eta} \cdot \frac{\partial \eta_{\alpha, 1}(P_t)}{\partial t} \bigg|_{t = 0} \right]\\
            &\propto \frac{\partial }{\partial t}\eta_{\alpha, 1}(P_t).
        \end{align*}
        Therefore, the total derivative of the perturbed moment condition at $t = 0$ is
        \begin{align*}
            \frac{\partial}{\partial t} \mathbb{E}_t\left[ R(1 - D) \left( m_{1, t}(\eta_{\alpha, 1}(P_t), X) - (1 - \alpha)\right) \right] \bigg|_{t = 0} = 0 = \text{I} + \text{II} + \text{III}.
        \end{align*}
        Then, 
        \begin{equation}
            \begin{aligned}
            \frac{\partial }{\partial t}&\eta_{\alpha, 1}(P_t) \propto\\
             &-\mathbb{E}\left[ s(\mathcal{O}) \cdot \left\{ R(1 - D)\left( m_1(\eta_{\alpha, 1}, X) - (1 - \alpha) \right) + \frac{DRe_R(X, 0)}{\pi_D(X) e_R(X, 1)} \cdot \left( \mathds{1}_{\left\{ V_1 < \eta_{\alpha, 1}\right\}} - m_1(\eta_{\alpha, 1}, X) \right) \right\} \right]\nonumber
        \end{aligned}
        \end{equation}
        
        By the Riesz representation theorem and Equation \ref{eq:if}, the efficient influence function $\psi$ is given by
        \begin{equation}
           \begin{aligned}
            \psi \left( \eta_{\alpha, 1}, X; m, e_R, \pi_D \right)& = \\
            &R(1 - D)\left( m_1(\eta_{\alpha, 1}, X) - (1 - \alpha) \right) + \frac{DRe_R(X, 0)}{\pi_D(X) e_R(X, 1)} \cdot \left( \mathds{1}_{\left\{ V_1 < \eta_{\alpha, 1}\right\}} - m_1(\eta_{\alpha, 1}, X) \right).\nonumber
        \end{aligned} 
        \end{equation}
        
        The EIF under $D = 0$ follows the same derivation. There are other ways to derive the EIF. See \citet{gao2025role} for a proof using parametric submodels and tangent spaces from a Hilbert space perspective.
    \end{proof}

    \subsection{Proof of Theorem \ref{thm:eifint}}\label{sec:pfeifint}

    \begin{proof}[Proof of Theorem \ref{thm:eifint}]
        The proof follows the same steps as in Theorem \ref{thm:eifcf}. 

        The $(1 - \gamma)$-quantile of the nonconformity score $V_\mathcal{C} = V(X, \mathcal{C}_i)$ for the attrition group $R = 0$ is identified by the moment condition
        \begin{align*}
            &\mathbb{E}\left[ m_{\mathcal{C}}(\eta_{\gamma, \mathcal{C}}, X, D) \mid R = 0 \right] - (1 - \gamma) = 0 \\
            \Rightarrow &\mathbb{E}\left[ (1 - R) \left( m_{\mathcal{C}}(\eta_{\gamma, \mathcal{C}}, X, D) - (1 - \gamma) \right) \right] = 0.
        \end{align*}

        Consider a pathwise perturbation of the true distribution $P$ along a score function $s(\mathcal{O})$, where $s(\mathcal{O})$ satisfies:
        \[\mathbb{E}[s(\mathcal{O})] = 0, \quad \mathbb{E}\left[ s(\mathcal{O})^2 \right] < \infty.\]
        The perturbed distribution is:
        \[P_t(\mathcal{O}) = (1 + t \cdot s(\mathcal{O})) P(\mathcal{O}),\]
        for small $t$, ensuring $P_t$ remains a valid probability distribution.

        Under $P_t$, the moment condition becomes
        \begin{align*}
            \mathbb{E}_t\left[ (1 - R) \left( m_{\mathcal{C}, t}\left( \eta_{\gamma, \mathcal{C}}(P_t), X, D \right) - (1 - \gamma)\right) \right] = 0.
        \end{align*}

        To find how $\eta_{\gamma, \mathcal{C}}(P_t)$ changes with $t$, we take the total derivative of the perturbed moment condition at w.r.t. $t$ at $t = 0$:
        \begin{align*}
            \frac{\partial}{\partial t} \mathbb{E}_t\left[ (1 - R) \left( m_{\mathcal{C}, t}\left( \eta_{\gamma, \mathcal{C}}(P_t), X, D \right) - (1 - \gamma)\right) \right] \bigg|_{t = 0} = 0.
        \end{align*}
        By the results shown in Theorem \ref{thm:eifcf}, we have 
        \begin{align*}
            \mathbb{E}_t\left[ (1 - R) \left( m_{\mathcal{C}, t}\left( \eta_{\gamma, \mathcal{C}}(P_t), X, D \right) - (1 - \gamma)\right) \right] &= \mathbb{E}\left[ (1 - R) \left( m_{\mathcal{C}}\left( \eta_{\gamma, \mathcal{C}}, X, D \right) - (1 - \gamma)\right) \right]\\
            &+ t \mathbb{E}\left[ s(\mathcal{O}) \cdot (1 - R) \left( m_{\mathcal{C}}\left( \eta_{\gamma, \mathcal{C}}, X, D \right) - (1 - \gamma)\right) \right].
        \end{align*}
        Then, by the chain rule, the derivative of the perturbed moment condition at $t = 0$ has three components: the perturbation of the expectation $\mathbb{E}$, the perturbation of the conditional probability $m_{\mathcal{C}}$, and the perturbation of the quantile $\eta_{\gamma, \mathcal{C}}$. So, we can write the total derivatives explicitly.

        By the result above, the perturbation of the expectation $\mathbb{E}$ gives
        \begin{align*}
            \text{I} &= \frac{\partial}{\partial t} \mathbb{E}_t\left[ (1 - R) \left( m_{\mathcal{C}}\left( \eta_{\gamma, \mathcal{C}}, X, D \right) - (1 - \gamma)\right) \right] \bigg|_{t = 0} \\
            &= \mathbb{E}\left[ s(\mathcal{O}) \cdot (1 - R) \left( m_{\mathcal{C}}\left( \eta_{\gamma, \mathcal{C}}, X, D \right) - (1 - \gamma)\right) \right].
        \end{align*}
        Then perturbation of the conditional probability $m_{\mathcal{C}}$ gives
        \begin{align*}
            \text{II} &= \mathbb{E}\left[ (1 - R) \frac{\partial m_{\mathcal{C}, t}\left( \eta_{\gamma, \mathcal{C}}, X, D \right)}{\partial t} \bigg|_{t = 0}\right].
        \end{align*}
        Under the perturbed distribution $P_t$, by Bayes' rule, the conditional probability becomes
        \begin{align*}
            m_{\mathcal{C}, t}\left( \eta_{\gamma, \mathcal{C}}, X, D \right) &= P_t\left( V_\mathcal{C} < \eta_{\gamma, \mathcal{C}} \mid X, D, R = 1 \right) \\
            &= \frac{P_t\left( V_\mathcal{C} < \eta_{\gamma, \mathcal{C}}, R = 1 \right)}{P_t(R = 1 \mid X, D)}\\
            &= \frac{\mathbb{E}_t\left[ \mathds{1}_{\left\{ V_\mathcal{C} < \eta_{\gamma, \mathcal{C}} \right\}} R \mid X, D \right]}{\mathbb{E}_t\left[ R \mid X, D \right]}\\
            &= \frac{\mathbb{E}\left[ \mathds{1}_{\left\{ V_\mathcal{C} < \eta_{\gamma, \mathcal{C}} \right\}} R \mid X, D \right] + t \mathbb{E}\left[s(\mathcal{O}) \cdot \mathds{1}_{\left\{ V_\mathcal{C} < \eta_{\gamma, \mathcal{C}} \right\}} R \mid X, D\right]}{\mathbb{E}[R \mid X, D] + t \mathbb{E}[s(\mathcal{O}) \cdot R \mid X, D]}\\
            & \approx m_{\mathcal{C}}\left( \eta_{\gamma, \mathcal{C}}, X, D \right) + t \left( \frac{\mathbb{E}\left[s(\mathcal{O}) \cdot \mathds{1}_{\left\{ V_\mathcal{C} < \eta_{\gamma, \mathcal{C}} \right\}} R \mid X, D \right]}{\mathbb{E}[R \mid X, D]} - m_\mathcal{C}\left( \eta_{\gamma, \mathcal{C}}, X, D \right) \cdot \frac{\mathbb{E}\left[ s(\mathcal{O}) \cdot R \mid X, D \right]}{\mathbb{E}\left[ R \mid X, D \right]} \right),
        \end{align*}
        where the last approximation comes from the first-order Taylor expansion. Then, taking the derivative of $m_{\mathcal{C}, t}\left( \eta_{\gamma, \mathcal{C}}, X, D \right)$ w.r.t. $t$ at $t = 0$:
        \begin{align*}
            \frac{\partial m_{\mathcal{C}, t}\left( \eta_{\gamma, \mathcal{C}}, X, D \right)}{\partial t} \bigg|_{t = 0} &= \frac{\mathbb{E}\left[ s(\mathcal{O}) \cdot R\left( \mathds{1}_{\left\{ V_\mathcal{C} < \eta_{\gamma, \mathcal{C}} \right\}} - m_\mathcal{C}(\eta_{\gamma, \mathcal{C}}, X, D) \right) \mid X, D \right]}{\mathbb{E}[R \mid X, D]}\\
            &= \mathbb{E}\left[ s(\mathcal{O}) \cdot \frac{R}{e_R(X, D)} \left( \mathds{1}_{\left\{ V_\mathcal{C} < \eta_{\gamma, \mathcal{C}} \right\}} - m_\mathcal{C}(\eta_{\gamma, \mathcal{C}}, X, D) \right) \mid X \right].
        \end{align*}
        Plug into II, we have 
        \begin{align*}
            \text{II} &= \mathbb{E}\left[ (1 - R) \mathbb{E}\left[ s(\mathcal{O}) \cdot \frac{R}{e_R(X, D)} \left( \mathds{1}_{\left\{ V_\mathcal{C} < \eta_{\gamma, \mathcal{C}} \right\}} - m_\mathcal{C}(\eta_{\gamma, \mathcal{C}}, X, D) \right) \mid X, D \right] \right]\\
            &= \mathbb{E}\left[ s(\mathcal{O}) \cdot R \frac{1 - e_R(X, D)}{e_R(X, D)} \left( \mathds{1}_{\left\{ V_\mathcal{C} < \eta_{\gamma, \mathcal{C}} \right\}} - m_\mathcal{C}(\eta_{\gamma, \mathcal{C}}, X, D) \right)\right]\\
            &= \mathbb{E}\left[ s(\mathcal{O}) \cdot \frac{R}{\pi_R(X, D)} \left( \mathds{1}_{\left\{ V_\mathcal{C} < \eta_{\gamma, \mathcal{C}} \right\}} - m_\mathcal{C}(\eta_{\gamma, \mathcal{C}}, X, D) \right)\right]
        \end{align*}
        The perturbation of the quantile $\eta_{\gamma, \mathcal{C}}$ gives
        \begin{align*}
            \text{III} &= \mathbb{E}\left[ R \frac{\partial m_\mathcal{C}\left( \eta_{\gamma, \mathcal{C}}, X, D \right)}{\partial \eta} \frac{\partial \eta_{\gamma, \mathcal{C}}(P_t)}{\partial t} \bigg|_{t = 0} \right]\\
            &\propto \frac{\partial }{\partial t}\eta_{\gamma, \mathcal{C}}(P_t).
        \end{align*}
        Therefore, the total derivative of the perturbed moment condition at $t = 0$ is
        \begin{align*}
            \frac{\partial}{\partial t} \mathbb{E}_t\left[ (1 - R) \left( m_{\mathcal{C}, t}\left( \eta_{\gamma, \mathcal{C}}(P_t), X, D \right) - (1 - \gamma)\right) \right] \bigg|_{t = 0} = 0 = \text{I} + \text{II} + \text{III}.
        \end{align*}
        Then, 
        \begin{align*}
            \frac{\partial }{\partial t}&\eta_{\gamma, \mathcal{C}}(P_t)\\
             &\propto -\mathbb{E}\left[ s(\mathcal{O}) \cdot \left\{ (1 - R)\left( m_\mathcal{C}(\eta_{\gamma, \mathcal{C}}, X, D) - (1 - \gamma) \right) + \frac{R}{\pi_R(X, D)} \cdot \left( \mathds{1}_{\left\{ V_\mathcal{C} < \eta_{\gamma, \mathcal{C}}\right\}} - m_\mathcal{C}(\eta_{\gamma, \mathcal{C}}, X, D) \right) \right\} \right].
        \end{align*}
        By the Riesz representation theorem, the efficient influence function $\psi_\mathcal{C}$ is given by 
        \begin{align*}
            \psi_\mathcal{C} \left( \eta_{\gamma, \mathcal{C}}, X; m_\mathcal{C}, e_R, \pi_R \right) &= (1 - R)\left( m_\mathcal{C}(\eta_{\gamma, \mathcal{C}}, X, D) - (1 - \gamma) \right)\\
            &+ \frac{R}{\pi_R(X, D)} \cdot \left( \mathds{1}_{\left\{ V_\mathcal{C} < \eta_{\gamma, \mathcal{C}}\right\}} - m_\mathcal{C}(\eta_{\gamma, \mathcal{C}}, X, D) \right).
        \end{align*} 
    \end{proof}

    \subsection{Proof of Theorem \ref{thm:asyite}}\label{sec:pfasyite}

    \begin{proof}[Proof of Theorem \ref{thm:asyite}]
        By Lemma \ref{lem:cov}, we can obtain that for the empirically estimated $\hat{\eta}_{\gamma, \mathcal{C}}$ using the calibration data $\mathcal{I}_2$ for conducting conformal inference on the attrition group,
        \begin{align*}
            \mathbb{P}\left( V_\mathcal{C} < \hat{\eta}_{\gamma, \mathcal{C}} \mid R = 0 \right) - (1 - \gamma) &= \frac{\mathbb{P}_{\mathcal{I}_2}\left[ \psi_\mathcal{C}\left( \hat{\eta}_{\gamma, \mathcal{C}}, X; \hat{m}_\mathcal{C}, \hat{\pi}_R \right) \right]}{\mathbb{P}(R = 0)}\\
            &+ \frac{\mathbb{E}\left[ \psi_\mathcal{C}\left( \hat{\eta}_{\gamma, \mathcal{C}}, X; \hat{m}_\mathcal{C}, \hat{\pi}_R \right) \right] - \mathbb{P}_{\mathcal{I}_2}\left[ \psi_\mathcal{C}\left( \hat{\eta}_{\gamma, \mathcal{C}}, X; \hat{m}_\mathcal{C}, \hat{\pi}_R \right) \right]}{\mathbb{P}(R = 0)}\\
            &+ \frac{\mathbb{E}\left[ \psi\left( \hat{\eta}_{\gamma, \mathcal{C}}, X; m_\mathcal{C}, \pi_R \right) \right] - \mathbb{E}\left[ \psi_\mathcal{C}\left( \hat{\eta}_{\gamma, \mathcal{C}}, X; \hat{m}_\mathcal{C}, \hat{\pi}_R \right) \right]}{\mathbb{P}(R = 0)}\\
            &\geq 0 + \text{I} + \text{II},
        \end{align*}
        where $\mathbb{P}_{\mathcal{I}_2}\left[ \cdot \right]$ denotes the sample mean over the calibration data. Here, $\mathbb{P}_{\mathcal{I}_2} \left[ \psi_\mathcal{C}\left( \hat{\eta}_{\gamma, \mathcal{C}}, X; \hat{m}_\mathcal{C}, \hat{\pi}_R \right) \right] \geq 0$ by definition. Term $\text{I}$ is negligible if $\psi_\mathcal{C}\left( \eta_{\gamma, \mathcal{C}}, X; m_\mathcal{C}, \pi_R \right)$ belongs to a Donsker class. If the Donsker condition is not satisfied, the sample splitting procedure proposed by \citet{chernozhukov2018double} can be used to ensure that $\text{I}$ is negligible, i.e., use $\mathcal{I}_1$ for  estimation and $\mathcal{I}_2$ for estimation of $\eta_{\gamma, \mathcal{C}}$. See Lemma \ref{lem:bound1} for more details on the bound of $\text{I}$.

        Term $\text{II}$ is the second-order remainder term, which is negligible if the s $m_\mathcal{C}\left( \eta_{\gamma, \mathcal{C}}, X, D \right)$ and $\pi_R(X, D)$ are estimated consistently. 
        \begin{align*}
            \mathbb{E}&\left[ \psi_\mathcal{C}\left( \hat{\eta}_{\gamma, \mathcal{C}}, X; \hat{m}_\mathcal{C}, \hat{\pi}_R \right) \right] - \mathbb{E}\left[ \psi\left( \hat{\eta}_{\gamma, \mathcal{C}}, X; m_\mathcal{C}, \pi_R \right) \right] \\
            =& \mathbb{E}\left[ \mathbb{P}\left( R = 0 \mid X, D \right) \left[ \hat{m}_\mathcal{C}\left( \hat{\eta}_{\gamma, \mathcal{C}}, X, D \right) - m_\mathcal{C}\left( \hat{\eta}_{\gamma, \mathcal{C}}, X, D \right) \right] \right]\\
            &+ \mathbb{E}\left[ \mathbb{P}\left( R = 1 \mid X, D \right) \cdot \frac{1}{\pi_R(X, D)} \left( m_{\mathcal{C}}\left( \hat{\eta}_{\gamma, \mathcal{C}}, X, D \right) - \hat{m}_{\mathcal{C}}\left( \hat{\eta}_{\gamma, \mathcal{C}}, X, D \right) \right) \right]\\
            &+ \mathbb{E}\left[ \mathbb{P}\left( R = 1 \mid X, D \right) \cdot \left( \frac{1}{\hat{\pi}_R(X, D)} - \frac{1}{\pi_R(X, D)} \right) \left( \mathds{1}_{\left\{ V_\mathcal{C} < \hat{\eta}_{\gamma, \mathcal{C}} \right\}} - \hat{m}_\mathcal{C}\left( \hat{\eta}_{\gamma, \mathcal{C}}, X, D \right) \right) \right]\\
            =& \mathbb{E}\left[ \mathbb{P}\left( R = 1 \mid X, D \right) \cdot \left( \frac{1}{\hat{\pi}_R(X, D)} - \frac{1}{\pi_R(X, D)} \right) \left( m_\mathcal{C}\left( \hat{\eta}_{\gamma, \mathcal{C}}, X, D\right) - \hat{m}_\mathcal{C}\left( \hat{\eta}_{\gamma, \mathcal{C}}, X, D \right) \right) \right].
        \end{align*}
        Therefore, 
        \begin{align*}
            \sup_{\eta_{\gamma, \mathcal{C}} \in \mathbb{R}} & \left| \mathbb{E}\left[ \psi_\mathcal{C}\left( \hat{\eta}_{\gamma, \mathcal{C}}, X; \hat{m}_\mathcal{C}, \hat{\pi}_R \right) \right] - \mathbb{E}\left[ \psi\left( \hat{\eta}_{\gamma, \mathcal{C}}, X; m_\mathcal{C}, \pi_R \right) \right] \right| \\
            =&\sup_{\eta_{\gamma, \mathcal{C}} \in \mathbb{R}} \left| \mathbb{E}\left[ \mathbb{P}\left( R = 1 \mid X, D \right) \cdot \left( \frac{1}{\hat{\pi}_R(X, D)} - \frac{1}{\pi_R(X, D)} \right) \left( m_\mathcal{C}\left( \hat{\eta}_{\gamma, \mathcal{C}}, X, D \right) - \hat{m}_\mathcal{C}\left( \hat{\eta}_{\gamma, \mathcal{C}}, X, D \right) \right) \right] \right|\\
            \lesssim & \norm{\hat{\pi}_R(X, D) - \pi_R(X, D)}_2 \cdot \sup_{\eta_{\gamma, \mathcal{C}}} \norm{\hat{m}_\mathcal{C}\left( \hat{\eta}_{\gamma, \mathcal{C}}, X, D \right) - m_\mathcal{C}\left( \hat{\eta}_{\gamma, \mathcal{C}}, X, D \right)}_2,
        \end{align*}
        where the last inequality follows from the Mean Value Theorem and the H\"{o}lder's inequality.

        Combining the bounds for I and II, we have 
        \begin{align*}
            \mathbb{P}\left( V_\mathcal{C} < \hat{\eta}_{\gamma, \mathcal{C}} \mid R = 0 \right) \geq& 1 - \gamma\\
            &- C_0 \frac{\pi_0^{-1} (1 + m_0)}{\mathbb{P}(R = 0)} \sqrt{\frac{\log(1 / \delta) + 1}{|\mathcal{I}_2|}}\\
            &- C_1 \frac{\norm{\hat{\pi}_R(X, D) - \pi_R(X, D)}_2 }{\mathbb{P}(R = 0)} \cdot \sup_{\eta} \norm{\hat{m}_\mathcal{C}\left( \hat{\eta}_{\gamma, \mathcal{C}}, X, D \right) - m_\mathcal{C}\left( \hat{\eta}_{\gamma, \mathcal{C}}, X, D \right)}_2
        \end{align*}
        with probability at least $1 - \delta$.
        
    \end{proof}

    \subsection{Additional Technical Details}

    \subsubsection{Proof of Lemma \ref{lem:cov}}\label{sec:pflemcov}

    Firstly, we can show that
    \begin{equation}\label{eq:lemcov1}
        \begin{aligned}
            \mathbb{P}\left( V_\mathcal{C} < \eta_{\gamma, \mathcal{C}} \mid R = 0 \right) &= \mathbb{E}\left[ \mathbb{E}\left[ V_\mathcal{C} < \eta_{\gamma, \mathcal{C}} \mid X, D, R = 0 \right] \mid R = 0 \right]\\
            \overset{\text{Assumption 2}}&{=} \mathbb{E}\left[ \mathbb{E}\left[ V_\mathcal{C} < \eta_{\gamma, \mathcal{C}} \mid X, D, R = 1 \right] \mid R = 0 \right]\\
            &= \int \int_{-\infty}^{\eta_{\gamma, \mathcal{C}}} f\left( \eta_\mathcal{C} \mid x, d, R = 1 \right)  d\eta_\mathcal{C} dP_{X, D \mid R = 0}(x, d)\\
            &= \int \int_{-\infty}^{\eta_{\gamma, \mathcal{C}}} \frac{f\left( x, d \mid R = 0 \right)}{f\left( x, d \mid R = 1 \right)} f\left( \eta_\mathcal{C} \mid x, d, R = 1 \right)  d\eta_\mathcal{C} dP_{X, D \mid R = 1}(x, d)\\
            &= \mathbb{E}\left[ \frac{f\left( X, D \mid R = 0 \right)}{f\left( X, D \mid R = 1 \right)}\cdot \mathds{1}_{\left\{ V_\mathcal{C} < \eta_{\gamma, \mathcal{C}} \right\}} \mid R = 1 \right].
        \end{aligned}
    \end{equation}
    On the other hand, we can show that 
    \begin{align}\label{eq:lemcov2}
        \mathbb{E}\left[ \psi_\mathcal{C}\left( \eta_{\gamma, \mathcal{C}}, X; m_\mathcal{C}, \pi_R \right) \right] = \mathbb{P}(R = 0) \cdot \left\{ \mathbb{E}\left[ \frac{f(X, D \mid R = 0)}{f(X, D \mid R = 1)} \cdot \mathds{1}_{\left\{ V_\mathcal{C} < \eta_{\gamma, \mathcal{C}} \right\}} \mid R = 1 \right] -  (1 - \gamma) \right\}.
    \end{align}
    Combining Equation \eqref{eq:lemcov1} and \eqref{eq:lemcov2} completes the proof. We prove Equation \eqref{eq:lemcov2} in two steps.
    \begin{align}\label{eq:lemcov2.1}
        \mathbb{E}\left[ \psi_\mathcal{C}\left( \eta_{\gamma, \mathcal{C}}, X; m_\mathcal{C}, \pi_R \right) \right] = \mathbb{E}\left[ \mathbb{P}\left( R = 0 \mid X, D \right) \left( \mathbb{P}\left( V_\mathcal{C} < \eta_{\gamma, \mathcal{C}} \mid X, D, R = 1 \right) - (1 - \gamma) \right) \right], \tag{\ref*{eq:lemcov2}.1} 
    \end{align}
    and
    \begin{equation}\label{eq:lemcov2.2}
        \begin{aligned}
        \mathbb{E}&\left[ \mathbb{P}\left( R = 0 \mid X, D \right) \left( \mathbb{P}\left( V_\mathcal{C} < \eta_{\gamma, \mathcal{C}} \mid X, D, R = 1 \right) - (1 - \gamma) \right) \right] \\
        &= \mathbb{P}(R = 0) \cdot \left\{ \mathbb{E}\left[ \frac{f(X, D \mid R = 0)}{f(X, D \mid R = 1)} \cdot \mathds{1}_{\left\{ V_\mathcal{C} < \eta_{\gamma, \mathcal{C}} \right\}} \mid R = 1 \right] -  (1 - \gamma) \right\} 
    \end{aligned}\tag{\ref*{eq:lemcov2}.2}
    \end{equation}
    
    For Equation \eqref{eq:lemcov2.1}, if $\pi_R(X, D) = \frac{\mathbb{P}\left( R = 1 \mid X, D \right)}{\mathbb{P}\left( R = 0 \mid X, D \right)}$ is correct, then we have
    \begin{align*}
        \mathbb{E}\left[ \frac{R}{\pi_R(X, D)} \mid X, D \right] = \mathbb{P}\left( R = 1 \mid X, D \right) \cdot \frac{\mathbb{P}\left( R = 0 \mid X, D \right)}{\mathbb{P}\left( R = 1 \mid X, D \right)} = \mathbb{P}\left( R = 0 \mid X, D \right).
    \end{align*}
    This implies that 
    \begin{align*}
        \mathbb{E}&\left[ \frac{R}{\pi_R(X, D)} \left[ \mathds{1}_{\left\{ V_\mathcal{C} < \eta_{\gamma, \mathcal{C}} \right\}} - m_\mathcal{C}(\eta_{\gamma, \mathcal{C}}, X, D) \right] \right] \\
        &= \mathbb{E}\left[ \mathbb{P}\left( R = 0 \mid X, D \right) \left[ \mathbb{P}\left( V_\mathcal{C} < \eta_{\gamma, \mathcal{C}}  \mid X, D, R = 1\right) - m_\mathcal{C}\left( \eta_{\gamma, \mathcal{C}}, X, D \right) \right] \right].
    \end{align*}
    Similarly,
    \begin{align*}
        \mathbb{E}\left[ (1 - R) \left[ m_\mathcal{C}\left( \eta_{\gamma, \mathcal{C}}, X, D \right) - (1 - \gamma) \right]\right] = \mathbb{E}\left[ \mathbb{P}\left( R = 0 \mid X, D \right) \left[ m_\mathcal{C}\left( \eta_{\gamma, \mathcal{C}}, X, D \right) - (1 - \gamma) \right] \right].
    \end{align*}
    Hence, if $\pi_R(\cdot)$ is the true density ratio, we have
    \begin{align*}
        \mathbb{E}\left[ \psi_\mathcal{C}\left( \eta_{\gamma, \mathcal{C}}, X; m_\mathcal{C}, \pi_R \right) \right] = \mathbb{E}\left[ \mathbb{P}\left( R = 0 \mid X, D \right) \left( \mathbb{P}\left( V_\mathcal{C} < \eta_{\gamma, \mathcal{C}} \mid X, D, R = 1 \right) - (1 - \gamma) \right) \right].
    \end{align*}
    This competes the proof of Equation \eqref{eq:lemcov2.1} when $\pi_R(\cdot)$ is estimated correctly.

    If the conditional mean $m_\mathcal{C}\left( \eta_{\gamma, \mathcal{C}}, X, D \right) = \mathbb{E}\left[ \mathds{1}_{\left\{ V_\mathcal{C} < \eta_{\gamma, \mathcal{C}} \right\}} \mid X, D, R = 1 \right]$ is correct, then we have 
    \begin{align*}
        \mathbb{E}\left[ \frac{R}{\pi_R(X, D)} \left[ \mathds{1}_{\left\{ V_\mathcal{C} < \eta_{\gamma, \mathcal{C}} \right\}} - m_\mathcal{C}\left( \eta_{\gamma, \mathcal{C}}, X, D \right) \right] \right] = 0
    \end{align*}
    Hence,
    \begin{align*}
        \mathbb{E}\left[ \psi_\mathcal{C}\left( \eta_{\gamma, \mathcal{C}}, X; m_\mathcal{C}, \pi_R \right) \right] &= \mathbb{E}\left[ (1 - R) \left[ m_\mathcal{C}\left( \eta_{\gamma, \mathcal{C}}, X, D \right) - (1 - \gamma) \right] \right]\\
        &= \mathbb{E}\left[ \mathbb{P}\left( R = 0 \mid X, D \right) \left[ \mathbb{P}\left( V_\mathcal{C} < \eta_{\gamma, \mathcal{C}} \mid X, D, R = 1 \right)  - (1 - \gamma)\right] \right].
    \end{align*}
    This completes the proof of Equation \eqref{eq:lemcov2.1} when $m_\mathcal{C}(\cdot)$ is estimated correctly.

    Then, for Equation \eqref{eq:lemcov2.2}, we can show that
    \begin{equation}\label{eq:lemcov}
        \begin{aligned}
            &\mathbb{E}\left[ \mathbb{P}\left( R = 0 \mid X, D \right) \left[ \mathbb{P}\left( V_\mathcal{C} < \eta_{\gamma, \mathcal{C}} \mid X, D, R = 1 \right)  - (1 - \gamma)\right] \right]\\
            =& \mathbb{E}\left[ R \cdot \frac{\mathbb{P}(R = 0 \mid X, D)}{\mathbb{P}(R = 1 \mid X, D)} \cdot \mathbb{P}\left( V_\mathcal{C} < \eta_{\gamma, \mathcal{C}} \mid X, D, R = 1 \right) \right] - \mathbb{P}(R = 0) \cdot (1 - \gamma)\\
            =& \frac{\mathbb{P}(R = 0)}{\mathbb{P}(R = 1)} \mathbb{E}\left[ R \cdot \frac{f(X, D \mid R = 0)}{f(X, D \mid R = 1)} \cdot \mathds{1}_{\left\{ V_\mathcal{C} < \eta_{\gamma, \mathcal{C}} \right\}} \right] - \mathbb{P}(R = 0) \cdot (1 - \gamma)\\
            %=& \frac{\mathbb{P}(R = 0)}{\mathbb{P}(R = 1)} \mathbb{E}_R \left[ \mathbb{E}\left[ R \cdot \frac{f(X, D \mid R = 0)}{f(X, D \mid R = 1)} \cdot \mathds{1}_{\left\{ V_\mathcal{C} < \eta_{\gamma, \mathcal{C}} \right\}} \mid X, D, R \right] \right] - \mathbb{P}(R = 0) \cdot (1 - \gamma)\\
            =& \frac{\mathbb{P}(R = 0)}{\mathbb{P}(R = 1)} \cdot \mathbb{P}(R = 1) \cdot \mathbb{E}\left[ \frac{f(X, D \mid R = 0)}{f(X, D \mid R = 1)} \cdot \mathds{1}_{\left\{ V_\mathcal{C} < \eta_{\gamma, \mathcal{C}} \right\}} \mid R = 1 \right] - \mathbb{P}(R = 0) \cdot (1 - \gamma)\\
            =& \mathbb{P}(R = 0) \cdot \mathbb{E}\left[ \frac{f(X, D \mid R = 0)}{f(X, D \mid R = 1)} \cdot \mathds{1}_{\left\{ V_\mathcal{C} < \eta_{\gamma, \mathcal{C}} \right\}} \mid R = 1 \right] - \mathbb{P}(R = 0) \cdot (1 - \gamma).
        \end{aligned}
    \end{equation}
    This completes the proof of Equation \eqref{eq:lemcov2.2}. 

    Thus, we have
    \begin{align*}
        \mathbb{P}\left( V_\mathcal{C} <\eta_{\gamma, \mathcal{C}} \mid R = 0 \right) = 1 - \gamma + \frac{\mathbb{E}\left[ \psi_\mathcal{C}\left( \eta_{\gamma, \mathcal{C}}, X; m_\mathcal{C}, \pi_R \right) \right]}{\mathbb{P}(R = 0)}.
    \end{align*}

    \subsubsection{Proof of Lemma \ref{lem:bound1}}

    \begin{lemma}\label{lem:bound1}
        Under regularity conditions (A1) and (A2), there exists a universal constant $C_0$ such that for any $\delta > 0$,
        \begin{align*}
            \mathbb{P}\left( | \mathrm{I} | \leq C_0 \frac{\pi_0^{-1}(1 + m_0)}{\mathbb{P}(R = 0)} \sqrt{\frac{\log(1 / \delta) + 1}{|\mathcal{I}_2|}} \mid \mathcal{I}_1\right) \geq 1 - \delta.
        \end{align*}
    \end{lemma}

    \begin{proof}[Proof of Lemma \ref{lem:bound1}]
        The proof is adapted from Theorem 3 of \citet{yang2024doubly} and Lemma A.1 from \citet{gao2025role}. Without loss of generality, assume the indexes in $\mathcal{I}_2$ is $1, \dots, N$, with $N \coloneq |\mathcal{I}_2|$. We can expand $\mathbb{P}_{\mathcal{I}_2} \left[ \psi_\mathcal{C}\left( \hat{\eta}_{\gamma, \mathcal{C}}, X; \hat{m}_\mathcal{C}, \hat{\pi}_R \right) \right] - \mathbb{E}\left[ \psi_\mathcal{C}\left( \hat{\eta}_{\gamma, \mathcal{C}}, X; \hat{m}_\mathcal{C}, \hat{\pi}_R \right) \right]$ into three parts.
        \begin{align*}
            \mathbb{P}_{\mathcal{I}_2} &\left[ \psi_\mathcal{C}\left( \hat{\eta}_{\gamma, \mathcal{C}}, X; \hat{m}_\mathcal{C}, \hat{\pi}_R \right) \right] - \mathbb{E}\left[ \psi_\mathcal{C}\left( \hat{\eta}_{\gamma, \mathcal{C}}, X; \hat{m}_\mathcal{C}, \hat{\pi}_R \right) \right] \\
            &= \frac{1}{N} \sum_{i = 1}^{N} \frac{R}{\hat{\pi}_R(X_i, D_i)} \mathds{1}_{\left\{ V_{\mathcal{C}, i} < \eta_{\gamma, \mathcal{C}} \right\}} - \mathbb{E}\left[ \frac{R}{\hat{\pi}_R(X_i, D_i)}\mathds{1}_{\left\{ V_{\mathcal{C}, i} < \eta_{\gamma, \mathcal{C}} \right\}} \right]\\
            &+ \frac{1}{N} \sum_{i = 1}^{N} \left( (1 - R) - \frac{R}{\hat{\pi}_R(X_i, D_i)} \hat{m}_\mathcal{C}\left(\eta_{\gamma, \mathcal{C}}, X_i, D_i \right) \right) - \mathbb{E}\left[ (1 - R) - \frac{R}{\hat{\pi}_R(X_i, D_i)} \hat{m}_\mathcal{C}\left(\eta_{\gamma, \mathcal{C}}, X_i, D_i \right) \right]\\
            &+ (1 - \gamma) \left[ \frac{1}{N} \sum_{i = 1}^{N} (1 - R) - \mathbb{P}(R = 0) \right]\\
            & \eqcolon \mathcal{R}_1(\eta_{\gamma, \mathcal{C}}) + \mathcal{R}_2(\eta_{\gamma, \mathcal{C}}) + \mathcal{R}_3(\eta_{\gamma, \mathcal{C}}),
        \end{align*}
        where these three terms will be controlled separately. For $\mathcal{O}_i = \left( X_i, Y_i, D_i, R_i \right), i \in \mathcal{I}_{2}$ and any function $f: \mathbb{R}^{d + 3} \to \mathbb{R}$, define
        \begin{align*}
            \mathbb{G}_Nf = \frac{1}{\sqrt{N}} \sum_{i = 1}^{N} \left[ f(\mathcal{O}_i) - \mathbb{E}\left[ f(\mathcal{O}_i) \right] \right]
        \end{align*}

        \noindent\textbf{Bound on} $\sup_{\eta_{\gamma, \mathcal{C}}} |\mathcal{R}_1(\eta_{\gamma, \mathcal{C}})|$: We define a class of functions $\mathcal{F}_1$:
        \begin{align*}
            \mathcal{F} \coloneq \left\{ f: f_\eta = \frac{R}{\pi_R(X, D)} \mathds{1}_{\left\{ V_\mathcal{C} < \eta_{\gamma, \mathcal{C}} \right\}}, \forall \eta \in \mathbb{R} \right\}.
        \end{align*}
        Notice that $\forall f_\mathcal{\eta} \in \mathcal{F}$, we have $|f_\mathcal{\eta}| \leq R\pi_0^{-1} \mathds{1}_{\left\{ V_\mathcal{C} < \eta_{\gamma, \mathcal{C}} \right\}}$. Therefore, $F(w) \coloneq R \pi_0^{-1}$ is an envelope function of $\mathcal{F}$. Let $\norm{\cdot}_{\mathcal{F}}$ denote the supremum norm over the class $\mathcal{F}$, i.e., $\norm{z}_\mathcal{F} = \sup_{f \in \mathcal{F}} |z(f)|$. 

        Applying Lemma \ref{lem:unibound} with $s(r, x) = \frac{R}{\hat{\pi}_R(X, D)}$ and $h(x, y) = V_\mathcal{C}$ gives 
        \begin{align*}
            \mathbb{E}\norm{\mathbb{G}_N}_\mathcal{F} \leq \mathfrak{C}_1\pi_0^{-1},
        \end{align*}
        where $\mathfrak{C}_1$ is a universal constant. Applying McDiarmid's inequality, we have 
        \begin{align}
            \mathbb{P}\left( \norm{\mathbb{G}_N}_\mathcal{F} - \mathbb{E}\norm{\mathbb{G}_N}_\mathcal{F} \geq t \right) \leq 2 \exp\left( -\frac{2t^2}{\sum_{i = 1}^{N} c_i^2} \right) \leq 2 \exp\left( -\frac{2t^2}{\sum_{i = 1}^{N}4\pi_0^{-2} / N} \right) = \exp\left( -\frac{t^2}{2\pi_0^{-2}} \right), \label{eq:bound1}
        \end{align}
        where 
        \begin{align*}
            c_i \leq \sup_{\mathcal{O}, \mathcal{O}'} \sup_{\eta} \sqrt{N}\left|\frac{1}{N}\frac{R_i}{\hat{\pi}_R(X_i, D_i)} \mathds{1}_{\left\{ V_{\mathcal{C}, i} < \eta \right\}} - \frac{1}{N}\frac{R_i'}{\hat{\pi}_R(X_i', D_i')} \mathds{1}_{\left\{ V_{\mathcal{C}, i'} < \eta \right\}}\right| \leq \frac{2\pi_0^{-1}}{\sqrt{N}}
        \end{align*}
        Substituting into Equation \eqref{eq:bound1}, we have 
        \begin{align}
            \mathbb{P}\left( \norm{\mathbb{G}_N}_{\mathcal{F}} \geq C_2 \pi_0^{-1} \sqrt{1 + \log\left( \frac{1}{\delta} \right)} \right) \leq \delta, \label{eq:boundR1}
        \end{align}
        for an absolute constant $C_2$.

        \noindent\textbf{Bound on} $\sup_{\eta_{\gamma, \mathcal{C}}}|\mathcal{R}_2(\eta_{\gamma, \mathcal{C}})|$: similar to the above, we define a class of functions $\mathcal{F}$: $\mathcal{F} \coloneq \left\{ f_\eta:  f_\eta = \left[ (1 - R) - \frac{R}{\hat{\pi}_R(X, D)} \right] \hat{m}_{\mathcal{C}}\left( \eta_{\gamma, \mathcal{C}}, X, D \right), \forall \eta \in \mathbb{R} \right\}$, where
        \begin{align*}
            f_\eta &= \left[ (1 - R) - \frac{R}{\hat{\pi}_R(X, D)} \right] \cdot \hat{m}_{\mathcal{C}}\left( \eta_{\gamma, \mathcal{C}}, X, D \right)\\
            &= \left[ (1 - R) - \frac{R}{\hat{\pi}_R(X, D)} \right] \int_{0}^{m_0} \mathds{1}_{\left\{ \hat{m}_\mathcal{C}\left( \eta_{\gamma, \mathcal{C}}, X, D \right) \geq u \right\}} du\\
            &= \int_{0}^{m_0} \left[ (1 - R) - \frac{R}{\hat{\pi}_R(X, D)} \right] \mathds{1}_{\left\{ \eta_{\gamma, \mathcal{C}} \geq h(X, u) \right\}} du,
        \end{align*}
        where the second equality is from the monotonicity of $\hat{m}_\mathcal{C}\left( \eta_{\gamma, \mathcal{C}}, X, D \right)$ in $\eta_{\gamma, \mathcal{C}}$. Then,
        \begin{align*}
            \sup_{\eta} |\mathbb{G}_N f| &= \sup_{\eta} \left| \mathbb{G}_N \left[ \int_{0}^{m_0} \left[ (1 - R) - \frac{R}{\hat{\pi}_R(X, D)} \right] \mathds{1}_{\left\{ \eta_{\gamma, \mathcal{C}} \geq h(X, u) \right\}} du \right] \right|\\
            &\leq \int_{0}^{m_0} \sup_{\eta} \left| \mathbb{G}_N \left[  \left[ (1 - R) - \frac{R}{\hat{\pi}_R(X, D)} \right] \mathds{1}_{\left\{ \eta_{\gamma, \mathcal{C}} \geq h(X, u) \right\}} \right] \right| du\\
        \end{align*}
        Therefore, taking expectation on both sides and applying Lemma \ref{lem:unibound}, we have
        \begin{align*}
            \mathbb{E}\norm{\mathbb{G}_N}_\mathcal{F} \lesssim \int_{0}^{m_0} \mathfrak{C}_2 \pi_0^{-1} du \lesssim \mathfrak{C}_2 m_0 \pi_0^{-1}.
        \end{align*}
        By McDiarmid's inequality, we have
        \begin{align}
            \mathbb{P}\left( \norm{\mathbb{G}_N}_\mathcal{F} - \mathbb{E}\norm{\mathbb{G}_N}_\mathcal{F} \geq t \right) \leq \exp\left( -\frac{2t^2}{\sum_{i = 1}^{N} c_i^2} \right) \leq \exp\left( -\frac{2t^2}{\sum_{i = 1}^{N} 4\pi_0^{-2} m_0^2 / N} \right) = \exp\left( \frac{-t^2}{2\pi_0^{-2} m_0^2} \right), \label{eq:bound2}
        \end{align}
        where 
        \begin{align*}
            c_i &\coloneq \sup_{\mathcal{O}, \mathcal{O}'} \sup_{\eta} \sqrt{N}\left| \frac{1}{N} \left[ (1 - R_i) - \frac{R_i}{\hat{\pi}_R(X_i, D_i)} \right] \hat{m}_\mathcal{C}\left( \eta, X_i, D_i \right) - \frac{1}{N} \left[ (1 - R_i') - \frac{R_i'}{\hat{\pi}_R(X_i', D_i')} \right] \hat{m}_\mathcal{C}\left( \eta, X_i', D_i' \right) \right|\\
            &\leq \frac{2m_0\pi_0^{-1}}{\sqrt{N}}.
        \end{align*}
        Substituting into Equation \eqref{eq:bound2}, we have
        \begin{align}
            \mathbb{P}\left( \norm{\mathbb{G}_N}_\mathcal{F} \geq C_3 m_0 \pi_0^{-1} \sqrt{ 1 + \log\left( \frac{1}{\delta} \right)} \right) \leq \delta, \label{eq:boundR2}
        \end{align}
        for an absolute constant $C_3$.

        \noindent\textbf{Bound on} $\sup_{\eta_{\gamma, \mathcal{C}}} |\mathcal{R}_3(\eta_{\gamma, \mathcal{C}})|$: Since the random variables $R_i$ is i.i.d, applying Hoeffding's inequality gives
        \begin{align*}
            \mathbb{P}\left( \frac{1}{N}\sum_{i = 1}^{N} (1 - R_i) - \mathbb{P}\left( R_i = 0 \right) \geq t \right) \leq \exp\left( -\frac{2t^2}{N} \right),
        \end{align*}
        which implies that
        \begin{align}
            \mathbb{P}\left( \mathcal{R}_3 \geq (1 - \gamma) \sqrt{\frac{1}{2N} \log\left( \frac{1}{\delta} \right)} \right) \leq \delta. \label{eq:boundR3}
        \end{align}

        Combining the bounds in Equations \eqref{eq:boundR1}, \eqref{eq:boundR2} and \eqref{eq:boundR3} using the union bound, we have that for $\delta$ > 0, there exists a universal constant $C_0$ such that
        \begin{align*}
            \sup_\eta |\mathcal{R}_1\left( \eta_{\gamma, \mathcal{C}} \right) + \mathcal{R}_2\left( \eta_{\gamma, \mathcal{C}} \right) + \mathcal{R}_3\left( \eta_{\gamma, \mathcal{C}} \right)| \lesssim C_0 \pi_0^{-1} (1 + m_0) \sqrt{\frac{\log(1 / \delta) + 1}{|\mathcal{I}_2|}}
        \end{align*}
        with probability at least $1 - \delta$.

    \end{proof}

    \subsubsection{Additional Lemmas}

    \begin{lemma}[Lemma 8 of \citet{yang2024doubly}]\label{lem:unibound}
        There exists a universal constant $\mathfrak{C} < \infty$ such that for any functions $s(t, x) \in [-\kappa_0. \kappa_0]$ and $h(x, y)$.
        \[\mathbb{E}\left[ \sup_{\eta} |\mathbb{G}_N\left[ s(r, x) \mathds{1}_{\left\{ h(x, y) \leq \eta \right\}} \right]| \right] \leq \mathfrak{C}\kappa_0.\] 
    \end{lemma}

    \clearpage

    \section{Algorithms}

    \subsection{Interval Estimates for ITE with Attrition by Semiparametric Efficient Estimator}

    \begin{algorithm}[!ht]
        \caption{Interval Estimates for ITE with Attrition by Semiparametric Efficient Estimator}
        \label{alg:iteattreif}
        \begin{algorithmic}
            \Require Level $\alpha$, level $\gamma$, data $\mathcal{Z} \equiv \left( X_i, Y_i, D_i, R_i \right)_{i = 1}^n$, function $\hat{q}_{\beta}(x; \mathcal{D})$ to fit $\beta$-th conditional quantile, functions $\hat{\pi}_D(x; \mathcal{D}), \hat{e}_R(x, d; \mathcal{D})$, and $\hat{\pi}_R(x, d; \mathcal{D})$ to fit the propensity score at $x$ (and $d$), function $\hat{m}(t, W; \mathcal{D})$ to fit the conditional CDF at $t$, $\hat{h}^{\text{L}}(x; \mathcal{D}), \hat{h}^{\text{R}}(x; \mathcal{D})$ to fit the conditional mean using $\mathcal{D}$ as data, efficient influence functions $\psi_d$ and $\psi_\mathcal{C}$ to identify the threshold of nonconformity score.
        \end{algorithmic}
        \textbf{Step I}: Data splitting
        \begin{algorithmic}[1]
            \State Split the data into three folds: a pretraining fold $\mathcal{Z}_{\text{pr}}$, a training fold $\mathcal{Z}_{\text{tr}}$, and a calibration fold $\mathcal{Z}_{\text{ca}}$. 
            \State Split the training fold into two subfolds $\mathcal{Z}_{\text{tr}, 1}$ and $\mathcal{Z}_{\text{tr}, 2}$.
        \end{algorithmic}
        \textbf{Step II}: Conformal Inference for Counterfactuals
        \begin{algorithmic}[1]
            \State For each $i \in \mathcal{I}_{\text{tr}} \cup \mathcal{I}_{\text{ca}}$ with $R_i = 1$, compute the nonconformity score following CQR: $V_i = \max\left\{ \hat{q}_{\alpha_{\text{lo}}}(X_i; \mathcal{Z}_{\text{pr}}) - Y_i, Y_i - \hat{q}_{\alpha_{\text{hi}}}(X_i; \mathcal{Z}_{\text{pr}}) \right\}$.
            \State For each $i \in \mathcal{I}_{\text{tr}, 1}$, obtain an initial estimator of the threshold $\hat{\eta}_{\alpha, d}^{\text{init}}$.
            \State For each $i \in \mathcal{I}_{\text{tr}, 2}$, train the conditional CDF $\hat{m}(\hat{\eta}_{\alpha, d}^{\text{init}}, X)$.
            \For{$i \in \mathcal{I}_\text{ca}$ with $D_i = 1$}
                \State Identify $\hat{\eta}_{\alpha, 0}$ by solving the eif $\psi_0$ using $\hat{\pi}_D(X_i; \mathcal{Z}_{\text{pr}}), \hat{e}_R(X_i, D_i = 1; \mathcal{Z}_{\text{pr}}),$ and $\hat{m}(\hat{\eta}_{\alpha, 1}^{\text{init}}, X_i, R_i = 1, D_i = 1; \mathcal{Z}_{\text{tr, 2}})$ as s.
                \State Construct $\mathcal{C}_0(x) = [\hat{Y}_i^{\text{L}}(0), \hat{Y}_i^{\text{R}}(0)] = [\hat{q}_{\alpha_{\text{lo}}}(X_i; \mathcal{Z}_{\text{pr}}) - \hat{\eta}_{\alpha, 0}, \hat{\eta}_{\alpha, 0} - \hat{q}_{\alpha_{\text{hi}}}(X_i; \mathcal{Z}_{\text{pr}})]$.
                \State Compute $\mathcal{C}_{\text{ITE}} = [Y_i(1) - \hat{Y}_i^{\text{R}}(0), Y_i(1) - \hat{Y}_i^{\text{R}}(0)]$.
            \EndFor
            \For{$i \in \mathcal{I}_\text{ca}$ with $D_i = 0$}
                \State Identify $\hat{\eta}_{\alpha, 1}$ by solving the eif $\psi_1$ using $\hat{\pi}_D(X_i; \mathcal{Z}_{\text{pr}}), \hat{e}_R(X_i, D_i = 0; \mathcal{Z}_{\text{pr}}),$ and $\hat{m}(\hat{\eta}_{\alpha, 0}^{\text{init}}, X_i, R_i = 1, D_i = 0; \mathcal{Z}_{\text{tr, 2}})$ as s.
                \State Construct $\mathcal{C}_1(x) = [\hat{Y}_i^{\text{L}}(1), \hat{Y}_i^{\text{R}}(1)] = [\hat{q}_{\alpha_{\text{lo}}}(X_i; \mathcal{Z}_{\text{pr}}) - \hat{\eta}_{\alpha, 1}, \hat{\eta}_{\alpha, 1} - \hat{q}_{\alpha_{\text{hi}}}(X_i; \mathcal{Z}_{\text{pr}})]$.
                \State Compute $\mathcal{C}_{\text{ITE}} = [\hat{Y}_i^{\text{L}}(1) - Y_i(0), \hat{Y}_i^{\text{R}} - Y_i(0)]$.
            \EndFor        
        \end{algorithmic}
        \textbf{Step III}: Interval estimates of ITE on the target data with attrition.
        \begin{algorithmic}[1]
            \State Take $\mathcal{Z}_{\text{ca}} \equiv (X_i, Y_i, D_i, R_i, \mathcal{C}_i)$ with $R_i = 1$ and $\mathcal{C}_i = [\mathcal{C}_i^{\text{L}}, \mathcal{C}_i^{\text{R}}]$ as the observed data. Take $\mathcal{Z}_{\text{att}}$ from $\mathcal{Z}$ with $R_i = 0$ as the attrition data.
            \State Split $\mathcal{Z}_{\text{ca}} \equiv \mathcal{Z}_{\text{obs}}$ into two folds $\mathcal{Z}_{\text{obstr}}$ and $\mathcal{Z}_{\text{obsca}}$.
            \State For each $i \in \mathcal{I}_{\text{obs}}$, compute score $V_i = \max\left\{ \hat{h}^{\text{L}}(X_i; \mathcal{Z}_{\text{obstr}}) - \mathcal{C}_i^{\text{L}}, \mathcal{C}_i^{\text{R}} - \hat{h}^{\text{R}}(X_i, \mathcal{Z}_{\text{obstr}}) \right\}$.
            \State For each $i \in \mathcal{I}_{\text{obsca}} \cup \mathcal{I}_{\text{att}}$, identify $\hat{\eta}_{\gamma, \mathcal{C}}$ by solving the eif $\psi_\mathcal{C}$ using $\hat{\pi}_R(X_i, D_i; \mathcal{Z}_{\text{obstr}})$, and $\hat{m}_{\mathcal{C}}(\hat{\eta}_{\gamma, \mathcal{C}}, X_i, D_i; \mathcal{Z}_{\text{obstr}})$ as s.
        \end{algorithmic}
        \begin{algorithmic}
            \Ensure $\check{\mathcal{C}}_{\text{ITE}}(x) = \left[\hat{h}^{\text{L}}(x; \mathcal{Z}_{\text{obstr}}) - \hat{\eta}_{\gamma, \mathcal{C}}, \hat{h}^{\text{R}}(x; \mathcal{Z}_{\text{obstr}}) + \hat{\eta}_{\gamma, \mathcal{C}}\right]$.
        \end{algorithmic}
    \end{algorithm}
    \clearpage

    \subsection{Weighted Split Conformalized Quantile Regression}

    \begin{algorithm}
        \caption{Weighted Split CQR}
        \label{alg:wcqr}
        \begin{algorithmic}
            \Require Level $\alpha$, data $\mathcal{D} \equiv \left( X_i, Y_i \right)_{i \in \mathcal{I}}$, testing point $x$, function $\hat{q}_\tau(x; \mathcal{D})$ to estimate the conditional quantile $\tau$ and function $w(x; \mathcal{D})$ to estimate the weight at $x$ using $\mathcal{D}$ as the data.
        \end{algorithmic}
        \textbf{Procedure}:
        \begin{algorithmic}[1]
            \State Split the data $\mathcal{D}$ into a training set $\mathcal{D}_{\text{tr}} \equiv \left( X_i, Y_i \right)_{i \in \mathcal{I}_{\text{tr}}}$ and a calibration set $\mathcal{D}_{\text{ca}} \equiv \left( X_i, Y_i \right)_{i \in \mathcal{I}_{\text{ca}}}$.
            \State For each $i \in \mathcal{I}_{\text{ca}}$, compute the score $S_i = \max\left\{ \hat{q}_{\alpha_{\text{lo}}}(X_i; \mathcal{D}_{\text{tr}}) - Y_i, Y_i - \hat{q}_{\alpha_{\text{hi}}}(X_i; \mathcal{D}_{\text{tr}}) \right\}$.
            \State For each $i \in \mathcal{I}_{\text{ca}}$, compute the weight $W_i = \hat{w}\left( X_i, \mathcal{D}_{\text{tr}} \right)$.
            \State Compute the normalized weights $\hat{p}_i(x) = \frac{W_i}{\sum_{i \in \mathcal{I}_{\text{ca}}} W_i + \hat{w}(x; \mathcal{D}_{\text{tr}})}$ and $\hat{p}_\infty(x) = \frac{\hat{w}(x; \mathcal{D}_{\text{tr}})}{\sum_{i \in \mathcal{I}_{\text{ca}}} W_i + \hat{w}(x; \mathcal{D}_{\text{tr}})}$.
            \State Compute $\eta(x)$ as the $\ceil{(1 - \alpha) \frac{|\mathcal{I}_{\text{ca}}| + 1}{|\mathcal{I}_{\text{ca}}|}}$-th quantile of the distribution $\sum_{i \in \mathcal{I}_{\text{ca}}} \hat{p}_i(x) \delta_{S_i} + \hat{p}_\infty(x) \delta_\infty$.
        \end{algorithmic}
        \begin{algorithmic}
            \Ensure Prediction interval $\mathcal{C}(x) = \left[ \hat{q}_{\alpha_{\text{lo}}}(x; \mathcal{D}_{\text{tr}}) - \eta(x), \hat{q}_{\alpha_{\text{hi}}}(x; \mathcal{D}_{\text{tr}}) + \eta(x) \right]$.
        \end{algorithmic}
    \end{algorithm}

    \clearpage

    \subsection{Unweighted Conformal inference for Interval Outcomes}

    \begin{algorithm}
        \caption{Unweighted Conformal Inference for Interval Outcomes}
        \label{alg:uci}
        \begin{algorithmic}[1]
            \Require {Level $\gamma$, data $\mathcal{D} \equiv (X_i, \mathcal{C}_i)_{i \in \mathcal{I}}$ where $\mathcal{C}_i = \left[ \mathcal{C}_i^{\text{L}}, \mathcal{C}_i^{\text{R}} \right]$, testing point $x$, functions $\hat{m}^{\text{L}}(x; \mathcal{D})$ and $\hat{m}^{\text{R}}(x; \mathcal{D})$ to fit the conditional mean/median of $\mathcal{C}_i^{\text{L}}, \mathcal{C}_i^{\text{R}}$ using $\mathcal{D}$ as the data.}
        \end{algorithmic}
        \textbf{Procedure}:
        \begin{algorithmic}[1]
            \State {Split the data $\mathcal{D}$ into a training set $\mathcal{D}_{\text{tr}} \equiv (X_i, \mathcal{C}_i)_{i \in \mathcal{I}_{\text{tr}}}$ and a calibration set $\mathcal{D}_{\text{ca}} \equiv (X_i, \mathcal{C}_i)_{i \in \mathcal{I}_{\text{ca}}}$.}
            \State {For each $i \in \mathcal{I}_\text{ca}$, compute the score $S_i =  \max\left\{ \hat{m}^{\text{L}}(x; \mathcal{D}_{\text{tr}}) - \mathcal{C}_i^{\text{L}}, \mathcal{C}_i^{\text{R}} - \hat{m}^{\text{R}}(x; \mathcal{D}_{\text{tr}}) \right\}$.}
            \State {Compute $\eta$ as the $\ceil{(1 - \gamma) \frac{|\mathcal{I}_{\text{ca}}| + 1}{|\mathcal{I}_{\text{ca}}|}}$-th quantile of the empirical distribution $\left\{ S_i: i \in \mathcal{I}_{\text{ca}} \right\}$.}
        \end{algorithmic}
        \begin{algorithmic}
            \Ensure {$\check{\mathcal{C}}_{\text{ITE}} = \left[ \hat{m}^{\text{L}}(x; \mathcal{D}_{\text{tr}}) - \eta, \hat{m}^{\text{R}}(x; \mathcal{D}_{\text{tr}}) + \eta \right]$.}
        \end{algorithmic}
    \end{algorithm}

    \clearpage

    \subsection{Interval Estimates for ITE with Attrition}

    \begin{algorithm}
        \caption{Interval Estimates for ITE with Attrition}
        \label{alg:iteattr}
        \begin{algorithmic}
            \Require Level $\alpha$, level $\gamma$ (only for the exact version), source data $\mathcal{Z} \equiv (X_i, Y_i, D_i)_{i = 1}^n$ without attrition, target data point $x$ with attrition.
        \end{algorithmic}
        \textbf{Step I}: Data splitting
        \begin{algorithmic}[1]
            \State Split the source data into two folds $\mathcal{Z}_1$ and $\mathcal{Z}_2$.
            \State Estimate the propensity score $\hat{e}_D(x)$ on $\mathcal{Z}_1$.
        \end{algorithmic}
        \textbf{Step II}: Counterfactual Inference on $\mathcal{Z}_2$.
        \begin{algorithmic}[1]
            \For {$i$ in $\mathcal{Z}_2$ with $D_i = 1$}
                \State Compute $\left[ \hat{Y}_i^{\text{L}}(0), \hat{Y}_i^{\text{R}}(0) \right]$ in Algorithm \ref{alg:wcqr} on $\mathcal{Z}_1$ with level $\alpha$ and $w_0(x) = \frac{\hat{e}_D(x)}{1 - \hat{e}_D(x)}$.
                \State Compute $\mathcal{C}_i = \left[ Y_i(1) - \hat{Y}_i^{\text{R}}(0), Y_i(1) - \hat{Y}_i^{\text{L}}(0) \right]$.
            \EndFor
            \For {$i$ in $\mathcal{Z}_2$ with $D_i = 0$}
                \State Compute $\left[ \hat{Y}_i^{\text{L}}(1), \hat{Y}_i^{\text{R}}(1) \right]$ in Algorithm \ref{alg:wcqr} on $\mathcal{Z}_1$ with level $\alpha$ and $w_1(x) = \frac{1 - \hat{e}_D(x)}{\hat{e}_D(x)}$.
                \State Compute $\mathcal{C}_i = \left[ \hat{Y}_i^{\text{L}}(1) - Y_i(0), \hat{Y}_i^{\text{R}}(1) - Y_i(0) \right]$.
            \EndFor
        \end{algorithmic}
        \textbf{Step III}: Interval estimates of ITE on the target data with attrition.
        \begin{algorithmic}[1]
            \State (Exact version) Apply Algorithm \ref{alg:uci} on $(X_i, \mathcal{C}_i)_{i \in \mathcal{Z}_2}$ with level $\gamma$, yielding an interval $\check{\mathcal{C}}_{\text{ITE}}(x)$.
            \State (Inexact version) Fit conditional quantiles of $\mathcal{C}_i^{\text{L}}$ and $\mathcal{C}_i^{\text{R}}$, yielding an interval $\check{\mathcal{C}}_{\text{ITE}}(x)$.
        \end{algorithmic}
        \begin{algorithmic}
            \Ensure $\check{\mathcal{C}}_{\text{ITE}}(x)$.
        \end{algorithmic}
    \end{algorithm}

    \clearpage

    \section{Additional Simulation Results}

    \subsection{Simulation Results for Unweighted Conformal Inference for ITE with Attrition}

    This section provides the simulation results for the unweighted conformal inference for interval outcomes with attrition proposed by \citep{lei2021conformala}. 
    
    We use the exact nested method and consider two scenarios with different missingness patterns induced by the attrition: MCAR and MAR. The data generating process (DGP) is as follows: covariate vector $X_i = \left( X_{i1}, \dots, X_{id} \right)^\top \sim \mathcal{N}(0, \mathbb{I}_d)$ with $d = 5$, independent across $i$. Given a covariate vector $X_i$, I generate the potential outcomes as
    \begin{align*}
        Y_i(1) = X_i^\top \beta + \varepsilon_{i1}, \quad Y_i(0) = \varepsilon_{i0}, \quad \varepsilon_{i1}, \varepsilon_{i0} \overset{\text{i.i.d.}}{\sim} \mathcal{N}(0, 1),
    \end{align*}
    where $\beta = \mathbf{1}_d$. The treatment assignment is generated following nonlinear propensity score,
    \begin{align*}
        e_D(X_i) = \Phi(X_{i1}), \quad D_i = \mathds{1}_{e_D(X_i) < U_i}, \quad U_i \sim \mathrm{UNIF}(0, 1).
    \end{align*}
    Also, I generate the missingness indicator $R$ separately for two missingness patterns. For MCAR, $R_i \overset{\mathrm{i.i.d.}}{\sim} \mathrm{BERN}(p)$ with $p = 0.8$. For MAR, following Assumption \ref{asm2}, $R$ is generated as
    \begin{align*}
        \mathrm{logit}\left( P(R_i = 1 \mid D_i, X_i) \right) = \kappa_0 + \kappa_1 D_i + \kappa_2 X_{i1} + \kappa_3 X_{i2},
    \end{align*}
    with parameters set to $(\kappa_0, \kappa_1, \kappa_2, \kappa_3) = (-0.2, 0.5, 0.2, -0.3)$. Following this DGP, the observations are given by the tuple $(X_i, D_i, R_i, R_i \cdot Y_i)$.

    I conduct simulations with sample size $n \in \left\{ 500, 1000, 2000, 5000 \right\}$ and 100 MC iterations. I compute the average empirical coverage and average interval length of the prediction intervals for ITE in the target group with attrition. The target coverage is set to be $95\%$. To fit the conditional quantiles, I use the quantile random forest.

    \begin{figure}[htp]
        \centering
        \caption{Unweighted Conformal Inference for Interval Estimates of ITE with Attrition}
        \label{fig:unwMC}
        \begin{threeparttable}
            \begin{minipage}{0.49\textwidth}
                \includegraphics[width = \textwidth]{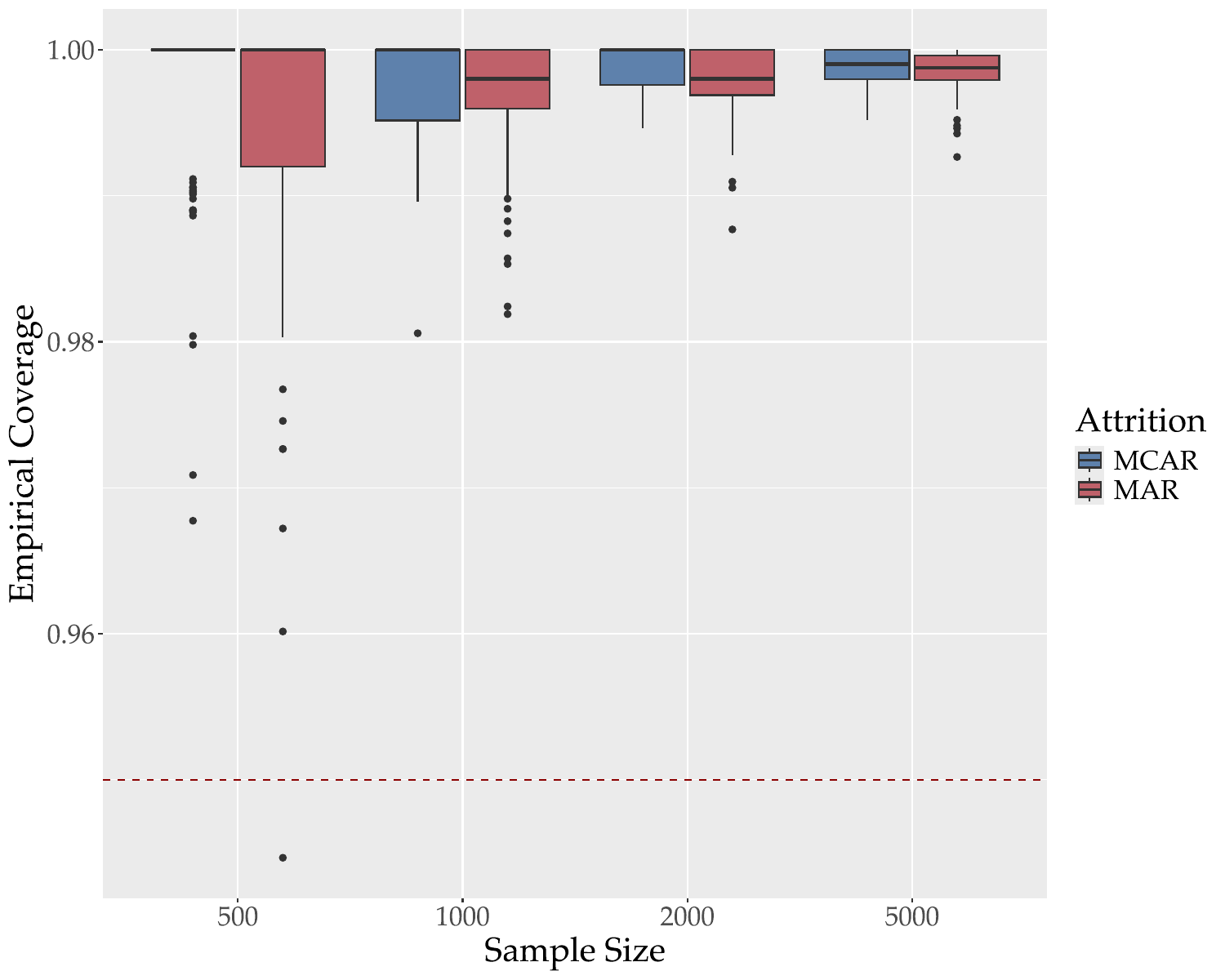}
            \end{minipage}
            \begin{minipage}{0.49\textwidth}
                \includegraphics[width = \textwidth]{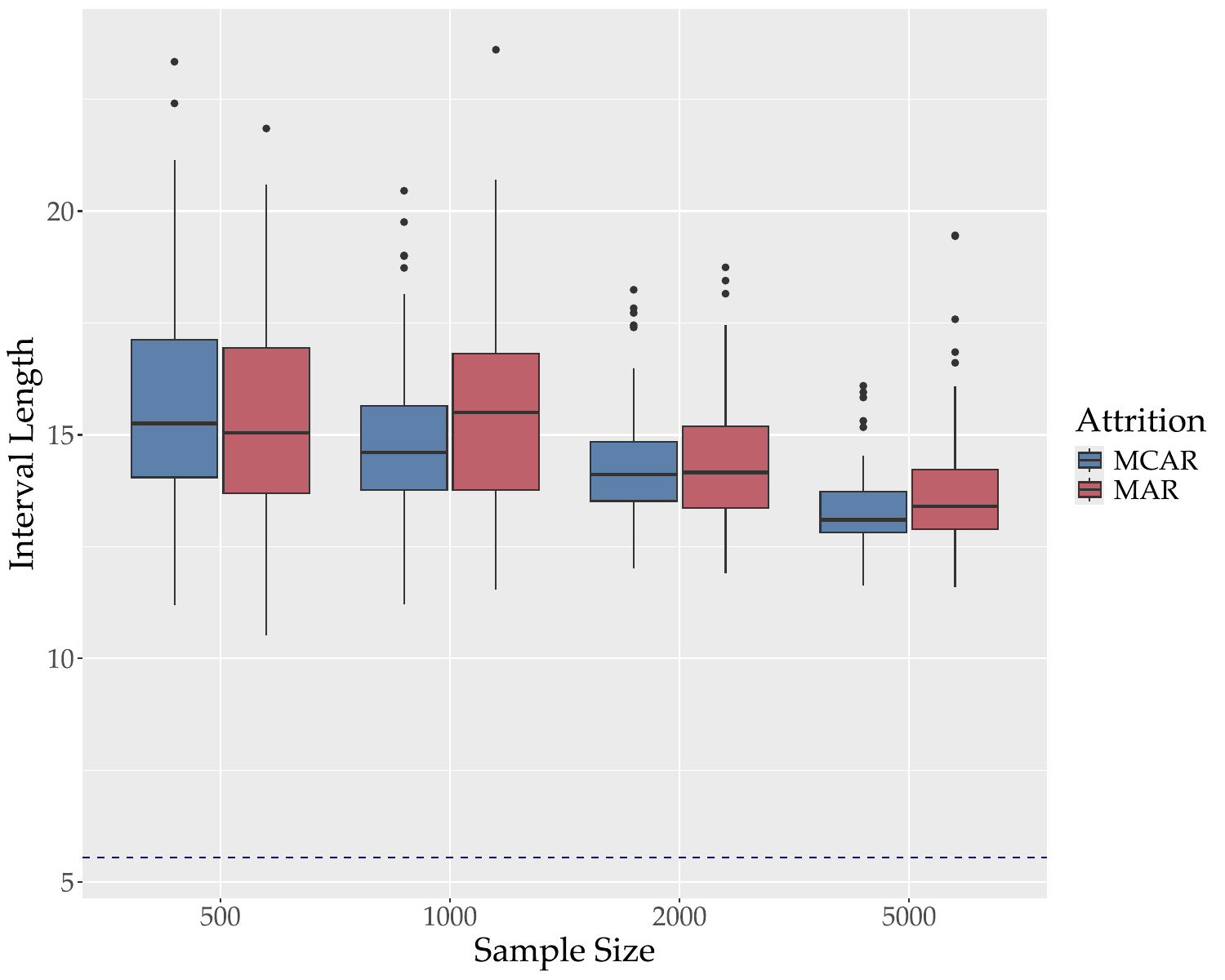}
            \end{minipage}
            \begin{tablenotes}
                \footnotesize
                \setlength\labelsep{0pt}
                \item \textit{Note}: This figure shows the simulation results of Exact method. The left panel shows the average empirical coverage and the right panel shows the average interval length of the prediction intervals for ITE in the target group with attrition. The red horizontal line corresponds to the target coverage of $95\%$ and the blue horizontal line corresponds to the average length of oracle intervals.
            \end{tablenotes}
        \end{threeparttable}
    \end{figure}

    \clearpage

    \subsection{Comparison with Different Multiple Imputation Procedures}

    Here, we consider two other multiple imputation procedures. The first one still follows the working flow of conformal inference. First, we impute the potential outcomes for the observed group, i.e., use $D = 1$ \& $R = 1$ to impute the counterfactuals of $D = 0$ \& $R = 1$, and vice versa. Second, we compute the ITE for the observed group. Third, we impute the ITE for the attrition group. The second one follows a reverse order. First, we impute the attrition group using observed group, i.e., use $D = 1$ \& $R = 1$ to impute $D = 1$ \& $R = 0$, and use $D = 0$ \& $R = 1$ to impute $D = 0$ \& $R = 0$. Then, we use the treatment group to impute the counterfactuals of the control group, and vice versa.

    Figure \ref{fig:MCcomp2covDGP1} and Figure \ref{fig:MCcomp2lenDGP1} show the comparison among three methods with the second procedure of multiple imputation.

    \begin{figure}[h]
        \centering
        \caption{Comparison of Empirical Coverage of Prediction Intervals for ITE with Attrition}
        \label{fig:MCcomp2covDGP1}
        \begin{threeparttable}
            \includegraphics[width=\textwidth]{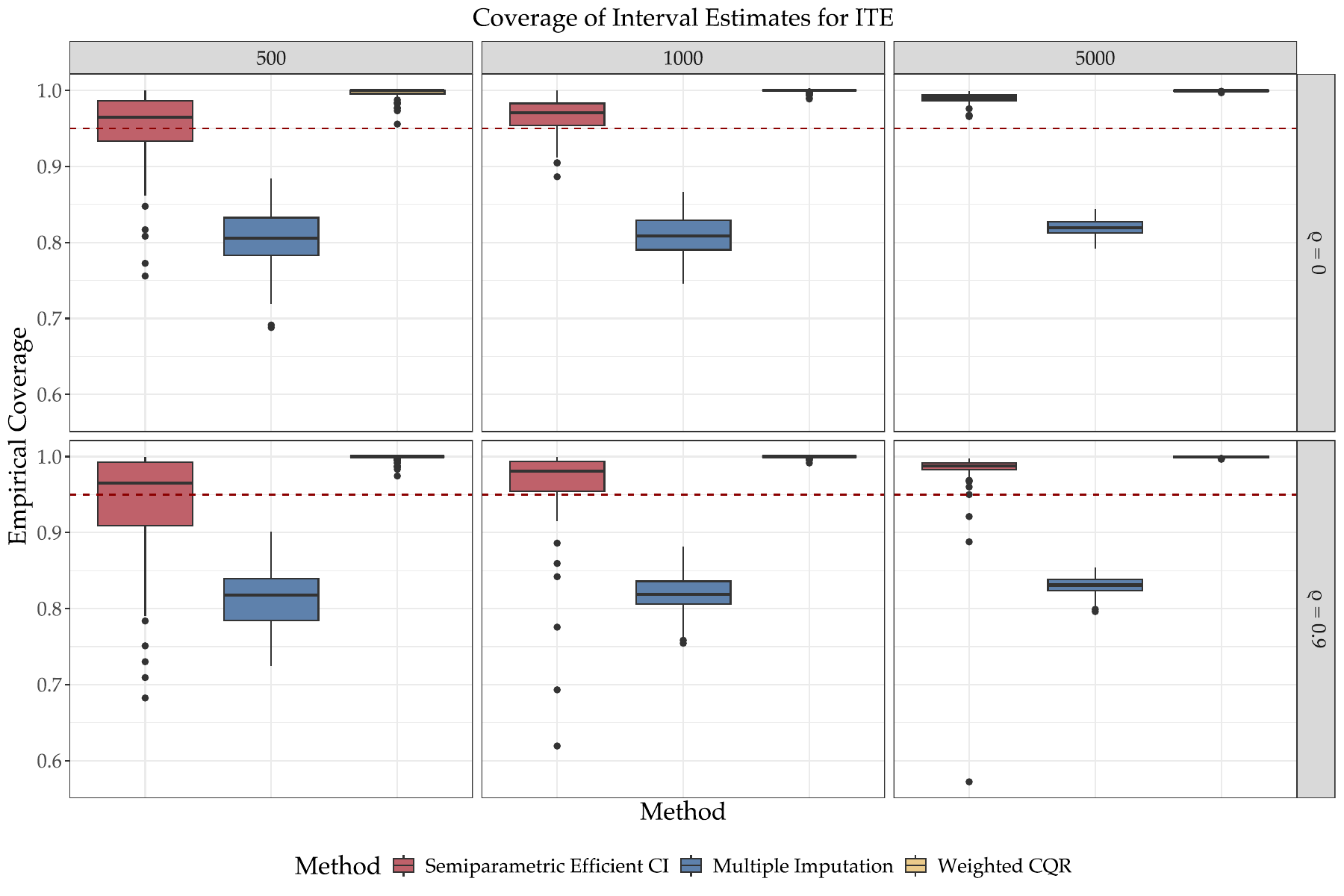}
            \begin{tablenotes}
                \footnotesize
                \setlength\labelsep{0pt}
                \item \textit{Note}: This figure shows the simulation results for the empirical coverage of prediction intervals constructed by conformal inference with semiparametric efficient estimator, multiple imputation with Amelia, and weighted CQR with unweighted nested approach for ITE of attrition group. The red horizontal line corresponds to the target coverage of $95\%$.
            \end{tablenotes}
        \end{threeparttable}
    \end{figure}

    \begin{figure}[ht]
        \centering
        \caption{Comparison of Average Length of Prediction Intervals for ITE with Attrition}
        \label{fig:MCcomp2lenDGP1}
        \begin{threeparttable}
            \includegraphics[width=\textwidth]{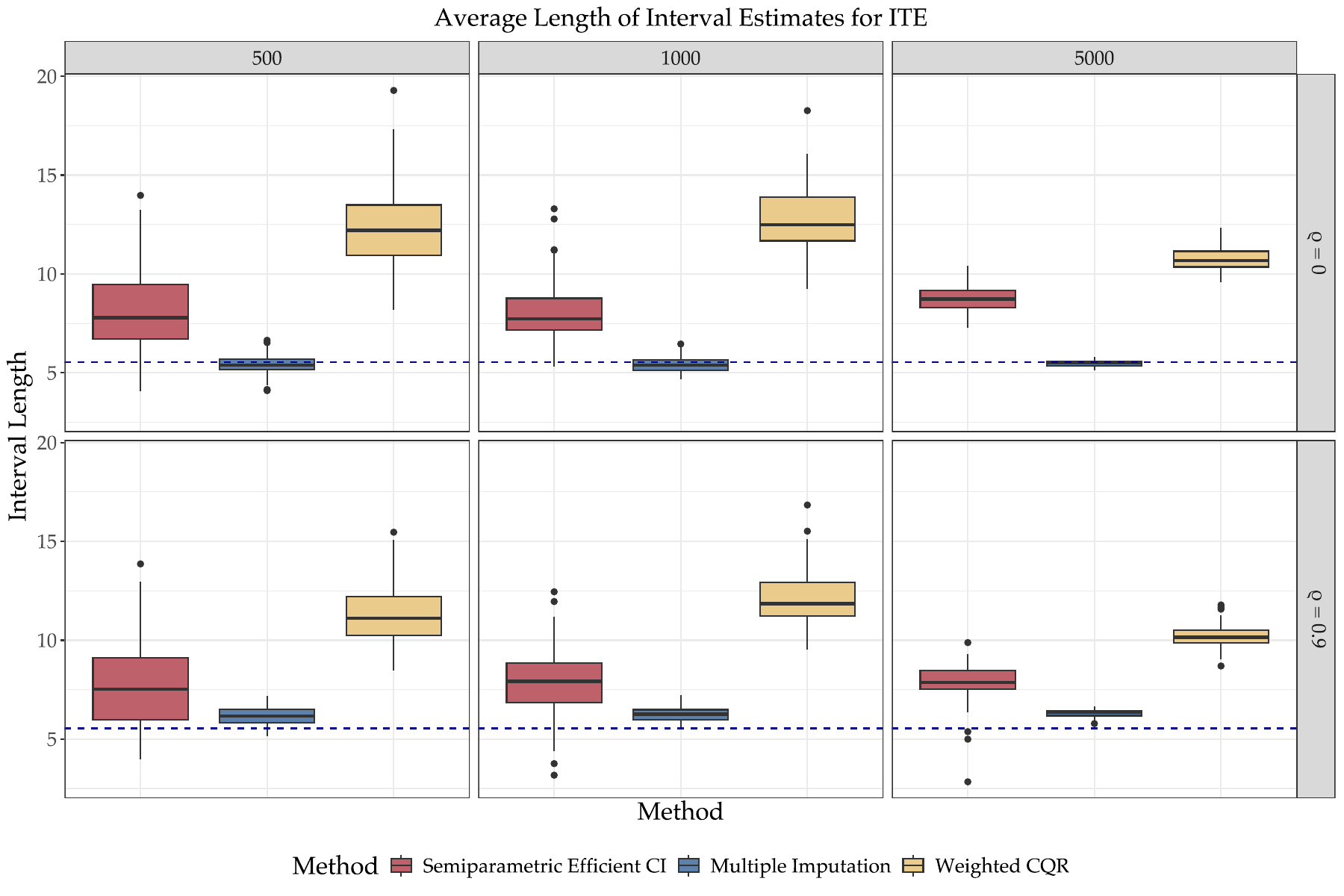}
            \begin{tablenotes}
                \footnotesize
                \setlength\labelsep{0pt}
                \item \textit{Note}: This figure shows the simulation results for the average length of prediction intervals constructed by conformal inference with semiparametric efficient estimator, multiple imputation with Amelia, and weighted CQR with unweighted nested approach for ITE of attrition group. The red horizontal line corresponds to the length of oracle intervals.
            \end{tablenotes}
        \end{threeparttable}
    \end{figure}

    \clearpage

    Figure \ref{fig:MCcomp3covDGP1} and Figure \ref{fig:MCcomp3lenDGP1} show the comparison among three methods with the third procedure of multiple imputation.

    \begin{figure}[h]
        \centering
        \caption{Comparison of Empirical Coverage of Prediction Intervals for ITE with Attrition}
        \label{fig:MCcomp3covDGP1}
        \begin{threeparttable}
            \includegraphics[width=\textwidth]{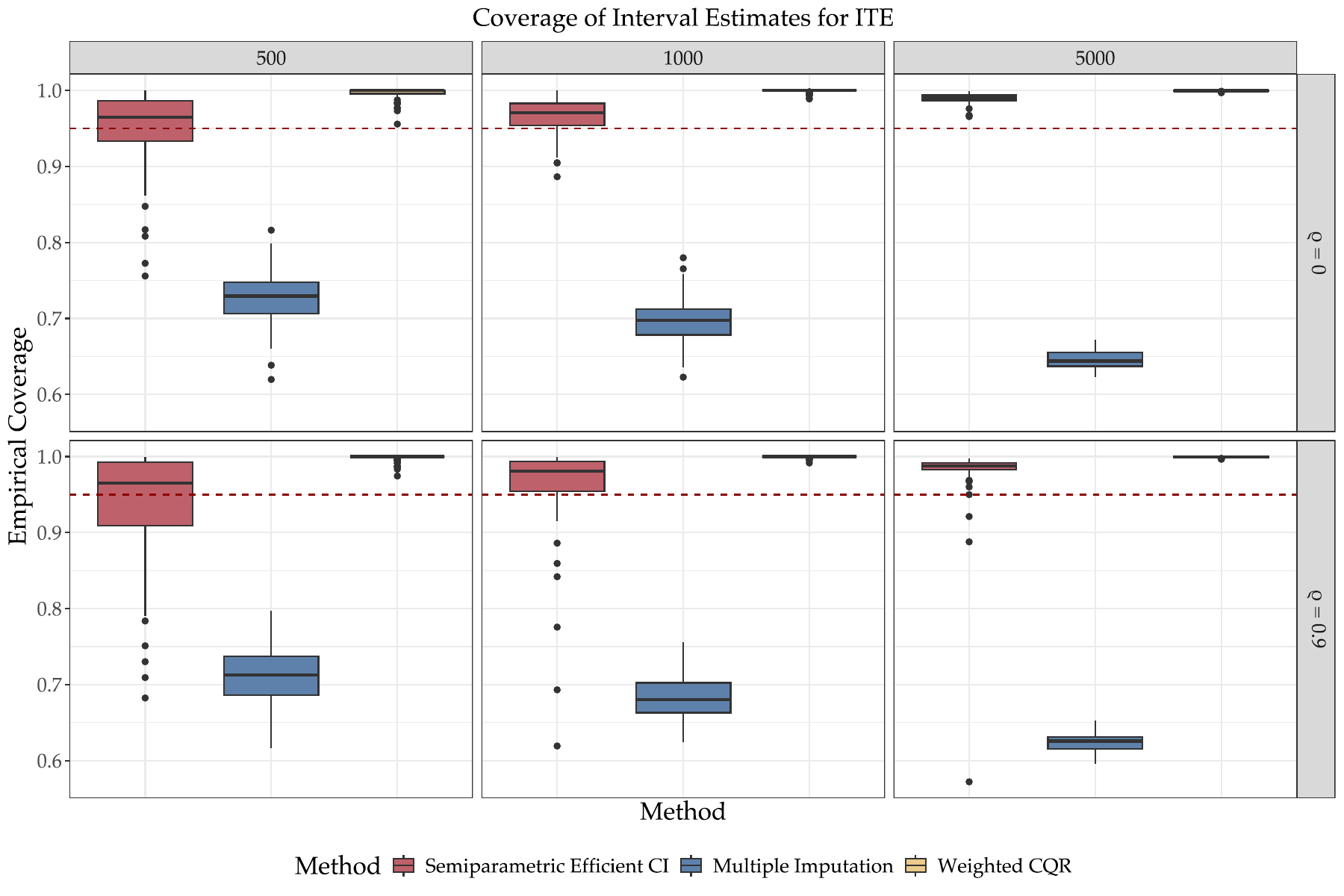}
            \begin{tablenotes}
                \footnotesize
                \setlength\labelsep{0pt}
                \item \textit{Note}: This figure shows the simulation results for the empirical coverage of prediction intervals constructed by conformal inference with semiparametric efficient estimator, multiple imputation with Amelia, and weighted CQR with unweighted nested approach for ITE of attrition group. The red horizontal line corresponds to the target coverage of $95\%$.
            \end{tablenotes}
        \end{threeparttable}
    \end{figure}

    \begin{figure}[ht]
        \centering
        \caption{Comparison of Average Length of Prediction Intervals for ITE with Attrition}
        \label{fig:MCcomp3lenDGP1}
        \begin{threeparttable}
            \includegraphics[width=\textwidth]{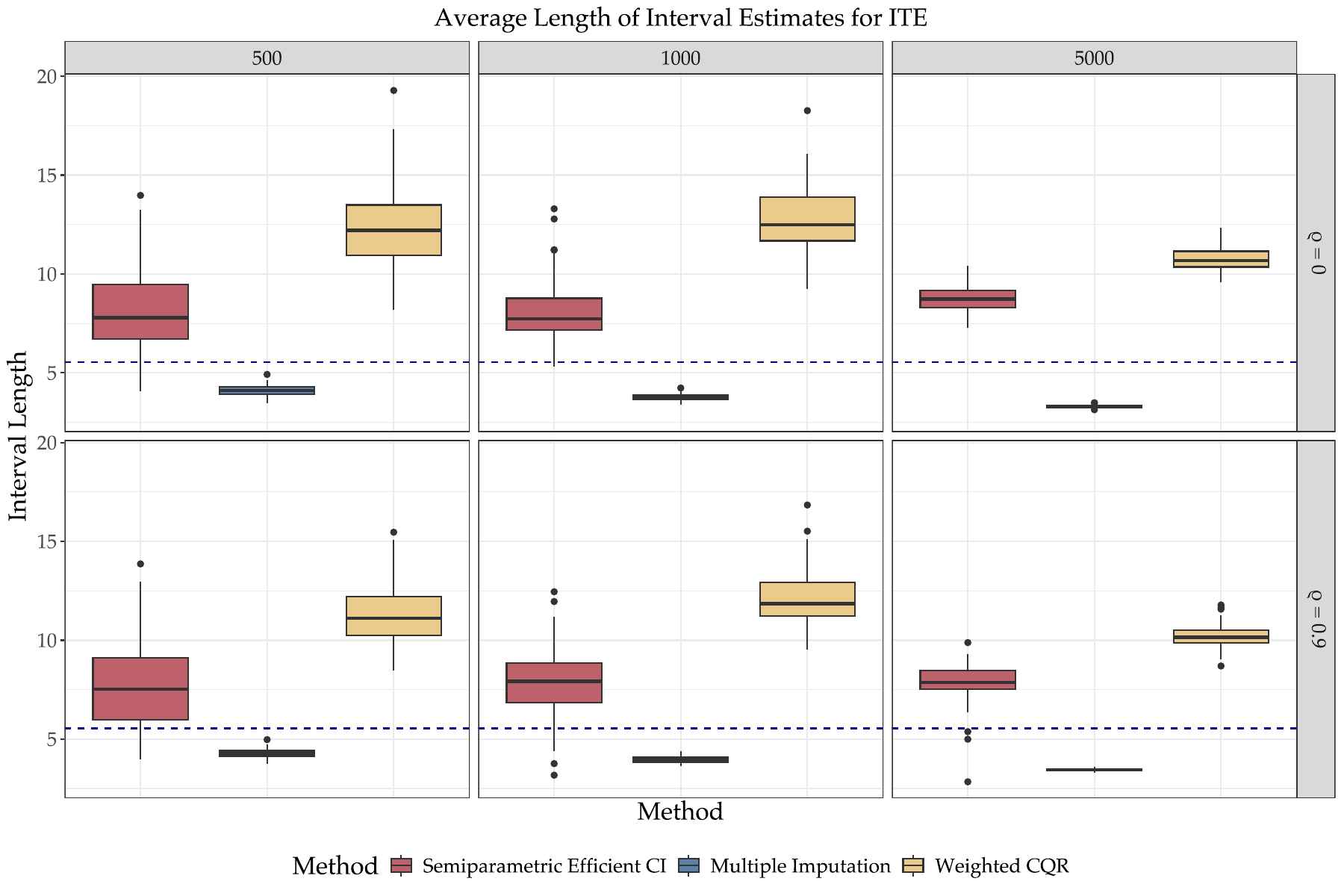}
            \begin{tablenotes}
                \footnotesize
                \setlength\labelsep{0pt}
                \item \textit{Note}: This figure shows the simulation results for the average length of prediction intervals constructed by conformal inference with semiparametric efficient estimator, multiple imputation with Amelia, and weighted CQR with unweighted nested approach for ITE of attrition group. The red horizontal line corresponds to the length of oracle intervals.
            \end{tablenotes}
        \end{threeparttable}
    \end{figure}

    \clearpage

    \subsection{Simulation Results with Different DGP}\label{sec:DGP2}

    In this section, I consider a different data generating process (DGP) to evaluate the performance of the proposed method. The covariate vector $X = \left( X_1, \dots, X_d \right)^\top$ is an equicorrelated multivarate Gaussian vector with mean zero and $\Var\left( X_i \right) = 1$ and $\Cov\left( X_i, X_j \right) = \rho$ for $i \neq j$. When $\rho = 0$, the covariates are independent. When $\rho > 0$, the covariates are positively correlated. The potential outcomes are generated as follows:
    \begin{align*}
        Y_1 &= X_1^2 + 0.2 X_2 + 1 / \log\left( 1 + \exp(X_3) \right) + 0.8 \exp(X_4) + \epsilon\\
        Y_0 &= 1 / \log\left( 1 + \exp\left( X_3 \right) \right) + \epsilon, \quad \epsilon \sim \mathcal{N}(0, 1).
    \end{align*}
    We still consider homescedastic noise. The propensity score of treatment $e_D(X)$ is generated as 
    \begin{align*}
        e_D(X) = \text{logit}^{-1} \left( -0.5X_1 - 0.3 X_2 + 0.2 X_3 \right).
    \end{align*}
    The propensity score of attrition $e_R(X, D)$ is generated as
    \begin{align*}
        e_R(X, D) = \text{logit}^{-1} \left( -1 + 0.3D + 0.5 X_1 - 0.4X_2 \right),
    \end{align*}
    which ensures the MAR assumption. Throughout the simulation, we set the dimension of the covariates $d = 10$.

    \clearpage

    \subsubsection{Simulation Results for Conformal Inference with Semiparametric Efficient Estimator}

    \begin{figure}[h]
        \centering
        \caption{MC Simulation Results of Conformal Inference for ITE with Attrition}
        \label{fig:MCallcov2}
        \begin{threeparttable}
            \includegraphics[width=\textwidth]{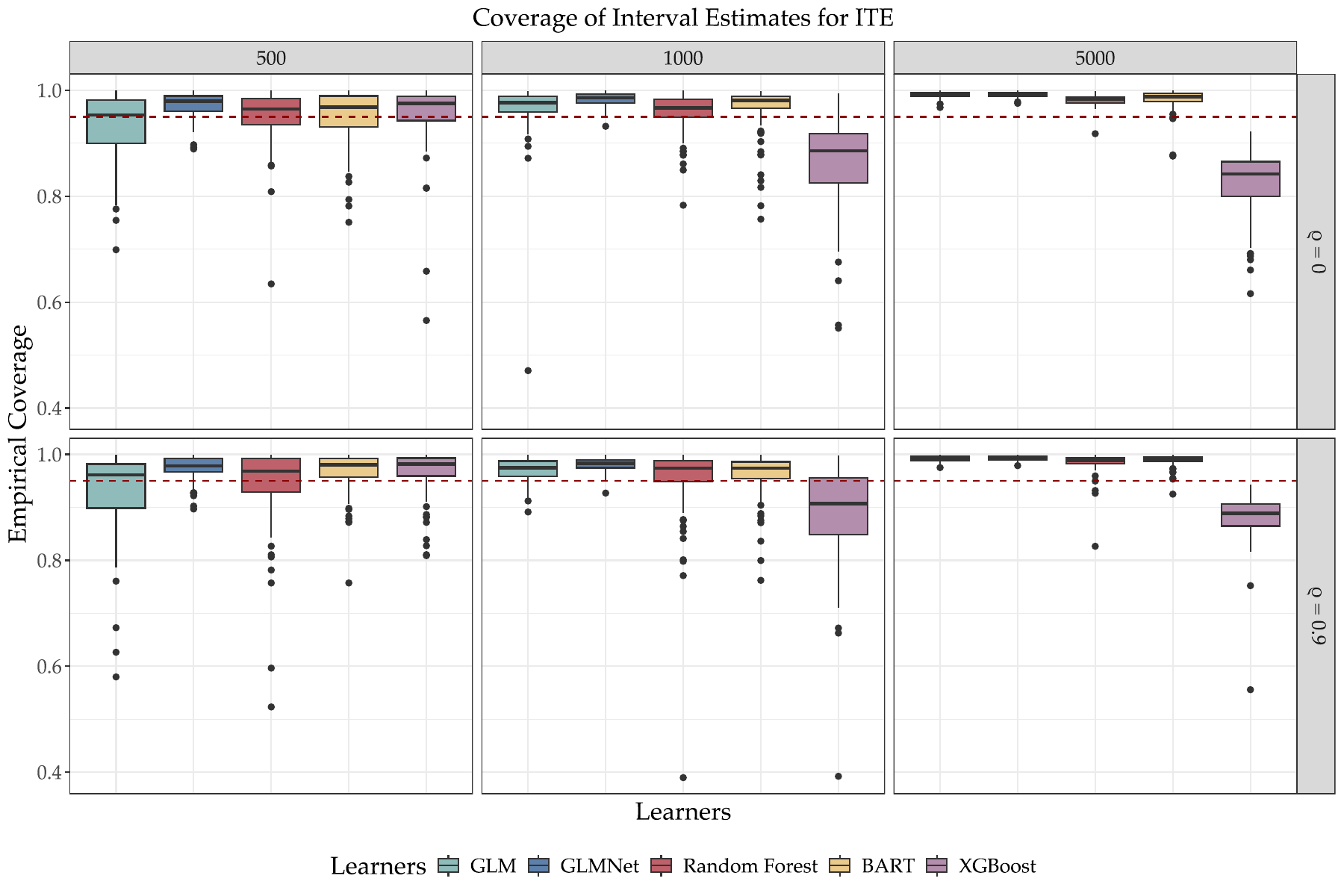}
            \begin{tablenotes}
                \footnotesize
                \setlength\labelsep{0pt}
                \item \textit{Note}: This figure shows the simulation results for the empirical coverage of prediction intervals constructed by semiparametric efficient estimator for ITE of attrition group following DGP2. The red horizontal line corresponds to the target coverage of $95\%$.
            \end{tablenotes}
        \end{threeparttable}
    \end{figure}

    \begin{figure}[h]
        \centering
        \caption{MC Simulation Results of Conformal Inference for ITE with Attrition}
        \label{fig:MCalllen2}
        \begin{threeparttable}
            \includegraphics[width=\textwidth]{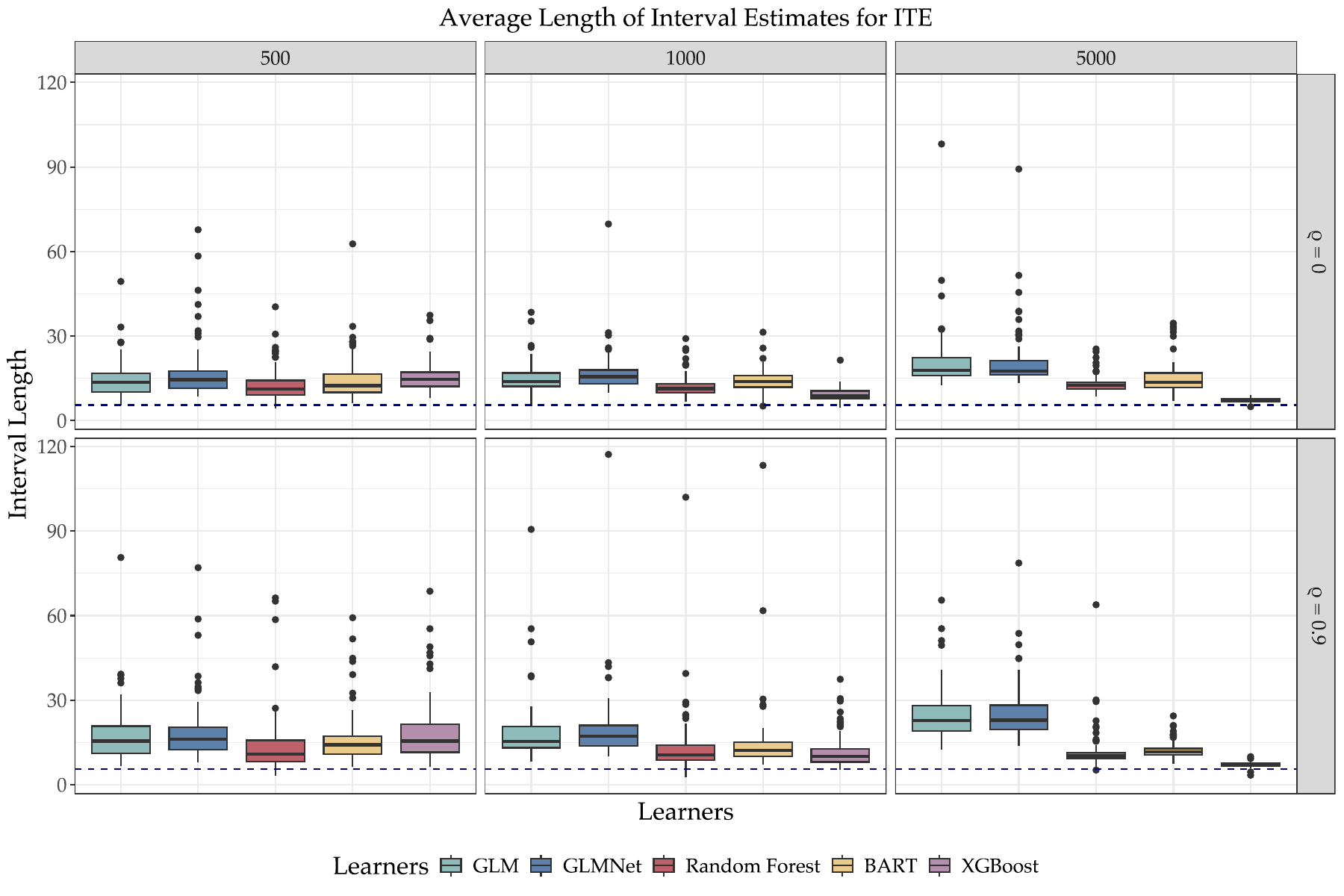}
            \begin{tablenotes}
                \footnotesize
                \setlength\labelsep{0pt}
                \item \textit{Note}: This figure shows the simulation results for the average length of prediction intervals constructed by semiparametric efficient estimator for ITE of attrition group following DGP2. The blue horizontal line corresponds to the length of oracle intervals.
            \end{tablenotes}
        \end{threeparttable}
    \end{figure}

    \clearpage

    \subsubsection{Simulation Results for Comparison with Other Methods}\label{sec:compDGP2}

    Similar to the second study of Section \ref{sec:MCsim}, we compare the empirical coverage and average length of the prediction intervals for ITE with attrition using conformal inference with semiparametric efficient estimator, multiple imputation with Amelia, and weighted CQR with unweighted and nested approach proposed by \citet{lei2021conformala}. 

    Figure \ref{fig:MCcomp2covDGP2} and Figure \ref{fig:MCcomp2lenDGP2} show the comparison among three methods with the first procedure of multiple imputation under the second DGP.

    \begin{figure}[h]
        \centering
        \caption{Comparison of Empirical Coverage of Prediction Intervals for ITE with Attrition}
        \label{fig:MCcomp1covDGP2}
        \begin{threeparttable}
            \includegraphics[width=\textwidth]{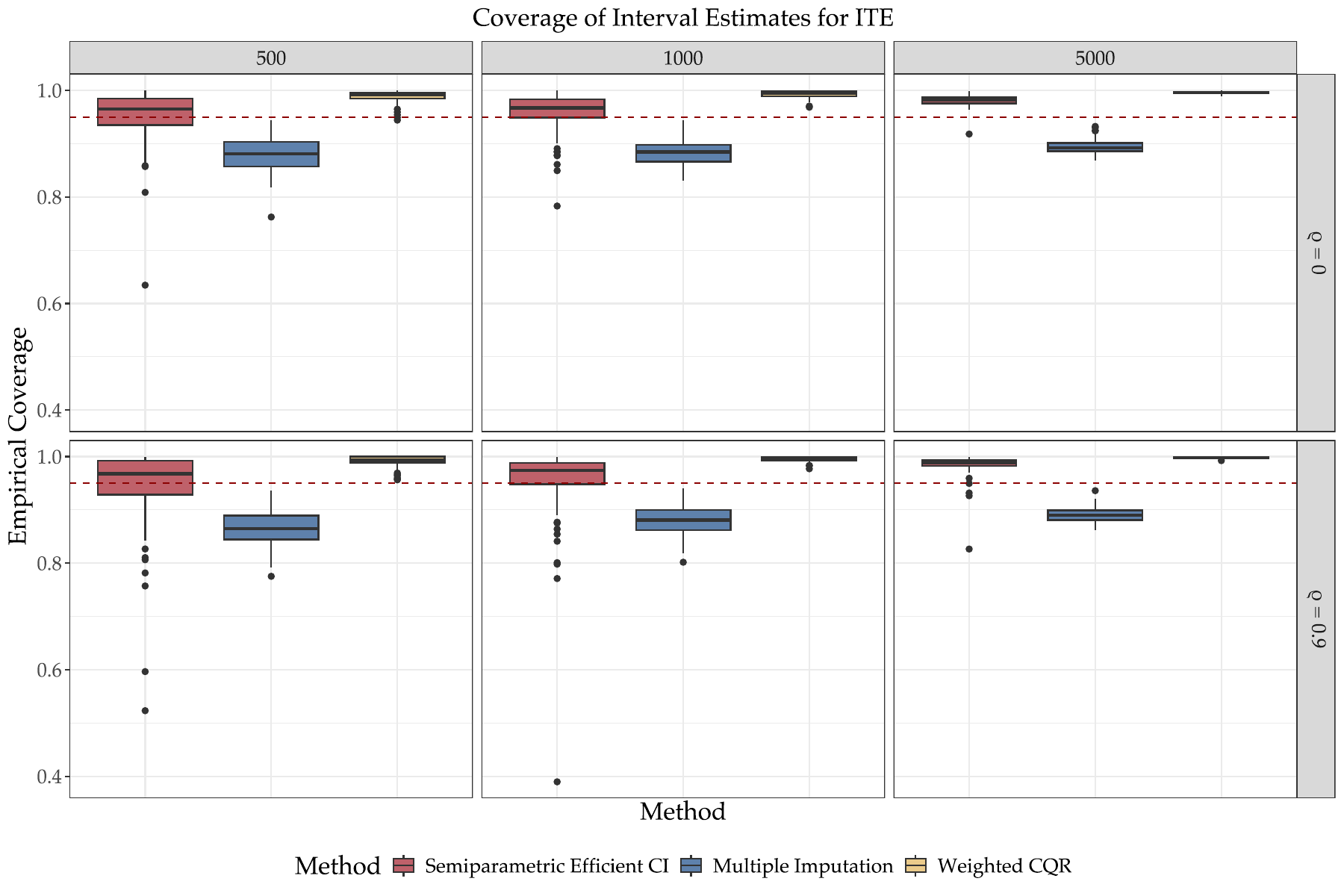}
            \begin{tablenotes}
                \footnotesize
                \setlength\labelsep{0pt}
                \item \textit{Note}: This figure shows the simulation results for the empirical coverage of prediction intervals constructed by conformal inference with semiparametric efficient estimator, multiple imputation with Amelia, and weighted CQR with unweighted nested approach for ITE of attrition group. The red horizontal line corresponds to the target coverage of $95\%$.
            \end{tablenotes}
        \end{threeparttable}
    \end{figure}

    \begin{figure}[h]
        \centering
        \caption{Comparison of Average Length of Prediction Intervals for ITE with Attrition}
        \label{fig:MCcomp1lenDGP2}
        \begin{threeparttable}
            \includegraphics[width=\textwidth]{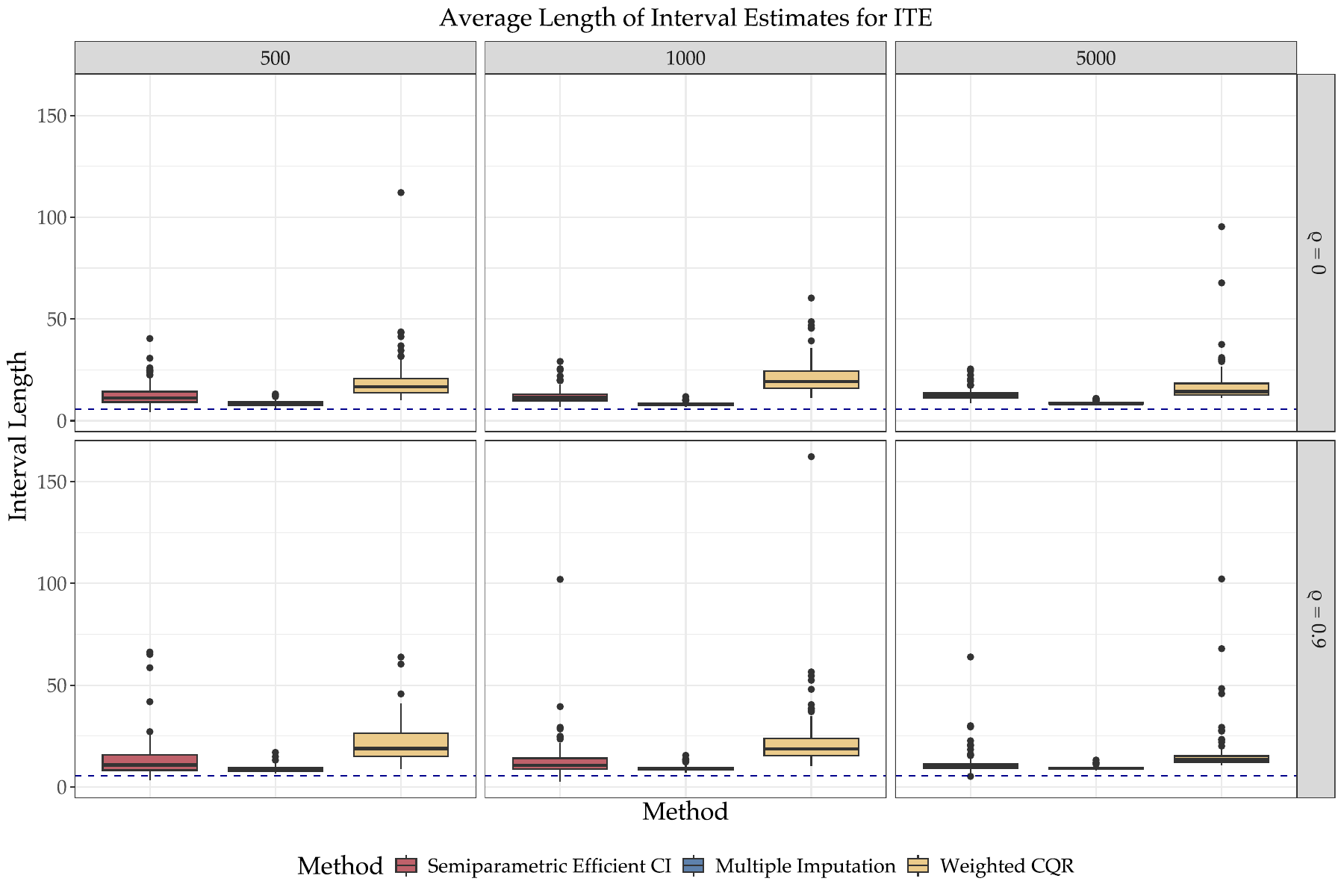}
            \begin{tablenotes}
                \footnotesize
                \setlength\labelsep{0pt}
                \item \textit{Note}: This figure shows the simulation results for the average length of prediction intervals constructed by conformal inference with semiparametric efficient estimator, multiple imputation with Amelia, and weighted CQR with unweighted nested approach for ITE of attrition group. The red horizontal line corresponds to the length of oracle intervals.
            \end{tablenotes}
        \end{threeparttable}
    \end{figure}

    \clearpage

    Figure \ref{fig:MCcomp2covDGP2} and Figure \ref{fig:MCcomp2lenDGP2} show the comparison among three methods with the second procedure of multiple imputation under the second DGP.

    \begin{figure}[h]
        \centering
        \caption{Comparison of Empirical Coverage of Prediction Intervals for ITE with Attrition}
        \label{fig:MCcomp2covDGP2}
        \begin{threeparttable}
            \includegraphics[width=\textwidth]{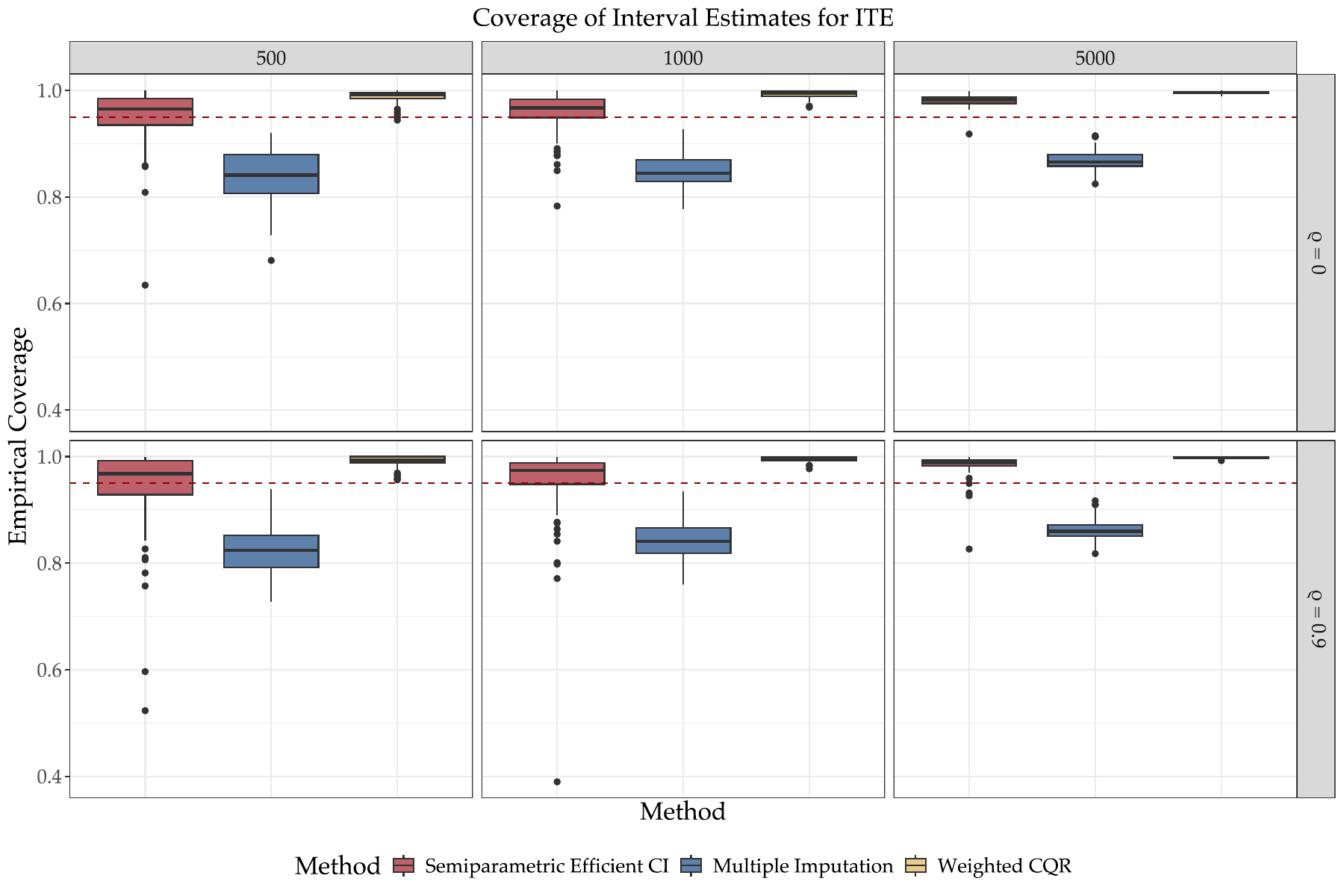}
            \begin{tablenotes}
                \footnotesize
                \setlength\labelsep{0pt}
                \item \textit{Note}: This figure shows the simulation results for the empirical coverage of prediction intervals constructed by conformal inference with semiparametric efficient estimator, multiple imputation with Amelia, and weighted CQR with unweighted nested approach for ITE of attrition group. The red horizontal line corresponds to the target coverage of $95\%$.
            \end{tablenotes}
        \end{threeparttable}
    \end{figure}

    \begin{figure}[h]
        \centering
        \caption{Comparison of Average Length of Prediction Intervals for ITE with Attrition}
        \label{fig:MCcomp2lenDGP2}
        \begin{threeparttable}
            \includegraphics[width=\textwidth]{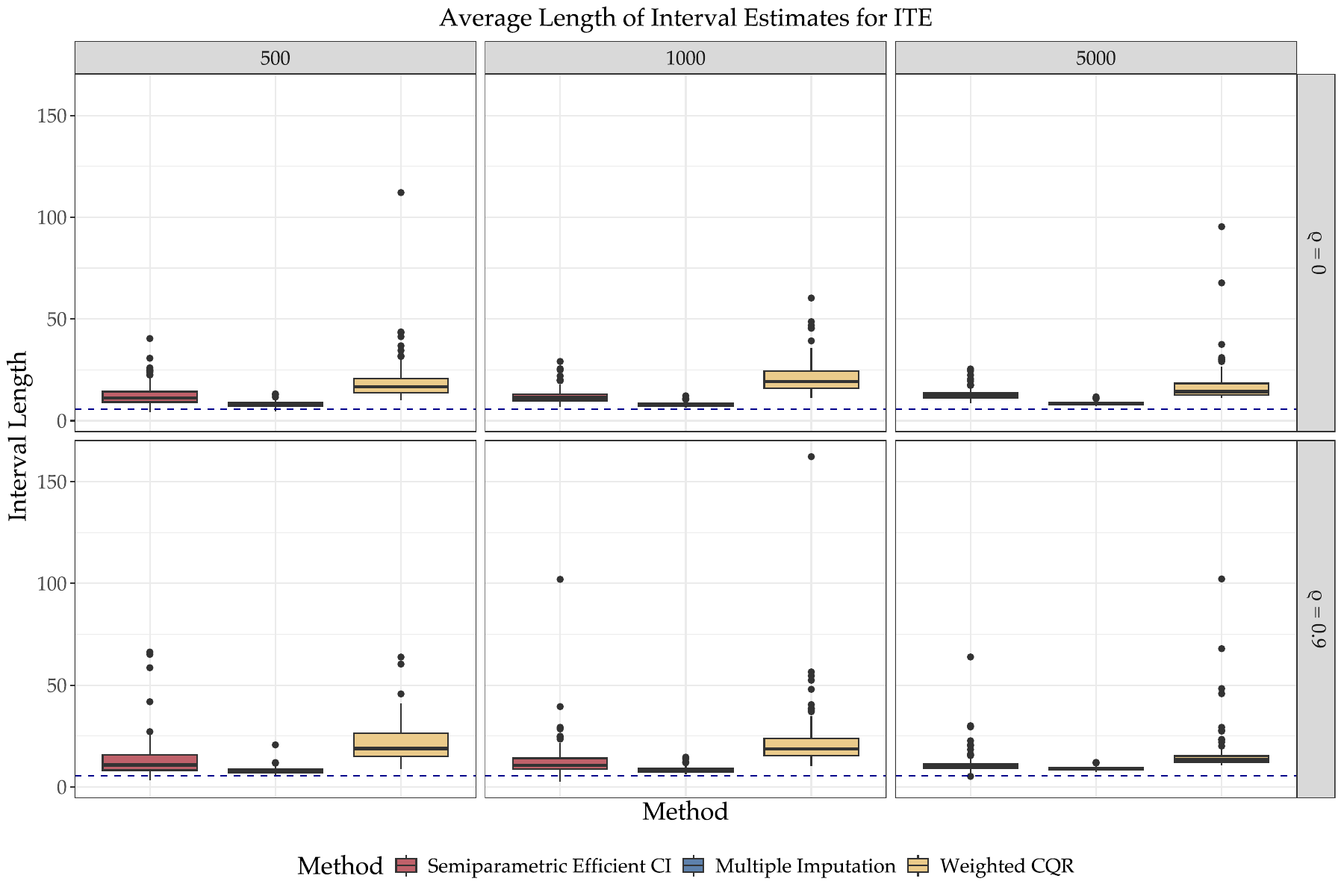}
            \begin{tablenotes}
                \footnotesize
                \setlength\labelsep{0pt}
                \item \textit{Note}: This figure shows the simulation results for the average length of prediction intervals constructed by conformal inference with semiparametric efficient estimator, multiple imputation with Amelia, and weighted CQR with unweighted nested approach for ITE of attrition group. The red horizontal line corresponds to the length of oracle intervals.
            \end{tablenotes}
        \end{threeparttable}
    \end{figure}

    \clearpage

    Figure \ref{fig:MCcomp3covDGP2} and Figure \ref{fig:MCcomp3lenDGP2} show the comparison among three methods with the third procedure of multiple imputation under the second DGP.

    \begin{figure}[h]
        \centering
        \caption{Comparison of Empirical Coverage of Prediction Intervals for ITE with Attrition}
        \label{fig:MCcomp3covDGP2}
        \begin{threeparttable}
            \includegraphics[width=\textwidth]{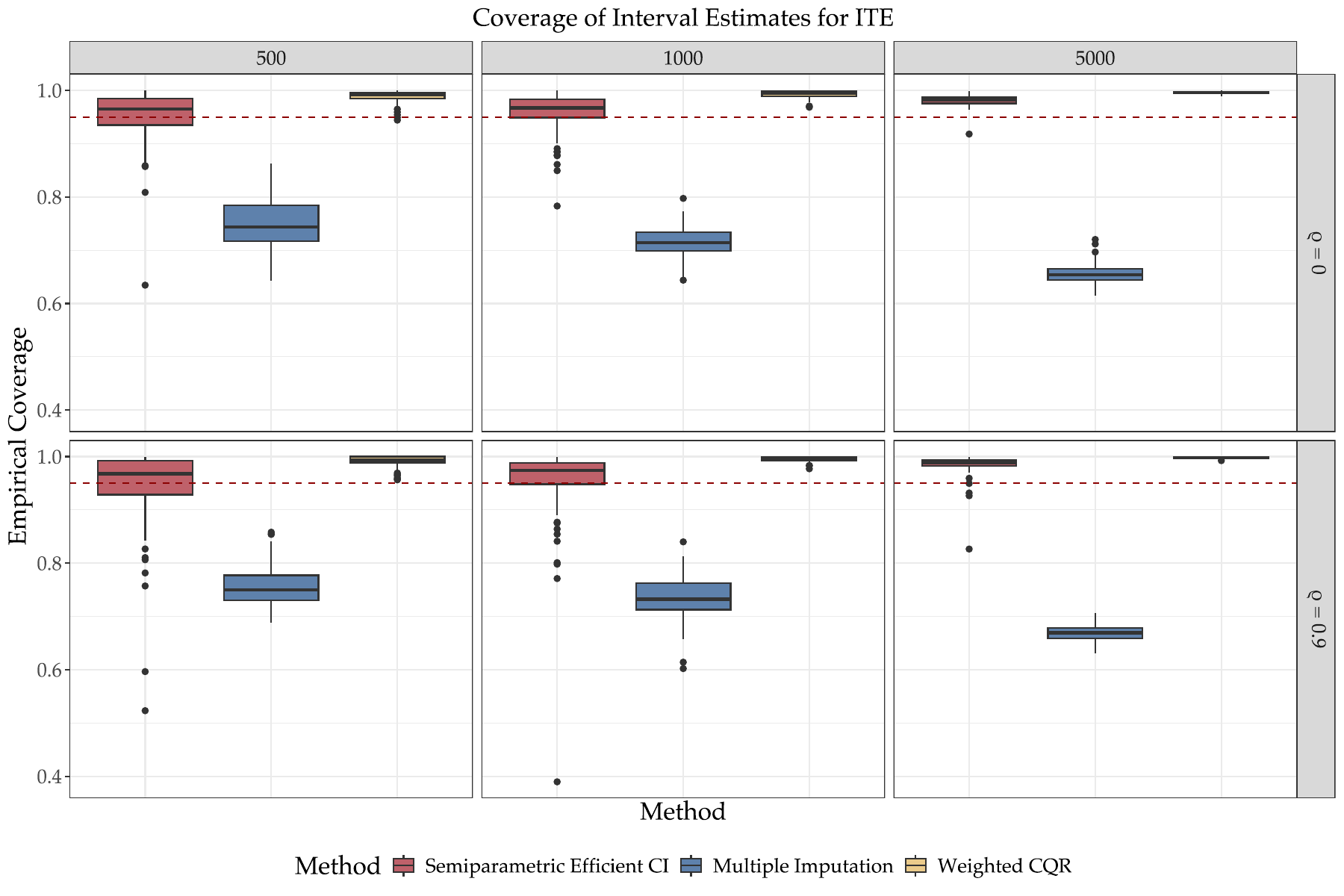}
            \begin{tablenotes}
                \footnotesize
                \setlength\labelsep{0pt}
                \item \textit{Note}: This figure shows the simulation results for the empirical coverage of prediction intervals constructed by conformal inference with semiparametric efficient estimator, multiple imputation with Amelia, and weighted CQR with unweighted nested approach for ITE of attrition group. The red horizontal line corresponds to the target coverage of $95\%$.
            \end{tablenotes}
        \end{threeparttable}
    \end{figure}

    \begin{figure}[h]
        \centering
        \caption{Comparison of Average Length of Prediction Intervals for ITE with Attrition}
        \label{fig:MCcomp3lenDGP2}
        \begin{threeparttable}
            \includegraphics[width=\textwidth]{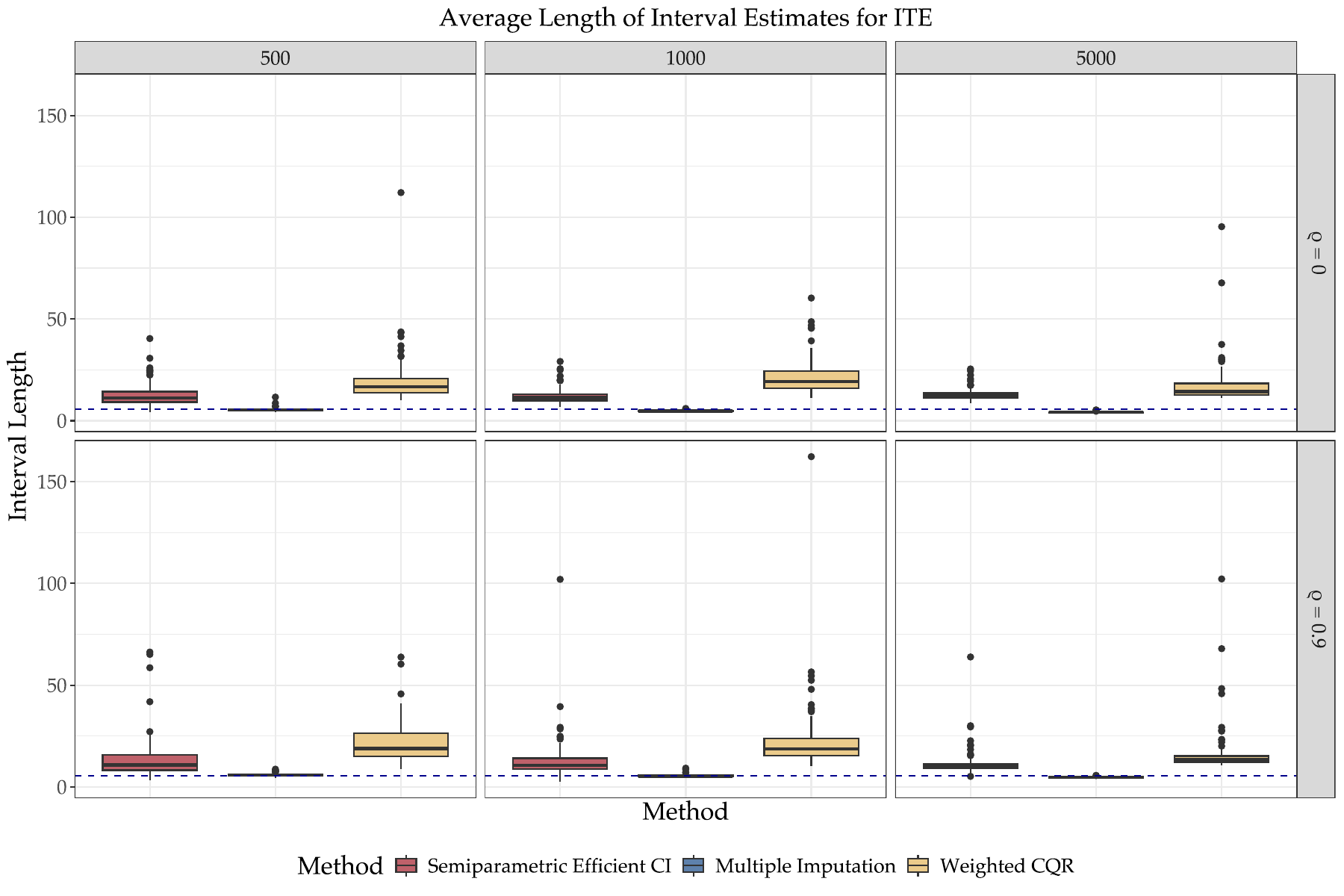}
            \begin{tablenotes}
                \footnotesize
                \setlength\labelsep{0pt}
                \item \textit{Note}: This figure shows the simulation results for the average length of prediction intervals constructed by conformal inference with semiparametric efficient estimator, multiple imputation with Amelia, and weighted CQR with unweighted nested approach for ITE of attrition group. The red horizontal line corresponds to the length of oracle intervals.
            \end{tablenotes}
        \end{threeparttable}
    \end{figure}
    
    \clearpage

    \section{Discussion on Replication Results}

    \subsection{Reanalysis of \citet{finkel2024can}}

    \begin{table}[h]
      \centering
      \caption{Balance Test: Observed vs.\ Attrition Groups}
      \label{tab:balance_finkel}
      \resizebox{\textwidth}{!}{%
      \begin{tabular}{lrrrrrrr}
        \toprule
        \multicolumn{1}{c}{ } & \multicolumn{3}{c}{Observed (R = 1)} & \multicolumn{3}{c}{Attrition (R = 0)} & \multicolumn{1}{c}{ } \\
        \cmidrule(l{3pt}r{3pt}){2-4} \cmidrule(l{3pt}r{3pt}){5-7}
        Variable & Mean\_Observed & SD\_Observed & N\_Observed & Mean\_Attrition & SD\_Attrition & N\_Attrition & p\_value\\
        \midrule
        D & 0.760 & 0.427 & 2203 & 0.761 & 0.427 & 1511 & 0.932\\
        age & 21.529 & 3.877 & 2203 & 21.635 & 4.131 & 1511 & 0.429\\
        ani\_2 & 0.068 & 0.252 & 2203 & 0.082 & 0.275 & 1511 & 0.115\\
        ani\_3 & 0.230 & 0.421 & 2203 & 0.197 & 0.398 & 1511 & 0.015\\
        ani\_4 & 0.326 & 0.469 & 2203 & 0.298 & 0.457 & 1511 & 0.064\\
        \addlinespace
        ani\_5 & 0.268 & 0.443 & 2203 & 0.289 & 0.454 & 1511 & 0.154\\
        autoc\_pref & 0.236 & 0.425 & 2203 & 0.255 & 0.436 & 1511 & 0.189\\
        dem\_pref & 0.585 & 0.493 & 2203 & 0.512 & 0.500 & 1511 & 0.000\\
        educ\_new\_2 & 0.108 & 0.310 & 2203 & 0.134 & 0.341 & 1511 & 0.017\\
        educ\_new\_3 & 0.245 & 0.430 & 2203 & 0.269 & 0.443 & 1511 & 0.101\\
        \addlinespace
        educ\_new\_4 & 0.393 & 0.488 & 2203 & 0.307 & 0.461 & 1511 & 0.000\\
        educ\_new\_5 & 0.213 & 0.410 & 2203 & 0.208 & 0.406 & 1511 & 0.684\\
        employed & 0.298 & 0.458 & 2203 & 0.347 & 0.476 & 1511 & 0.002\\
        f\_ani\_2 & 0.030 & 0.171 & 2203 & 0.050 & 0.219 & 1511 & 0.002\\
        f\_educ\_new\_2 & 0.025 & 0.155 & 2203 & 0.034 & 0.181 & 1511 & 0.105\\
        \addlinespace
        f\_polint\_new\_2 & 0.044 & 0.206 & 2203 & 0.063 & 0.243 & 1511 & 0.016\\
        female & 0.326 & 0.469 & 2203 & 0.361 & 0.480 & 1511 & 0.031\\
        polint\_new\_2 & 0.284 & 0.451 & 2203 & 0.249 & 0.432 & 1511 & 0.016\\
        polint\_new\_3 & 0.379 & 0.485 & 2203 & 0.349 & 0.477 & 1511 & 0.061\\
        polint\_new\_4 & 0.125 & 0.331 & 2203 & 0.156 & 0.363 & 1511 & 0.010\\
        \addlinespace
        regis\_dnw & 0.128 & 0.334 & 2203 & 0.162 & 0.369 & 1511 & 0.004\\
        regis\_no & 0.277 & 0.448 & 2203 & 0.273 & 0.445 & 1511 & 0.754\\
        \bottomrule
      \end{tabular}%
      }
    \end{table}

    \clearpage

    \begin{table}[h]
      \caption{HTE Regression: Political Interest Interactions}
      \centering
      \label{tab:hte}
      \begin{threeparttable}
      \begin{tabular}{lc}
      \toprule
        & Authoritarian Support\\
      \midrule
      D & $-0.024$\\
        & $(0.024)$\\[3pt]
      polint\_new\_2 & $-0.077^{***}$\\
        & $(0.028)$\\[3pt]
      polint\_new\_3 & $-0.091^{***}$\\
        & $(0.026)$\\[3pt]
      polint\_new\_4 & $-0.087^{**}$\\
        & $(0.035)$\\[3pt]
      D $\times$ polint\_new\_2 & $0.027$\\
        & $(0.032)$\\[3pt]
      D $\times$ polint\_new\_3 & $0.013$\\
        & $(0.030)$\\[3pt]
      D $\times$ polint\_new\_4 & $0.088^{**}$\\
        & $(0.040)$\\[3pt]
      Intercept & $0.324^{***}$\\
        & $(0.021)$\\
      \midrule
      Observations & 2,203\\
      \bottomrule
      \end{tabular}
      \begin{tablenotes}
        \footnotesize
        \setlength\labelsep{0pt}
        \item \textit{Note}: Standard errors in parentheses. $^{*}p<0.1$; $^{**}p<0.05$; $^{***}p<0.01$.
      \end{tablenotes}
      \end{threeparttable}
    \end{table}

\end{document}